\documentclass{amsart}

\usepackage[reset,a4paper,vmargin=3truecm,hmargin=2truecm]{geometry}
\usepackage{amsmath,amsthm,amssymb,amscd,amsthm,epic}
\usepackage{bbm,graphicx,bm}
\usepackage{mathtools}
\usepackage{leftidx}

\usepackage{subcaption}
\usepackage{tikz}

\theoremstyle{plain}
\newtheorem{thm}{Theorem}[section]
\newtheorem{lem}[thm]{Lemma}
\newtheorem{prop}[thm]{Proposition}
\newtheorem{cor}[thm]{Corollary}

\theoremstyle{definition}
\newtheorem{dfn}{Definition}[section]
\newtheorem{ex}[thm]{Example}

\theoremstyle{remark}
\newtheorem{rem}{Remark}[section]
\DeclareMathOperator{\tr}{tr}

\DeclareMathOperator{\spn}{span}

\makeatletter

\newcommand{\HRabi}{H_{\text{R}}}
\newcommand{\ZRabi}{Z_{\text{R}}}
\newcommand{\KRabi}{K_{\text{R}}}

\newcommand{\N}{\mathbb{N}} 
\newcommand{\Z}{\mathbb{Z}} 
\newcommand{\R}{\mathbb{R}} 
\newcommand{\C}{\mathbb{C}} 

\newcommand{\e}{\varepsilon}

\def\smallunderbrace#1{\mathop{\vtop{\m@th\ialign{##\crcr
   $\hfil\displaystyle{#1}\hfil$\crcr
   \noalign{\kern3\p@\nointerlineskip}%
   \tiny\upbracefill\crcr\noalign{\kern3\p@}}}}\limits}

\newcommand{\QEDhere}{\pushQED{\qed}\qedhere\popQED}

\newcommand{\mat}[1]{\begin{bmatrix}#1\end{bmatrix}}

\DeclareMathOperator{\pref}{Pre}

\newcommand{\matrixU}[1] {
  \mat{ u^{#1}  & 0 \\  0 & u^{-#1}}
}

\newcommand{\matrixId}{
  \mat{ \phantom{-}1  & \phantom{-}0 \\  \phantom{-}0 & \phantom{-}1}
}

\newcommand{\matrixZZ}{
  \mat{ \phantom{-}1 & -1 \\ -1 & \phantom{-}1 }
}

\newcommand{\matrixZO}{
  \mat{ -1 & -1 \\ \phantom{-}1 & \phantom{-}1 }
}
\newcommand{\matrixOZ}{
  \mat{ -1 & \phantom{-}1 \\ -1 & \phantom{-}1 }
}
\newcommand{\matrixOO}{
  \mat{ \phantom{-}1 & \phantom{-}1 \\ \phantom{-}1 & \phantom{-}1 }
}

\newcommand{\setA}[3]{ \mathcal{A}^{(#1)}_{#2 #3} }
\newcommand{\setC}[3]{ \mathcal{C}^{(#1)}_{#2 #3} }

\newcommand{\bv}{\mathbf{v}}
\newcommand{\bs}{\mathbf{s}}
\newcommand{\br}{\mathbf{r}}
\newcommand{\bZ}{\mathbf{0}}
\newcommand{\bO}{\mathbf{1}}

\newcommand{\bI}{\mathbf{I}}
\newcommand{\bJ}{\mathbf{J}}
\newcommand{\bA}{\mathbf{A}}
\newcommand{\bM}{\mathbf{M}}
\newcommand{\bC}{\mathbf{C}}
\newcommand{\bU}{\mathbf{U}}

\newcommand{\btZ}[1]{\Z_{2}^{#1}}
\newcommand{\btZdual}[1]{\widehat{\Z_{2}^{#1}}}

\begin{document}


\title{Heat kernel for the quantum Rabi model}
\author{Cid Reyes-Bustos}
\author{Masato Wakayama}

\date{\today}




\begin{abstract}  
  
The quantum Rabi model (QRM) is widely recognized as a particularly important model in quantum optics and beyond. It is considered to be the simplest and most fundamental system describing quantum light-matter interaction. The objective of the paper is to give an analytical formula of the heat kernel of the Hamiltonian explicitly by infinite series of iterated integrals. 
The derivation of the formula is based on the direct evaluation of the Trotter-Kato product formula without the use of Feynman-Kac path integrals.
More precisely, the infinite sum in the expression of the heat kernel arises from the reduction of the Trotter-Kato product formula into sums over the orbits of the action of the infinite symmetric group $\mathfrak{S}_\infty$ on  the group \(\Z_2^{\infty} \), and the iterated integrals are then considered as the orbital integral for each orbit. Here, the groups \(\Z_2^{\infty} \) and $\mathfrak{S}_\infty$ are the inductive limit of the families  $\{\Z_2^n\}_{n\geq0}$ and $\{\mathfrak{S}_n\}_{n\geq0}$, respectively. In order to complete the reduction, an extensive study of harmonic (Fourier) analysis on the inductive family of abelian groups $\Z_2^n\, (n \geq0)$ together with a graph theoretical investigation is crucial.
To the best knowledge of the authors, this is the first explicit computation for obtaining a closed formula of the heat kernel for a non-trivial realistic interacting quantum system. The heat kernel of this model is further given by a two-by-two matrix valued function and is expressed as a direct sum of two respective heat kernels representing the parity ($\Z_2$-symmetry) decomposition of the Hamiltonian by parity. 
\;

\end{abstract}


\maketitle

\tableofcontents

\section{Introduction} \label{sec:intro}

The quantum Rabi model (QRM) is widely recognized as the simplest and most fundamental model describing quantum light-matter interactions, that is, the interaction between a two-level system and a bosonic field mode (see e.g. \cite{bcbs2016} for a recent collection 
of introductory, survey and original articles after Isidor Rabi's seminal papers \cite{Rabi1936, Rabi1937} on the semi-classical (Rabi) model in 1936 and 1937 and the full quantization \cite{JC1963} established in 1963 by Jaynes and Cummings).
Quantum interaction models, of which the QRM may be considered as a distinguished representative, have been considered not only by theorists but also by experimentalists (e.g. \cite{YS2018}) as indispensable models 
  for advancing research on quantum computing and its implementation (see e.g \cite{BMSSRWU2017, HR2008}).
In spite of recent progress on theoretical/mathematical and numerical studies (e.g. see \cite{bcbs2016, Le2016, BdMZ2019, KRW2017} and their references), knowledge in several fundamental areas is still limited. For instance, explicit formulas for time evolution of the corresponding systems remain largely unknown (though certain partial results have been discussed in, e.g. \cite{LCDFGZ1987} for the Spin-Boson model and \cite{AYH1970, CR1995} for the Kondo effect), and similarly, qualitative information about the spectrum in several coupling regimes (stipulated by the physical parameters defining the model) and eigenvalues distribution (e.g. the unsolved conjecture for the QRM in \cite{B2011PRL}) remains sparse and mysterious (see, e.g. \cite{RS2017PRA}).

The Hamiltonian \(\HRabi\) of the QRM is precisely given by
\[
  \HRabi := \omega a^{\dagger}a + \Delta \sigma_z + g (a + a^{\dagger}) \sigma_x .
\]
Here, \(a^{\dagger}\) and \(a\) are the creation and annihilation operators of the single bosonic mode (\([a,a^{\dagger}]=1 \)), $\sigma_x, \sigma_z$ are the Pauli matrices (sometimes written as \(\sigma_1\) and \(\sigma_3\), but since there is no risk of confusion with the variable \(x\) to appear below in the heat kernel, we use the usual notations), $2\Delta$ is the energy difference between the two levels, $g$ denotes the coupling strength between the two-level system and the bosonic mode with frequency $\omega$ (subsequently, we set $\omega=1$ without loss of generality). The integrability of the QRM was established in \cite{B2011PRL} using the well-known $\Z_2$-symmetry of the Hamiltonian \(\HRabi\), usually called parity. The QRM actually appears ubiquitously in various quantum systems including cavity and circuit quantum electrodynamics, quantum dots and artificial atoms with potential applications in quantum information technologies (see \cite{HR2008, bcbs2016}). 
For instance, recent experimental results \cite{YS2016,YS2018}, aimed at measuring light shifts of superconducting flux qubits deep-strongly coupled to LC oscillators, agree with theoretical predictions based on the QRM and its asymmetric version. It is actually shown in \cite{YS2018} that in the deep-strong regime (i.e. in the case $(\Delta, \omega) \ll g$), the energy eigenstates are well described by entangled states. 

The purpose of the present paper is to obtain  closed explicit expressions for the heat kernel and the partition function of the QRM. Let us briefly recall the definitions. The heat kernel $\KRabi(x,y,t;g,\Delta)$ of $\HRabi$ is the integral kernel corresponding to the operator
\(e^{-t \HRabi} \) (one-parameter semigroup), that is, $\KRabi(x,y,t;g,\Delta)$ satisfies
\[
e^{-t \HRabi}\phi(x)= \int_{-\infty}^\infty  \KRabi(x,y,t;g,\Delta) \phi(y)dy
\]
for a compactly supported smooth function $\phi: \R \to \C^2$. Precisely, $K_{\text{Rabi}}(x,y,t;g,\Delta)$ is the (two-by-two matrix valued) function satisfying $\frac{\partial}{\partial t} \KRabi(x,y,t;g,\Delta)= - \HRabi \, \KRabi(x,y,t;g,\Delta)$ for all $t>0$ and
$\lim_{t \to 0}\KRabi(x,y,t;g,\Delta) = \delta_x(y) \bf{I}_2$ for $x,y \in \R$. 

In statistical physics, the partition function of a system is of fundamental importance as it describes the statistical properties of the system in thermodynamic equilibrium as a function of temperature and other parameters, such as the volume enclosing a gas. The partition function $\ZRabi(\beta;g,\Delta)$ of the quantum system QRM is given by the trace of the Boltzmann factor $e^{-\beta E(\mu)}$, $E(\mu)$ being the energy of the state $\mu$:
\[
  \ZRabi(\beta;g,\Delta):= \text{Tr}[e^{-t \HRabi}]= \sum_{\mu\in \Omega} e^{-\beta E(\mu)},
\]
where $\Omega$ denotes the set of all possible (eigen-)states of $\HRabi$. 

In general, the computation of the heat kernel of an operator is considered to be a difficult problem and it often involves the use of mathematically transcendental techniques such the Feynman path integrals or the well-defined (rigorous) Feynman-Kac formulas (for the relation between the path integral and the (Lie-) Trotter-Kato product formula, see e.g. \cite{Calin2011} and  \cite{I2003}). For instance, the heat kernel was expressed by the Feynman-Kac path integral in \cite{HH2012}. In contrast, the method presented in this paper for the explicit computation of the heat kernel of the QRM is based on  detailed calculations using the Trotter-Kato product formula (or the exponential product formula for the semigroup) directly  for a pair of (in general) self-adjoint unbounded operators \cite{Kato1978, AI2008}. In particular, we do not employ path-integrals or probabilistic methods, instead, we deal with the limit appearing in the Trotter-Kato product formula by using harmonic (Fourier) analysis on the inductive family of abelian finite groups \(\{\Z_2^{n}\}_{n\geq0}\).
Each elements of  \(\Z_2^{n}\), for \(n \in \N \), may be interpreted as a path between two points alternating
between two ``states''.  For illustrative purposes, in figure \ref{fig:paths} we show two paths
corresponding to two elements of \(\Z_2^{n}\) where the two states are denoted by ``+'' and ``-''. In this way,
in our computation, in place of all paths in the path integral, we employ all paths in the inductive limit $\Z_2^\infty:=\varinjlim \Z_2^n$ as \(n \to \infty\) (with a natural point measure). Our infinite, yet countable, number of paths in $\Z_2^\infty$
may then be considered as sort of representatives of all paths in the Feynman-Kac path integral. 
Symbolically, we may think there is some equivalence $\sim$ such that 
\[
  \{ \text{ paths } \}/\sim \quad \longleftrightarrow \, \Z_2^{\infty}
\]
is a bijection. A deeper understanding of this expected relation is an important open question.
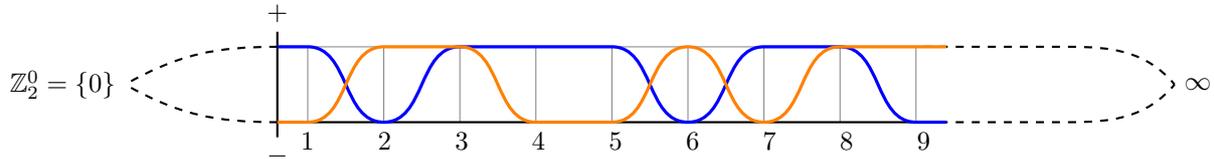
\begin{figure}[!ht]
  \centering
  \begin{tikzpicture}[domain=0:4]

    \draw[step=1,very thin,color=gray] (1,0.0) grid (9,1);

    \draw[thick] (1,0) -- (9.1,0) 
    node[pos=1,below] {$9$}
    node[pos=0,below]{1} node[pos=0.125,below]{2} node[pos=0.25,below]{3}
    node[pos=0.375,below]{4} node[pos=0.5,below]{5} node[pos=0.625,below]{6}
    node[pos=0.750,below]{7} node[pos=0.875,below]{8}; 

    \draw[thick,dashed] (0.6,1) to[out=180,in=30] (-1.4,0.5) ;
    \draw[thick,dashed] (0.6,0) to[out=180,in=-30] (-1.4,0.5) node[left] {$\Z_2^0=\{0\}$};

    \draw[thick,dashed] (9.1,1) -- (11,1) ;
    \draw[thick,dashed] (9.1,0) -- (11,0) ;
    \draw[thick,dashed] (11,1) to[out=0,in=130] (12.4,0.5) ;
    \draw[thick,dashed] (11,0) to[out=0,in=-130] (12.4,0.5) node[right] {$\infty$} ;

    \draw[thick] (0.6,-0.2) node[below] {$-$} -- (0.6,1.2) node[above] {$+$};

    \draw[color=blue,very thick] (0.6,1) -- (1,1) to[out=0,in=180] (2,0) to[out=0,in=180] (3,1) to[out=0,in=180] (4,1) to[out=0,in=180] (5,1) to[out=0,in=180] (6,0) to[out=0,in=180] (7,1) to[out=0,in=180] (8,1) to[out=0,in=180] (9,0) --  (9.4,0); 

    \draw[color=orange,very thick] (0.6,0) to[out=0,in=180] (1,0) to[out=0,in=180] (2,1) to[out=0,in=180] (3,1) to[out=0,in=180] (4,0) to[out=0,in=180] (5,0) to[out=0,in=180] (6,1) to[out=0,in=180] (7,0) to[out=0,in=180] (8,1) to[out=0,in=180] (9,1) to[out=0,in=180] (9.4,1); 

    
  \end{tikzpicture}
  \caption{Two paths in $\Z_2^{n}$ for \(n=9\).}
  \label{fig:paths}
\end{figure}

The resulting formulas for the heat kernel and partition function are given as infinite series (actually, a uniformly convergent power series in the parameter \(\Delta\)) where each of the summands consist of a $k$-iterated integral $(k=0,1,2,\ldots)$. 
Concretely, for the heat kernel of QRM, the main result of this paper, given in Theorem \ref{thm:heat_kernel}, is the analytical formula of the form
\begin{align*}
  \KRabi(x,y,t ;g,\Delta) =  K_0(x,y,t;g) \sum_{\lambda=0}^{\infty} (t\Delta)^{\lambda} e^{-2g^2 (\coth(\tfrac{t}2))^{(-1)^\lambda}}  \Phi_\lambda(x,y,t;g).
\end{align*}
Here, \(K_0(x,y,t;g)\) is explicitly given through Mehler's kernel, that is, the heat propagator of the quantum harmonic oscillator (see Theorem \ref{thm:heat_kernel} for the definition) and \(\Phi_\lambda(x,y,t;g)\) is a \(2 \times 2\) matrix-valued
function given by
\begin{align*}
  \Phi_\lambda(x,y,t;g) := \idotsint\limits_{0\leq \mu_1 \leq \cdots \leq \mu_\lambda \leq 1}  e^{4g^2 \frac{\cosh(t(1-\mu_\lambda))}{\sinh(t)}(\frac{1+(-1)^\lambda}{2}) + \xi_{\lambda}(\bm{\mu_{\lambda}},t)}\begin{bmatrix}
    (-1)^{\lambda} \cosh  &  (-1)^{\lambda+1} \sinh  \\
    -\sinh &  \cosh
  \end{bmatrix}
             \left( \theta_{\lambda}(x,y,\bm{\mu_{\lambda}},t) \right) d \bm{\mu_{\lambda}}. 
\end{align*}
We refer to \S~\ref{sec:limit} for the definition of the functions \(\xi_\lambda(\bm{ \mu_{\lambda}},t)\) and \(\theta_{\lambda}(x,y,\bm{\mu_{\lambda}},t)\) which are explicitly represented by finite sums of hyperbolic functions. An analogous expression is given in Corollary \ref{cor:Partition_function} for the partition function of QRM obtained from the heat kernel formula by taking trace.
In the language of $\Z_2^\infty$-paths described above, the infinite sum (series) in the expression of the heat kernel is considered to be taken over the orbits $\mathcal{O}_\lambda$ of the infinite symmetric group $\mathfrak{S}_\infty$ defined by the inductive family of symmetric groups $\{\mathfrak{S}_n\}_{n\geq0}$ and the summand given by iterated integral (over the $\lambda$th simplex) can be regarded as the orbital integral for each orbit $\mathcal{O}_\lambda$ in $\Z_2^\infty$ (cf. Lemma \ref{lem:sumint}). A short but detailed discussion is given in \S \ref{sec:remark-computations} and we leave to \cite{RW2020z} the more detailed discussion along the induced representation viewpoint of Tsilevich and Vershik \cite{TV2007}.
  
It is relevant to mention that similar expressions have been found in the study of evolution operators (propagators) or related quantities in other physical models. For the Hamiltonian associated to the Kondo problem (the model for a quantum impurity coupled to a large reservoir of non-interacting electrons), a special matrix coefficient (a correlation function) of the heat kernel was obtained in \cite{AYH1970} (see Equation (1) therein) as a power series with coefficients given by iterated integrals (see also \cite{Kondo2012} for an extensive discussion). This study is continued in \cite{CR1995} to obtain numerical results on the behavior of the correlation function for long times and its asymptotics for the Kondo problem. An attempt to obtain the thermodynamic properties of a Kondo impurity using the Monte Carlo method was first considered in \cite{SS1971}, but in contrast to the method developed in \cite{CR1995}, it involves simulation in the grand ensemble, a technically more difficult problem. 
For the Spin-Boson model, which may be regarded as a generalization of the QRM, the formal expression for \(P(t) := \langle \sigma_z(t) \rangle \) where \(\sigma_z(t)\) is the Heisenberg representation with respect to the full Hamiltonian at time \(t\), a quantity related to the qualitative behaviour of the system, is given in (4.17-19) of
\cite{LCDFGZ1987} as a power series in a parameter with iterated integral coefficients. In both cases, the formulas
are obtained by the evaluation of Feynman-Kac path-integrals. Some recent developments, mainly on the experimental implementation side can be found in the review \cite{LeHur2018}.

On the topic of the QRM, in \cite{ZZ1988} the authors gave an approximated (or incomplete) formula for the propagator
using path-integral techniques. For the Spin-Boson model, and the QRM as a special case, a 
Feynman-Kac formula for the heat kernel (semi-group generator) was obtained in \cite{HH2012,HHL2012} via a Poisson point process and a Euclidean field.
However, for the study of longtime behavior of the system, the use of numerical computations or approximations is inevitable (see for instance \cite{LF2011, CPMP}). In addition, we note that the heat kernel of the QRM (Theorem \ref{thm:heat_kernel}) might be computed from the Feynman-Kac path integral expression of Theorem 3.2 in \cite{HH2012} using the method computing $P(t)$ developed in \cite{LCDFGZ1987} (or, back to the definition of the path integral and make it to be a Riemann sum as in this paper) but it seems to require a similar volume of computation with a necessary mathematically rigorous discussion.

At this point, it is significant to mention a mathematical model that shares some features with the QRM and that actually, in a specific sense, may be considered to be a generalization of the QRM.
The non-commutative harmonic oscillator (NCHO) is the model defined by a deformation of the tensor product of the quantum harmonic oscillator (cf. \cite{HT1992}) and the two dimensional trivial representation of Lie algebra $\mathfrak{sl}_2$ (see \cite{PW2001, P2010S}). The QRM is obtained from the NCHO (with generalized parameters in the definition of the Hamiltonian) through a confluence process (i.e. two regular singularities merge to an irregular singularity) in the Heun ODE picture of respective models \cite{W2015IMRN}. Considering the NCHO as a sort of {\it covering} of the QRM (in the sense of the confluence procedure), one might expect the explicit derivation of the heat kernel to be simpler than in the case of the QRM, however, the heat kernel for the NCHO is yet to be obtained.
It is worth remarking here that the spectrum of the NCHO is known to posses many rich arithmetic structure (e.g. modular forms, particular congruences, automorphic integrals) via the special values (i.e. moments of partition function) of its spectral zeta function \cite{IW2005a,IW2005b,KW2004,KW2007,KW2012RIMS,KW2019,L2018,LOS2016PAMS}. We expect the spectrum of QRM (via spectral zeta function \cite{RW2020z, Sugi2016}) to have similar rich number theoretic structure as well. 

The paper is organized as follows. A large part of the paper (\S \ref{sec:spvl} through \S \ref{sec:limit}) is devoted to obtaining an explicit expression of the heat kernel of the QRM. In \S \ref{sec:spvl}, we make preliminary calculations based on the Trotter-Kato product formula by employing newly defined two-by-two matrix-valued creation and annihilation operators $b$ and $b^\dag$ depending on the parameter $g$.
The resulting expression is a limit (\(N \to \infty \)) that,  at a glance, resembles a Riemann sum where
each of the summands contains a sum of exponential terms over subsets \(\setC{\ell}{i}{j}\) of
\(\Z_2^N\) for \(2 \le \ell \leq N\) and \(i,j \in \{0,1\} \) (see Definition \ref{dfn:subsets} in \S \ref{sec:matrixG}).
However, the presence of (infinitely many) changes of signature with non-trivial coefficients in the exponential terms is an obstacle for the direct evaluation of this expression.

In order to overcome the aforementioned difficulties in the evaluation of the limit, in \S \ref{sec:four-transf} we make use of harmonic analysis on the finite groups \(\Z_2^{k} \, (k \in \Z_{\geq 0})\). 
Concretely, by noticing a natural bijection between \(\setC{\ell}{i}{j}\) and \(\Z_2^{\ell-2}\), we transform the sum over \(\setC{\ell}{i}{j}\)
into an equivalent sum over the dual group of \( \Z_2^{\ell-2}\) using the Fourier transform. To simplify the resulting expressions, in addition to the standard theory of harmonic analysis, we also develop certain graph theoretical and combinatorial techniques in \S \ref{sec:graph-theoretical}. The transformed sum is then seen to consist of a radial function part (for \(\rho \in \Z_2^{\ell-2}\)) that is controlled by fixing \(|\rho |= \lambda \), while the sum of remaining part is evaluated as a multiple integral in \S \ref{sec:transf-summ-into}. We remark that transforming the computation into the dual stage (i.e. the Fourier image) is not only indispensable in order to evaluate the sum in practice, but it also reveals certain structural information that appears in the final expression of the heat kernel  (cf. Lemma \ref{lem:sumexp}). Precisely, one of the advantages of employing harmonic analysis on $\Z_2^k$ is that  the product containing sign changes from the  non-commutative part appearing in the (finite approximate) expression of the heat kernel \eqref{eq:integralbasic} can be transformed into a single trigonometric function of  a alternating sum of ``$q$-numbers" $[i]_q=\frac{1-q^i}{1-q}$ (similar to the evaluation of a Gauss sum in number theory). The remaining advantage is that it gives a systematic treatment of the object by separating the radial part and the others in the dual group of $\Z_2^k$ (see also Remark \ref{q-Fourier}). It should be noticed that  the graph theoretical discussion in \S \ref{sec:graph-theoretical} enables us to obtain such separation. 

The limit is evaluated in  \S \ref{sec:limit}, thus completing the computation of the heat kernel. The final result (Theorem \ref{thm:heat_kernel}) shows that the heat kernel of the QRM is expressed as an infinite sum of the terms given
by $k$-iterated integrals $(k=0,1,2,\ldots)$. It is important to notice that the point-wise convergence of the iterated integral kernels to the heat kernel follows from the fact that the corresponding sequence of the Trotter-Kato approximation operators converges not only in the strong operator topology but in the operator norm topology as well (see \cite{AI2008} and references therein). The explicit form of the partition function (Corollary \ref{cor:Partition_function}) then follows directly from that of the heat kernel. As mentioned above, the Hamiltonian for the QRM possesses a parity (\(\Z_2\)-symmetry). From this fact, we see that the heat kernel for the model, given by two-by-two matrix of operators, is expressed as the direct sum of two heat kernels which represent the parity decomposition (Theorem \ref{Split_Kernel}). We then derive the explicit formula for the partition function for each parity (Corollary \ref{cor:parityPart}).

To the best knowledge of the authors, this is the first non-trivial example of explicit computation of the heat kernel of an interacting quantum system.  Although similar formulas have been derived before (e.g. for the spin dynamics of the Kondo model or of the general Spin-Boson model), this was done only for the reduced density matrix of the two-level system (the spin) and not for the full system. Actually, as already mentioned, for the heat kernel of the Kondo Hamiltonian, only a special matrix element was obtained in \cite{AYH1970, CR1995} as a sum of iterated integrals  (see also Appendices B, C, D in \cite{LCDFGZ1987}).

The method developed in this paper using the Trotter-Kato product formula may be generalized to other similar quantum systems. For instance, it may be extended in a straightforward way for the study of the heat kernel of generalizations of the QRM like the asymmetric quantum Rabi model (AQRM) or the Dicke model (see e.g \cite{B2013MfI}) and the two photons quantum Rabi model (see, e.g. \cite{RMR2020}). We also expect it may be used for more complicated models like the Spin-Boson model. We note, in particular, that the method does not use any $\Z_2$-symmetry of the QRM Hamiltonian (see Remark \ref{non-commutativity}). We believe that this method may play the role of a compass in the study of other Hamiltonians and their heat kernels. Actually, the time evolution operator $e^{-it\HRabi}\, (t\in \R)$, obtained by analytic continuation of $e^{-t\HRabi}$ with respect to \(t\), is of great importance in physics.
Here, according to Stone's theorem, the operator $e^{-it\HRabi}\, (t\in \R)$ is unitary and describes a strongly continuous one-parameter group of unitary transformations in the Hilbert space since $\HRabi$ is self-adjoint. For the details, we direct the reader to \cite{RW2020z}.
Moreover, there has been active studies and significant efforts on the estimation of the Trotter error in view of potential applications, e.g. to quantum simulation. See the recent study \cite{C2021a} (and \cite{C2021} for an updated
version) and the references therein. Although the QRM Hamiltonian is simpler than the general quantum systems considered in the aforementioned studies, we expect that the method for computation developed in the present paper may 
contribute to the estimation of the Trotter error. We will consider this point in another occasion.

Furthermore, we remark that although the QRM is in the scientific spotlight in theoretical and experimental physics, a full-fledged classification and consequent theoretical prediction of coupling regimes remains unclear (see, e.g. \cite{RS2017PRA}).  In particular, the current coupling regime classification was initially based on the agreement between the QRM and its rotating-wave approximation, for instance, in \cite{Lv2018}  using a trapped-ion system the authors demonstrate the breakdown of the rotating-wave approximation of the QRM as the parameters move from one coupling regime to another. Further approximations that work over the different regions under certain conditions have been proposed (see \cite{LFB2020} for the approximation of the ground state for AQRM in all coupling regimes). A coupling regime classification based on the spectrum has been proposed in \cite{RS2017PRA}, which in addition to the parameters of the system takes into account the energy the system can access. We expect that precise numerical computations based by the explicit analytical formulas of the heat kernel and partition function may also contribute to investigations in this direction.

\section{Preliminary calculations based on the Trotter-Kato product formula}
\label{sec:spvl}

The Hamiltonian of the quantum Rabi model (QRM) is given by
\[
  \HRabi = a^{\dagger}a + \Delta \sigma_z + g (a + a^{\dagger}) \sigma_x ,
\]
where \(\sigma_x,\sigma_z \) are the
Pauli matrices
\[
  \sigma_x =
  \begin{pmatrix}
    0 & 1  \\ 1 & 0
  \end{pmatrix}, \qquad
  \sigma_z =
  \begin{pmatrix}
    1 & 0 \\ 0 & -1
  \end{pmatrix},
\]
and \(a^{\dagger},a\) are the creation and annihilation operators of the quantum harmonic oscillator satisfying the commutation relation \([a,a^{\dag}]=1 \). In this paper we tacitly assume \(\hbar = \omega= 1 \) without loss of generality.

Since the third term $g(a+a^\dag)\sigma_x$ of the Hamiltonian $\HRabi$ does not commute with $a^\dag a$ and $\Delta\sigma_z$, in order to reduce the number of non-commuting relations, we make a change of operator in the following way. Set \( b = b(g) := a +  g \sigma_x \) (and whence \( b^\dag= a^\dag +  g \sigma_x \)). Thus we easily obtain the following lemma.

\begin{lem} 
The Hamiltonian \(\HRabi \) is expressed as
\begin{align*}
  \HRabi 
  & = (a^{\dagger} + g \sigma_x)(a + g \sigma_x) + \Delta \sigma_z - g^2  \\
  & = b^{\dagger} b - g^2 + \Delta \sigma_z.  \QEDhere
\end{align*} 
\end{lem}
Notice that while the operators \( b \), \(b^{\dagger}\) satisfy the commutation relation \([b, b^{\dagger} ]= \bI_2 \),
the operator \(b\) does not commute with \(\Delta \sigma_z\). In this sense, we regard
\[
  b^{\dagger} b - g^2
\]
as a (two dimensional) non-commutative version of the quantum harmonic oscillator.
From the commutation relation \([b, b^{\dagger} ] = \bI_2\), it is clear that the
operator \(b^{\dagger}b - g^2 \) is self-adjoint and bounded below. The operator \(\Delta \sigma_z \) is also self-adjoint and bounded for trivial reasons.
We remark that is it a well-known fact (see for example \cite{Sugi2016}) that \(\HRabi\) is a self-adjoint bounded below operator. Therefore, the operators \(b^{\dagger}b - g^2\) and \(\Delta \sigma_z \) satisfy
the conditions of the Trotter-Kato product formula (cf. \cite{Calin2011,Kato1978,Weidmann1980}) and
we have
\[
  e^{- t \HRabi} = e^{- t (b^{\dagger}b -g^2 + \Delta \sigma_z)} = \lim_{N\to \infty} (e^{-t (b^{\dagger}b -g^2)/N} e^{-t(\Delta \sigma_z)/N})^N,
\]
in the strong operator topology. Moreover the sequence $\{(e^{-t (b^{\dagger}b -g^2)/N} e^{-t(\Delta \sigma_z)/N})^N\}_{N=1,2,\ldots}$ of Trotter-Kato's approximation operators convergences in the operator norm topology when $N\to\infty$. In fact, we have 
\[
 ||e^{- t (b^{\dagger}b -g^2 + \Delta \sigma_z)} - (e^{-t (b^{\dagger}b -g^2)/N} e^{-t(\Delta \sigma_z)/N})^N||_{\text{op}} =O(N^{-1}), 
\]
where $||A||_{\text{op}}:=\sup_{v\not=0}\frac{||Av||}{||v||}$ denotes the operator norm (see the review paper \cite{IT2009} for the general theory leading  to this fact). Moreover, pointwise uniformly convergence of the iterated integral kernels to the heat kernel follows from the convergence in operator norm topology (\cite{AI2008,IT2009}, cf. \cite{IT2006}). In Section \ref{sec:analyt-prop-heat} we briefly discuss how the pointwise and uniform convergence can be verified directly from the resulting series. 
 
The objective of this section is to compute the integral kernel \(D_N(x,y,t) \) of the $N$-th power operator
\[
  (e^{-t (b^{\dagger}b -g^2)} e^{-t(\Delta \sigma_z)})^N,
\]
explicitly. Concretely, in \S \ref{sec:quantum-rabi-model} we compute the integral kernel \(K(x,y,t)\) of the
operator
\[
  e^{-t (b^{\dagger}b -g^2)} e^{-t(\Delta \sigma_z)},
\]
following the standard procedure for the quantum harmonic oscillator. The  computation of the $N$-th power kernel  
is divided into a scalar part in \S \ref{sec:integral} and a non-commutative part in \S \ref{sec:matrixG}.

In this paper we consider the Hamiltonian \(\HRabi\) as an operator acting on
the Hilbert space \(\mathcal{H} = L^2(\R) \otimes \C^2\).
For convenience of the reader we recall that the creation and annihilation operators are realized by
\[
  a = \frac{1}{\sqrt2} \left( x + \frac{d}{d x} \right), \qquad a^{\dag} = \frac{1}{\sqrt2} \left( x - \frac{d}{d x} \right)
\]
as operators acting on the Schwartz space \(\mathcal{S}(\R) (\subset L^2(\R)\) which has a basis consisting of Hermite functions, cf. \cite{HT1992, AAR1999}).

In order to improve clarity, we do not specify dependencies to the system parameters $g,\Delta$ in the notation for the functions used in
the intermediate computations. In other words, it may be assumed that these functions tacitly depend on the system parameters.

\subsection{Quantum Rabi model and quantum harmonic oscillators}
\label{sec:quantum-rabi-model}

As a first step in the computation of the heat kernel of the QRM, in this subsection we compute the integral kernel \(D(x,y,t)\) of
the operator
\[
  e^{-t (b^{\dagger}b -g^2)} e^{-t(\Delta \sigma_z)}.
\]
First, we notice that by the elementary identity
\[
  e^{-t(\Delta \sigma_z)} =
  \mat{
    e^{-t \Delta} & 0  \\
    0 & e^{ t \Delta}
  },
\]
the integral kernel for the operator \(e^{-t(\Delta \sigma_z)} \) is given by \(D_2(x,y,t) = e^{-t(\Delta \sigma_z)} \delta(x-y)\), where \(\delta\)
is the Dirac measure. Thus, the remainder of this subsection is dedicated to computing the integral kernel of 
\(e^{-b^{\dagger} b - g^2} \). As we have remarked before, the commutation identity  \([b, b^{\dagger} ] = \bI_2\) holds  and
thus the computation of the integral kernel follows the general procedure for the quantum harmonic oscillator.

In particular, if we find the ground state \(\psi_0\) for \(b\), that is, a solution  of \(b \psi_0  = 0\), then the spectrum
of \(b^{\dagger} b\) is equal to \(\Z_{\geq 0}\) with each eigenvalue having multiplicity \(2\).

The general solution of the differential equation system \(b  \psi = 0  \) with \(\psi = {}^{T}(\psi_1,\psi_2) \) is given by
\begin{align*}
  \psi_1(x) &= c_1 e^{-x^2/2  - \sqrt{2} g x}  - c_2 e^{-x^2/2  + \sqrt{2} g x}, \\
  \psi_2(x) &= c_1 e^{-x^2/2  - \sqrt{2} g x}  + c_2 e^{-x^2/2  + \sqrt{2} g x}
\end{align*}
for arbitrary constants \(c_1, c_2 \in \C\). It is clear then that
\begin{align*}
  \bar{\psi}_0^{(1)}(x) = e^{-x^2/2  - \sqrt{2} g x}
  \begin{bmatrix}
    1 \\
    1
  \end{bmatrix},
  \qquad \bar{\psi}_0^{(2)}(x) = e^{-x^2/2  + \sqrt{2} g x}
  \begin{bmatrix}
    -1 \\
    1
  \end{bmatrix},
\end{align*}
are two linearly independent eigenfunctions of \( b^{\dagger} b\) corresponding to eigenvalue \(\lambda = 0 \).
We obtain directly
\[
  \left(\bar{\psi}_0^{(1)},\bar{\psi}_0^{(1)}\right)_{\mathcal{H}} = \left( \bar{\psi}_0^{(2)},
    \bar{\psi}_0^{(2)}\right)_{\mathcal{H}} =  2 e^{2 g^2} \sqrt{\pi},
\]
where \((\cdot,\cdot)_{\mathcal{H}}\) is the inner-product induced in \( \mathcal{H} = L^2(\R) \otimes \C^{2}\) by the usual \( L^2(\R)\) inner-product.

For \(\lambda = n\) the orthonormal eigenstates are given by
\begin{align*}
  \psi_n^{(1)}(x) = \frac{1}{\sqrt{2} \, e^{ g^2}} H_n(x + \sqrt{2} g) \frac{  \bar{\psi}_0^{(1)}}{(\sqrt{\pi} \, n! \, 2^{n})^{1/2}},
  \qquad \psi_n^{(2)}(x) = \frac{1}{\sqrt{2} \, e^{ g^2}} H_n(x - \sqrt{2} g) \frac{\bar{\psi}_0^{(2)}}{(\sqrt{\pi} \, n! \, 2^{n})^{1/2}},
\end{align*}
where \(H_n \) is the \(n\)-th Hermite polynomial. Due to the normalization factor \(e^{-g^2}\), we have
\begin{align*}
  \psi_n^{(1)}(x) = \frac{1}{\sqrt{2} } H_n(x + \sqrt{2} g) \frac{e^{-(x+\sqrt{2}g)^2/2}}{(\sqrt{\pi} \, n! \, 2^{n})^{1/2}} \begin{bmatrix}
    1 \\
    1
  \end{bmatrix},
  \qquad \psi_n^{(2)}(x) = \frac{1}{\sqrt{2} } H_n(x - \sqrt{2} g) \frac{ e^{-(x-\sqrt{2}g)^2/2 }}{(\sqrt{\pi} \, n! \, 2^{n})^{1/2}}  \begin{bmatrix}
    -1 \\
    1
  \end{bmatrix}.
\end{align*}

The heat kernel \(D_1(x,y,t)\) of \(b^{\dagger} b - g^2 \) is given formally by the Schwartz kernel
\[
  D_1(x,y,t) := \sum_{\lambda} \psi_\lambda(x) {}^{T} \psi_\lambda(y) e^{-(\lambda-g^2 ) t},
\]
where the sum is over the eigenvalues \(\lambda\) of \(b^{\dagger} b \) (counting multiplicities) and \(\psi_\lambda(x) \) is the eigenfunction corresponding
to the eigenvalue \(\lambda\). It is left to verify the convergence and that \( D_1(x,y,t) \to \delta(x-y)\bI_2\) as \(t \to 0 \).

Convergence follows component-wise by Mehler's formula (Poisson kernel expression, cf.  \cite{AAR1999,Calin2011}) for Hermite polynomials
\[
  \sum_{n=0}^{\infty} \frac{H_n(x) H_n(y)}{2^n n!} r^n = (1-r^2)^{-1/2} e^{(2 x y r - (x^2 + y^2)r^2)/(1-r^2)},
\]
valid for \( |r | < 1\).

For the second property, recall that, as \(r \to 1 \) the following completeness identity holds
\[
  \sum_{n=0}^{\infty} \frac{H_n(x) H_n(y) e^{- \frac{x^2+y^2}{2}}}{\sqrt{\pi} \, 2^n n!} r^n = \delta(x-y),  \qquad \qquad  r\to 1
\]
in the sense of distributions.

Thus, applying the substitutions \(x \to x \pm \sqrt{2} g  \), \( y\to y \pm \sqrt{2} g  \) we
see that
\[
  \sum_{n=0}^{\infty} \psi_n^{(1)}(x) {}^{T}\psi_n^{(1)}(y) r^n =
  \frac{\delta(x-y)}{2} \matrixOO ,  \qquad \qquad  r\to 1,
\]
and
\[
  \sum_{n=0}^{\infty} \psi_n^{(2)}(x) {}^{T}\psi_n^{(2)}(y) r^n =   \frac{\delta(x-y)}{2} \matrixZZ,  \qquad \qquad  r\to 1,
\]
giving the desired expression when \(r = e^{-t} \).

We write
\[
  D_1(x,y,t) = \sum_{j=1}^{2} \sum_{n=0}^{\infty} \psi_n^{(j)}(x) {}^{T} \psi_{n}^{(j)}(y) e^{-t(n-g^2)},
\]
then, by Mehler's formula we have
\begin{align*}
  \sum_{n=0}^{\infty} \psi_n^{(1)}(x) {}^{T} \psi_{n}^{(1)}(y) u^n u^{g^2} &= \frac{1}{2 \sqrt{\pi}} \sum_{n}^{\infty} \frac{H_n(x+\sqrt{2}g)H_n(y+\sqrt{2}g)}{2^n n!}
                                                        u^n e^{-\frac12( (x+\sqrt{2}g)^2 + (y+\sqrt{2}g)^2) } u^{-g^2} \matrixOO \\
  &=\frac{1}{2 \sqrt{\pi (1-u^2)}} \exp\left( -\frac{1-u}{1+u} \frac{(x+y + 2\sqrt{2}g)^2}{4} - \frac{1+u}{1-u}\frac{(x-y)^2}{4} \right)  u^{-g^2} \matrixOO,
\end{align*}
 with \(u := e^{-t} \), and similarly, we obtain
\begin{align*}
  \sum_{n=0}^{\infty} \psi_n^{(2)}(x) {}^{T} \psi_{n}^{(2)}(y) u^n u^{g^2} &=\frac{1}{2 \sqrt{\pi (1-u^2)}} \exp\left( -\frac{1-u}{1+u} \frac{(x+y - 2\sqrt{2}g)^2}{4} - \frac{1+u}{1-u}\frac{(x-y)^2}{4} \right)  u^{-g^2} \matrixZZ.
\end{align*}

By factoring out common terms we obtain
\begin{align*}
  D_1(x,y,t) = \frac{1}{ \sqrt{\pi (1-u^2)} } &\exp\left(  -\frac{1-u}{1+u}  \frac{((x+y)^2 + 8 g^2)}{4} -  \frac{1+u}{1-u} \frac{(x-y)^2}{4} \right) u^{-g^2}\\
             & \times \frac12 \left( \exp\left( -\frac{1-u}{1+u} \sqrt{2} g (x+y)  \right) \bM_{1 1} +\exp\left( \frac{1-u}{1+u} \sqrt{2} g (x+y)  \right) \bM_{0 0} \right),
\end{align*}
where
\begin{equation*}
  \bM_{0 0} := \matrixZZ \qquad  \bM_{1 1} := \matrixOO.
\end{equation*} 

The matrix terms (including the scalar factor \(\frac12\)) are equal to
\[
  \mat{
    \cosh(\frac{1-u}{1+u} \sqrt{2} g (x+y)) & -\sinh(\frac{1-u}{1+u} \sqrt{2} g (x+y)) \\
    -\sinh(\frac{1-u}{1+u} \sqrt{2} g (x+y)) & \cosh(\frac{1-u}{1+u} \sqrt{2} g (x+y)) \\
  } = \exp\left(- \frac{1-u}{1+u} \sqrt{2} g (x+y) \sigma_x \right),
\]
and the final expression for the heat kernel of \(b^{\dagger} b - g^2\) is
\begin{align*}
 D_1(x,y,t) = \frac{1}{\sqrt{\pi (1-u^2)}}\exp\left( -\frac{1-u}{1+u} \frac{((x+y)^2 + 8 g^2)}{4} - \frac{1+u}{1-u}\frac{(x-y)^2}{4} \right) & u^{-g^2} \exp\left(- \frac{1-u}{1+u} \sqrt{2} g (x+y) \sigma_x \right).
\end{align*}

Summarizing the discussion above, we have the following explicit description for \(K(x,y,t)\).

\begin{prop}
\label{prop:local_kernel}
The integral kernel \(D(x,y,t)\) for \(e^{-t (b^{\dagger}b -g^2)} e^{-t(\Delta \sigma_z)}  \) is given,
 with $u=e^{-t}$, by
\begin{align*}
  D(x,y,t) 
 =  \frac{u^{-g^2}}{\sqrt{\pi (1-u^2)}} & \exp\left( -\frac{1-u}{1+u} \frac{((x+y)^2 + 8 g^2)}{4} -
      \frac{1+u}{1-u}\frac{(x-y)^2}{4} \right)  \\
      & \times \exp\left(- \frac{1-u}{1+u} \sqrt{2} g (x+y) \sigma_x \right) e^{-t(\Delta \sigma_z)}.
\end{align*}
\end{prop}

\begin{proof}
Since 
\[
 D(x,y,t) =  \int_{-\infty}^{\infty} D_1(x,z,t) D_2(z,y,t) d z
\]
the desired expression follows from the definition of the Dirac distribution \(\delta\). 
\end{proof}

To simplify later computations we write
\begin{align*}
  D(x,y,t)  = K_0(x,y,u) \exp \left(-  2 g^2  \frac{1-u}{1+u}   \right)  \exp\left(- \frac{1-u}{1+u} \sqrt{2} g (x+y) \sigma_x \right) e^{-t(\Delta \sigma_z)},
\end{align*}
with
\begin{equation}
  \label{eq:K0}
 K_0(x,y,u) := \frac{u^{-g^2}}{\sqrt{\pi (1-u^2)}} \exp\left( - \frac{1+u^2}{2(1-u^2)} (x^2 + y^2) +  \frac{2 u x y}{1-u^2} \right).
\end{equation}

\subsection{The $N$-th power kernel $D_N(x,y,t)$}

In the remainder of this section we compute explicitly the integral kernel \(D_N(x,y,t) \) of the
operator
\[
  \left(e^{-t (b^{\dagger}b -g^2)} e^{-t(\Delta \sigma_z)} \right)^N
\]
given by the integral
\begin{equation}
  \label{eq:integralbasic}
  \int_{-\infty}^{\infty} \cdots \int_{-\infty}^{\infty} D(x,v_1,t) D(v_1,v_2,t) \cdots D(v_{N-1},y,t) d v_{N-1} d v_{N-2} \cdots d v_1.  
\end{equation} 

Writing \(v_0 = x\) and \(v_N = y\), we see that the integrand of \eqref{eq:integralbasic} is given by the product of the scalar factor
\[
  \frac{u^{-N g^2}}{(\pi (1-u^2))^{N/2}} \exp\left(\sum^{N}_{i=1} \left(-   \frac{1+u^2}{2(1-u^2)} (v_i^2 + v_{i-1}^2)
      +   \frac{2 u v_i v_{i-1}}{1-u^2}  -  2  g^2 \frac{1-u}{1+u} \right) \right)
\]
and the matrix factor
\begin{equation}
  \label{eq:matterms}
  \overrightarrow{\prod}_{i=1}^{N} \Bigg\{\left[\cosh\left( \sqrt{2} g \frac{1-u}{1+u}  (v_{i-1}+v_i)\right) \bI - \sinh\left( \sqrt{2} g \frac{1-u}{1+u}  (v_{i-1}+ v_i)\right)\bJ \right] u^{\Delta\sigma_z}\Bigg\},
\end{equation}
where \(\overrightarrow{\prod}_{i=1}^{N}A_i\) denotes the (ordered) product $A_1A_2\cdots A_N$ of the matrices $A_i$'s and we write $\bJ = \sigma_x= \mat{0&1\\1&0 }$ for
simplicity.

Let us introduce some general notation. We write \(\btZ{k} \) for \(\{0,1\}^k\) with \(k \geq 1\), both as a set and as an abelian group (that is, for the group \((\Z/2\Z)^{k} = \Z/2\Z \times \Z/2\Z \times \cdots \times \Z/2\Z\)($\, k$-times)), for \( k= 0\) we define \(\btZ{0} = \{0\} \) both as a set and (trivial) group. To simplify the notation, at times we consider an element \(\bs \in \btZ{k}\) as a function \(\bs : \{0,1,\cdots,k\} \to \{0,1\}\) where \(\bs(i)\) is the \(i\)-th component of
\(\bs \in \btZ{k} \). 

Since, for \(\alpha\in \R \), we have
\begin{align*}
  \cosh\left(\alpha\right) \bI - \sinh\left(\alpha\right)\bJ = \frac12\left( \bI + \bJ \right) e^{-\alpha}
  + \frac12\left( \bI - \bJ \right) e^{\alpha},
\end{align*}
the multiplication of matrices in \eqref{eq:matterms} gives a linear combination of terms
\begin{equation}
\label{eq:non-com-term}
  G_N(u,\bs) \prod_{i=1}^{N}\exp\left( (-1)^{\bs(i)} \sqrt{2} g \frac{1-u}{1+u} (v_i + v_{i-1}) \right),
\end{equation}
where the choice of  \(\bs \in \btZ{N}\) depends on the factors \(\sinh\) and \(\cosh\) appearing in the expansion of \eqref{eq:matterms} and
\(G_N(u, \bs)\) is a matrix-valued function given by
\[
  G_N(u,\bs) :=  \frac{1}{2^N} \overrightarrow{\prod}_{i=1}^{N}[\bI+(-1)^{1-\bs(j)}\bJ]u^{\Delta\sigma_z}.
\]

In addition, for \(\bs \in \btZ{N}\),  by defining
\begin{align} \label{def:integral}
  I_N &(v_0,v_N,u,\bs) \nonumber \\
 & := \frac{u^{-N g^2}}{(\pi (1-u^2))^{N/2}} \int_{-\infty}^{\infty} \cdots \int_{-\infty}^{\infty} \exp\left(\sum^{N}_{i=1} \left(-   \frac{1+u^2}{2(1-u^2)} (v_i^2 + v_{i-1}^2) +  \frac{2  u v_i v_{i-1}}{1-u^2} -  2 g^2 \frac{1-u}{1+u} \right) \right) \nonumber \\
                  &\quad \times \exp\left(\sqrt{2} g  \frac{1-u}{1+u} \sum_{i=1}^{N} (-1)^{\bs(i)}  (v_i + v_{i-1}) \right) d v_{N-1} d v_{N-2} \cdots  d v_1,
\end{align}
we see that $D_N(x,y,t)$ is given by
\begin{align}
  \label{eq:Npkernel}
  D_N(x,y,t) = \sum_{\bs \in \btZ{N}} G_N(u,\bs) I_N(x,y,u, \, \bs).
\end{align}

\subsection{Scalar part} \label{sec:integral}

The computation of \(D_N(x,y,t)\) is, by \eqref{eq:Npkernel}, is divided into a scalar part, given by
\(I_N(x,y,u, \, \bs)\) and a non-commutative part \(G_N(u,\bs)\). In this subsection we compute the integrals
in the expression of \(I_N (v_0,v_N,u,\bs) \) via multivariate Gaussian integration.

Notice that the variables \(v_0\) and \(v_N\) are not to be integrated in \eqref{def:integral}. Therefore,
\(I_N(v_0,v_N,u,\bs)\) can be rewritten as
\begin{align*}
  \frac{u^{-N g^2}}{(\pi (1-u^2))^{N/2}} &\exp\left( -\frac{1+u^2}{2(1-u^2)}(v_0^2 + v_N^2) + \sqrt{2}g \frac{1-u}{1+u}((-1)^{\bs(1)}v_0 + (-1)^{\bs(N)}v_N) - 2 N g^2 \frac{1-u}{1+u}  \right) \\
  &\int_{-\infty}^{\infty} \cdots \int_{-\infty}^{\infty} \exp\left(- \frac{1+u^2}{1-u^2} \sum^{N-1}_{i=1} v_i^2 +  \frac{2  u}{1-u^2} \sum_{i=1}^{N-2}  v_i v_{i+1}  + \frac{2 u }{1-u^2} (v_0 v_1 +  v_{N-1} v_N) \right) \\
                  & \times \exp\left(  \sqrt{2} g \frac{1-u}{1+u}  \sum_{i=1}^{N-1} ((-1)^{\bs(i+1)} + (-1)^{\bs(i)}) v_i \right) d v_{N-1} d v_{N-2} \cdots  d v_1.
\end{align*}

The quadratic form in variables $v_i\,(i=1,2,\ldots, N-1)$ inside of the exponential in the integrand above is equal to
\[
  (1+u^2) \sum^{N-1}_{i=1} v_i^2 - 2  u \sum_{i=1}^{N-2} v_i v_{i+1} = {}^{T} \bv \bA_{N-1} \bv,
\]
for a vector \(\bv \in \R^{N-1}\) and tridiagonal matrix \(\bA_{N-1}\) given by
\[
   \bv = {}^{T} (v_1,v_2,\cdots,v_{N-1}), \qquad \bA_{N-1} =
   \begin{pmatrix}
     1+u^2 & -u & 0  & \cdots & 0& 0 \\
     -u & 1 + u^2 & -u & \cdots & 0& 0 \\
     0 & -u &1 + u^2 &  \cdots & 0& 0 \\
     \vdots & \vdots & \vdots & \ddots & \vdots & \vdots  \\
     0 & 0 & 0 & \cdots & 1+u^2 & -u  \\
     0 & 0 & 0 & \cdots & -u & 1+u^2
   \end{pmatrix}.
\]
  
Moreover, by defining
\[
  \mathbf{B}(\bs) := \frac{2 u}{1-u^2} \left( v_0 \bm{e}_1  +  v_N \bm{e}_{N-1} \right) +  \sqrt{2} g \frac{1-u}{1+u}  \mathbf{C}(\bs),
\]
where \(\bm{e}_i\) is the \(i\)-th standard basis vector of \(\R^{N-1}\) and
\[
  \mathbf{C}(\bs) := {}^T\left[(-1)^{\bs(1)} + (-1)^{\bs(2)}, (-1)^{\bs(2)} + (-1)^{\bs(3)}, \cdots, (-1)^{\bs(N-1)} + (-1)^{\bs(N)}  \right],
\]
we write \(I_N(v_0,v_N,u,\bs)\) as
\begin{align} \label{eq:integral_1}
  \frac{u^{-N g^2}}{(\pi (1-u^2)^{N/2}} &\exp \left( -\frac{1+u^2}{2(1-u^2)}(v_0^2 + v_N^2) + \sqrt{2}g \frac{1-u}{1+u}((-1)^{\bs(1)}v_0 + (-1)^{\bs(N)}v_N) - 2 N  g^2 \frac{1-u}{1+u}   \right) \\ \nonumber
                                     & \times \int_{-\infty}^{\infty} \cdots \int_{-\infty}^{\infty} \exp \left( - \frac{1}{1-u^2} {}^{T} \bv \bA_{N-1} \bv + {}^T \mathbf{B}(s) \bv \right)  d v_{N-1} d v_{N-2} \cdots  d v_1.
\end{align}

Next, we obtain the expression for \(\det(\bA_{N-1})\). For that, we need a lemma on Chebyshev polynomials of the second kind \(U_n(x)\), defined by the three-term recurrence relation
\begin{align*}
  U_{n+1}(x) &= 2 x U_n(x) - U_{n-1}(x),
\end{align*}
with initial values \(U_0(x) = 1\) and \(U_1(x) = 2 x\).

\begin{lem}
  \label{lem:chebyshev}
  For \(n \in \Z_{\geq 0}\), we have
  \[
    U_n \left(-\frac{1+u^2}{2u} \right) = (-1)^n \frac{1-u^{2(n+1)}}{u^n (1-u^2)}.
  \]
\end{lem}

\begin{proof}
  Set \(z = -\frac{1+u^2}{2u}\). Then, the result clearly holds for
  \(U_0(z)\) and
  \[
    U_1(z) =  2 z =  -\frac{1+u^2}{u} = - \frac{1-u^4}{u(1-u^2)}.
  \]
  The recurrence relation for \(U_n(z) \)gives
  \begin{align*}
    U_{n+1}(z) = 2 z U_n(z) - U_{n-1}(z)  &= - (-1)^n \frac{1+u^2}{u}  \left( \frac{1-u^{2(n+1)}}{u^n (1-u^2)} \right) - (-1)^{n-1}  \left( \frac{1-u^{2 n}}{u^{n-1} (1-u^2)} \right) \\
    &=  (-1)^{n+1} \frac{1}{u^{n+1}(1-u^2)} \left( (1+u^2)(1-u^{2(n+1)}) - u^2(1-u^{2n})) \right). \\
                                          &= (-1)^{n+1} \frac{1-u^{2(n+2)}}{u^{n+1}(1-u^2)},
  \end{align*}
  as desired.
\end{proof}

\begin{lem}
  \label{lem:matA}
  For \(N \geq 2 \), the matrix \(\bA_{N-1}\) is positive definite and its determinant is given by
  \[
    \det(\bA_{N-1}) = (1+u^2 + u^4 + \cdots  + u^{2(N-1)}) = \frac{1-u^{2 N}}{1-u^2}.
  \]
  Furthermore, the inverse of \(\bA_{N-1}\) is symmetric and given by
  \[
    (\bA_{N-1}^{-1})_{i j} = u^{j-i} \frac{(1-u^{2 i})(1-u^{2(N -j)})}{(1-u^{2 N})(1-u^2)},
  \]
  for \(i \leq j\).
\end{lem}

\begin{proof}
  The matrix \(\bA_{N-1}\) is symmetric and since \(1+u^2 > 2 u^2 \) for \( 0 < u <1 \), by the Gershgorin circle
  theorem (see \cite{Var2002})
  all the eigenvalues of \(\bA_{N-1} \) are positive. Therefore, \(\bA_{N-1}\) is positive
  definite (see also \cite{AF2011}). The determinant expression is obtained by direct computation.
  From \cite{daFonseca2001}, it is known that the inverse $\bA_{N-1}$ is given by
  \[
    (\bA_{N-1}^{-1})_{i j} = (-1)^{i+j+1} \frac1u \frac{U_{i-1}(-\frac{1+u^2}{2 u}) U_{N-1-j}(-\frac{1+u^2}{2 u})}{U_{N-1}(-\frac{1+u^2}{2 u})}
  \]
  for \(i \leq j\) and where \(U_n(x)\) is the Chebyshev polynomials of the second kind. The desired expression then follows from Lemma \ref{lem:chebyshev}.
\end{proof}

Let us introduce notation to simplify the expression of  \(I_N(x,y,u, \,\bs) \). For \(\bs \in \btZ{N}\) and \(i,j \in \{1,2,\cdots,N\}\), define
\begin{gather} \label{eq:notation}
  \eta_i(\bs) := (-1)^{\bs(i)}+(-1)^{\bs(i+1)},\\
  \Lambda^{(j)}(u) := u^{j-1} \left( 1 - u^{2 (N-j) +1} \right), \qquad   \Omega^{(i,j)}(u) =  u^{j-i} \left(1-u^{2 i} \right) \left( 1 -u^{2 (N - j)} \right). \nonumber
\end{gather}

\begin{thm}
  \label{thm:IN}
  For \(N \in \Z_{\geq 1} \), we have
  \begin{align} \label{eq:IN}
  I_N(x,y,&u, \,\bs) =  K_0(x,y, u^N) \exp \left( \sqrt{2} g \frac{(1-u)}{(1-u^{2 N})}  \sum_{j=1}^{N} (-1)^{\bs(j)}\left(  x \Lambda^{(j)}(u) + y \Lambda^{(N-j+1)}(u)  \right) \right)   \nonumber   \\
                &\times  \exp \left(    \frac{  g ^2 (1-u)^2}{2(1+u)^2 (1-u^{2 N})} \bigg(  \sum_{i=1}^{N-1} \eta_i(\bs)^2 \Omega^{(i,i)}(u)  + 2 \sum_{i<j} \eta_i(\bs) \eta_j(\bs) \Omega^{(i,j)}(u) \bigg) -   \frac{2 N g^2 (1-u)}{1+u} \right).
\end{align}
\end{thm}

\begin{proof}
  Since \(\bA_{N-1}\) is positive definite by Lemma \ref{lem:matA}, by multivariate Gaussian integration (see e.g. \cite{Chirik2009}) in \eqref{eq:integral_1} we obtain
  \begin{align*}
    \int_{-\infty}^{\infty} \cdots \int_{-\infty}^{\infty} & \exp \left( - \frac{1}{1-u^2}  {}^{T} \bv \bA_{N-1} \bv + {}^T \mathbf{B}(s) \bv \right)  d v_{N-1} d v_{N-2} \cdots  d v_1 \\
                            &= \sqrt{ \frac{(1-u^2)^{N-1} \pi^{N-1}}{\det(\bA_{N-1})} } \exp\left(  \frac{1-u^2}{4} {}^{T} \mathbf{B}(s) (\bA_{N-1})^{-1} \mathbf{B}(s) \right) \\
                            &= \sqrt{\frac{\pi^{N-1}(1-u^2)^N}{(1-u^{2 N})}} \exp\left(  \frac{1-u^2}{4} {}^{T} \mathbf{B}(s) (\bA_{N-1})^{-1} \mathbf{B}(s) \right),
  \end{align*}

  Thus, we have
  \begin{align*} 
    I_N &(v_0,v_N,u,\bs) \nonumber \\
   &= \frac{u^{-N g^2}}{\sqrt{\pi (1-u^{2 N})}} \exp\left( -\frac{1+u^2}{2(1-u^2)}(v_0^2 + v_N^2) + \sqrt{2}g \frac{1-u}{1+u}\left((-1)^{\bs(1)}v_0 + (-1)^{\bs(N)}v_N \right) - 2 N g^2 \frac{1-u}{1+u}  \right)  \nonumber \\
    &\qquad \times \exp\left(  \frac{1-u^2}{4} {}^{T} \mathbf{B(\bs)} \bA_{N-1}^{-1} \mathbf{B(\bs)} \right). 
  \end{align*}

  From the definitions, we see that
  \begin{align*}
    {}^{T} \mathbf{B(\bs)} \bA_{N-1}^{-1} \mathbf{B(\bs)}  = &  \left(  \frac{2 u  }{1-u^2} \left( v_0 {}^T \bm{e}_1  + v_N {}^T \bm{e}_{N-1}\right) \right) \bA_{N-1}^{-1} \left( \frac{2 u  }{1-u^2} \left( v_0 \bm{e}_1  + v_N \bm{e}_{N-1}\right)\right) \\
   & + 2 \left( \frac{2 u  }{1-u^2} \left( v_0 {}^T \bm{e}_1  + v_N {}^T \bm{e}_{N-1}\right) \right) \bA_{N-1}^{-1} \left( \sqrt{2} g \frac{1-u}{1+u}  \mathbf{C}(\bs) \right) \\
   &+ \left( 2  g^2 \frac{(1-u)^2}{(1+u)^2} \right) {}^T\mathbf{C}(\bs) \bA_{N-1}^{-1} \mathbf{C}(\bs) ,
  \end{align*}
  the second line is justified by the symmetry of the inverse of the matrix \(\bA_{N-1}^{-1}\).  By Lemma \ref{lem:matA}, we have
  \begin{align*}
    &\frac{1-u^2}{4} \left(  \frac{2 u  }{1-u^2} \left( v_0 {}^T \bm{e}_1  + v_N {}^T \bm{e}_{N-1}\right) \right)  \bA_{N-1}^{-1} \left( \frac{2 u  }{1-u^2} \left( v_0 \bm{e}_1  + v_N \bm{e}_{N-1}\right)\right) \\
     &\qquad = \frac{ u^2}{(1-u^2)} \left( \frac{1-u^{2(N-1)}}{1-u^{2 N}} (v_0^2 + v_N^2 ) + 2 u^{N-2}\frac{1-u^2}{1-u^{2 N}} v_0 v_N \right),
  \end{align*}
  adding the term \(-\frac{1+u^2}{2(1-u^2)}(v_0^2 + v_N^2) \) we obtain
  \[
    -\frac{1+u^{2 N}}{2(1-u^{2 N})}(v_0^2 + v_N^2) + \frac{2 u^N}{1-u^{2 N}} v_0 v_N,
  \]
  giving the expression \(K_0(x,y,u^N) \) of \eqref{eq:IN} by setting \(v_0=x \) and \(v_N = y \)).

  For the second term, we have
  \begin{align*}
    &  \frac{2(1-u^2)}{4} \left( \frac{2 u  }{1-u^2} \left( v_0 {}^T \bm{e}_1  + v_N {}^T \bm{e}_{N-1}\right) \right) \bA_{N-1}^{-1} \left( \sqrt{2} g \frac{1-u}{1+u}  \mathbf{C}(s) \right) = \frac{\sqrt{2} g(1-u)}{(1-u^{2 N})(1+u)} \\
    & \qquad \times \left( v_0 \sum_{j=1}^{N-1} u^j (1-u^{2 (N-j)}) ((-1)^{\bs(j+1)}+(-1)^{\bs(j)}) + v_N \sum_{j=1}^{N-1} u^{N-j} (1-u^{2 j}) ((-1)^{\bs(j+1)}+(-1)^{\bs(j)})\right).
  \end{align*}
  By rewriting the first sum by using the identity
  \[
    u^j (1-u^{2 (N-j)}) +  u^{j-1} (1-u^{2 (N-j+1)}) = (1+u) u^{j-1} (1-u^{2 (N-j)+1}),
  \]
  and adding the term \( (-1)^{\bs(1)} (1-u^{2 N}) v_0\), we obtain the expression
  \[
    v_0 \sum_{j=1}^{N} (-1)^{\bs(j)} u^{j-1}(1-u^{2(N-j)+1}),
  \]
  and similarly for the second sum, giving the expression in the sum in the first line of \eqref{eq:IN}

  Finally, the term
  \[
    \frac{1-u^2}{4} \left( 2  g^2 \frac{(1-u)^2}{(1+u)^2} \right) \left({}^T\mathbf{C}(s) \bA_{N-1}^{-1} \mathbf{C}(s) \right),
  \]
  is given by
  \begin{align*}
    \frac{  g ^2 (1-u)^2}{2(1+u)^2 (1-u^{2 N})} \bigg( & \sum_{i=1}^{N-1} ( (-1)^{\bs(i+1)}+(-1)^{\bs(i)} )^2 (1-u^{2 i})(1-u^{2(N-i)}) \\
              & + 2 \sum_{i<j} ( (-1)^{\bs(i+1)}+(-1)^{\bs(i)} ) ( (-1)^{\bs(j+1)}+(-1)^{\bs(j)}) u^{j-i} (1-u^{2 i})(1-u^{2(N-j)}) \bigg),
  \end{align*}
  yielding the expression in the second line of \eqref{eq:IN}. The proof is completed by setting \( v_0=x \) and \(  v_N =y \).
\end{proof}

\subsection{Non-commutative part}
\label{sec:matrixG}

In this section we explicitly describe the matrix-valued function $G_k(u,\Delta,s)$ for \(k \geq 1 \), then
by using the resulting expression and  the previous computation for \(I_k(x,y,u,\bs)\) we give
a limit formula for the heat kernel of QRM reminding us of a Riemannian sum. 

To simplify the notation, we denote by \(\bM_{i j}\), for \(i,j = 0,1\),  the matrices
\begin{equation}
  \label{eq:fourmatrices}
  \bM_{0 0} := \matrixZZ, \bM_{0 1} := \matrixZO, \bM_{1 0} := \matrixOZ, \bM_{1 1} := \matrixOO,
\end{equation}
where we included the previously defined matrices \(\bM_{0 0}\) and \(\bM_{1 1}\) for reference.

\begin{prop}
\label{prop:matrices}
  For \(\bs\in \btZ{k}\), we have
  \begin{equation*} 
        G_k(u,\bs) = \frac{ \prod^{k-1}_{i=1}(1+(-1)^{\bs(i) -\bs(i+1)} u^{2\Delta})}{u^{(k-1)\Delta} 2^k} \bM_k(\bs) \matrixU{\Delta},
  \end{equation*}
  where the matrix \(\bM_k(\bs)\) is given by
  \[
    \bM_k(\bs) =
      \mat{ 1  & (-1)^{1-\bs(1)} \\ (-1)^{1-\bs(1)} & 1 }
       \overrightarrow{\prod}_{i=1}^{k-1}  \left( \matrixId +
      \mat{
          -|\bs(i) - \bs(i+1)|  & -(\bs(i+1) - \bs(i)) \\
          (\bs(i+1) - \bs(i)) & -|\bs(i) - \bs(i+1)|
        } \right).
  \]
\end{prop}

Before proving Proposition \ref{prop:matrices}, we observe that the matrix \(\bM_k(\bs) \) is only one of the
matrices \(\bM_{0 0}\), \(\bM_{0 1}\), \(\bM_{1 0}, \) and \(\bM_{1 1}\) (see \eqref{eq:fourmatrices}).
In fact, \(\bM_k(\bs) \) only depends on the first and last entry of \(\bs \in \btZ{k}\). 

\begin{lem}
  \label{lem:matword}
  Let \(\bs \in \btZ{k}\). If \(\bs(1) = i \) and \(\bs(k) = j \), then
  \[
     \bM_k(\bs)  = \bM_{i j},
   \]
   for \(i,j = 0,1 \).
\end{lem}

\begin{proof}
  Let us consider only the case \(\bs(1)=0\), since the case \(\bs(1) =1 \) is proved in a similar fashion.
  Notice that if \(\bs(i+1) = \bs(i) \), the matrix inside the product in the definition of \(\bM_k(\bs)\)
  corresponding to the index \(i \in \{1,2,\cdots,k-1\}\) is the identity. Let us consider the vector 
 \(\bs\) as a word on the alphabet \(\Z_2 = \{0,1\}\) in the standard way and \(\bar{\bs}\) the word resulting
 of removing contiguous occurrences of ones or zeros. Then, if \(\bs(1) = 0\), and \(\bs(k)= 0 \),
 \(\bar{\bs}= 0 (1 0)^k \)  with \(k \in \Z_{\geq 0}\) and  here exponentiation means concatenation of words.
 From the definition of \(\bM_k(\bs) \) we see then that
  \[
   \matrixZZ \left( \mat{ 0 & -1 \\ 1 & 0 } \mat{ 0 & 1 \\ -1 &0 }\right)^k  = \matrixZZ,
  \]
  since the expression in the parenthesis is equal to the identity matrix. Similarly, for \(\bs(1) = 0 \) and
  \(\bs(k) = 1 \), we have \(\bar{\bs}= 0 (1 0)^k 1\) for \(k \in \Z_{\geq 0}\). Therefore,
  \[
    \matrixZZ \left( \mat{ 0 & -1 \\ 1 & 0 } \mat{ 0 & 1 \\ -1 &0 }\right)^k \mat{ 0 & -1 \\ 1 & 0 } = \matrixZO ,
  \]
  and the result follows.

\end{proof}

\begin{proof}[Proof of Proposition \ref{prop:matrices}]
  The case \(k=1 \) is trivial. Furthermore we easily by direct computation that
  \begin{align*}
    G_2(u,(0,0)) &= \frac{1+u^{2\Delta}}{2^2 u^{\Delta}} \matrixZZ \matrixU{\Delta}, & G_2(u,(0,1)) &= \frac{1-u^{2\Delta}}{2^2 u^{\Delta}} \matrixZO \matrixU{\Delta} \\
    G_2(u,(1,0)) &= \frac{1-u^{2\Delta}}{2^2 u^{\Delta}} \matrixOZ \matrixU{\Delta},  &   G_2(u,(1,1)) &= \frac{1+u^{2\Delta}}{2^2 u^{\Delta}} \matrixOO \matrixU{\Delta}.
  \end{align*}
  Now, we suppose the result holds for \(k \in \Z_{\geq2}\). Let \(\bs \in \btZ{k+1}\) and consider
  \[
    G_{k+1}(u,\bs) =  \frac{1}{2^{k+1}} \overrightarrow{\prod}_{i=1}^{k+1}[\bI+(-1)^{1-\bs(j)}\bJ] \matrixU{\Delta},
  \]
  by the hypothesis, this is just
  \[
    \frac{ \prod^{k-1}_{i=1}(1+(-1)^{\bs(i) -\bs(i-1)} u^{2\Delta})}{u^{(k-1)\Delta} 2^{k+1}} \bM_k(s') \matrixU{\Delta} [\bI+(-1)^{1-\bs(k+1)}\bJ]
    \matrixU{\Delta},
  \]
  with \(\bs' \in \btZ{k}\).

  Suppose that \(\bs(k+1) = 0\), we consider the product
  \[
    \bM_k(\bs') \matrixU{\Delta} \matrixZZ,
  \]
  we are going to verify the result for the possible combinations of \(\bs(1)\) and \(\bs(k)\).

  First, if \(\bs(1)=0\) and \(\bs(k)=0 \), by Lemma \ref{lem:matword}, the above product is
  \[
    \matrixZZ \matrixU{\Delta} \matrixZZ = \frac{1+u^{2\Delta}}{u^{\Delta}} \matrixZZ = \frac{1+u^{2\Delta}}{u^{\Delta}}   \bM_k(s') \matrixId,
  \]
  which is the desired expression.
  In the case \(\bs(1)=0\) and \(\bs(k)=1 \), we have
  \[
    \matrixZO \matrixU{\Delta} \matrixZZ = \frac{1-u^{2\Delta}}{u^{\Delta}} \matrixZZ = \frac{1-u^{2\Delta}}{u^{\Delta}}  \bM_k(s') \mat{ \phantom{-}0 & \phantom{-}1 \\ -1 & \phantom{-}0 }.,
  \]
  while  in the case \(\bs(1)=1\) and \(\bs(k)=0 \) we have
  \[
    \matrixOZ \matrixU{\Delta} \matrixZZ = \frac{1+u^{2\Delta}}{u^{\Delta}} \matrixOZ = \frac{1+u^{2\Delta}}{u^{\Delta}}   \bM_k(s') \matrixId,
  \]
  and finally, the in the case \(\bs(1)=1\) and \(\bs(k)=1 \), we have
  \[
    \matrixOO \matrixU{\Delta} \matrixZZ = \frac{1-u^{2\Delta}}{u^{\Delta}} \matrixOZ = \frac{1-u^{2\Delta}}{u^{\Delta}}   \bM_k(s') \mat{ \phantom{-}0 & \phantom{-}1 \\ -1 & \phantom{-}0 }.
  \]
  The case of \(\bs(k+1) =1  \) is completely analogous.
\end{proof}

By \eqref{eq:Npkernel}, the heat kernel of the QRM is given by the limit expression
\begin{align*}
  \KRabi(x,y,t; g, \Delta) =\lim_{N\to \infty} \sum_{\bs \in \btZ{N}} G_N(u^{\frac1N},\bs) I_N(x,y,u^{\frac1N}, \, \bs).
\end{align*}

To deal with the sum over \( \btZ{N}\) in the expression above, we introduce  a  partition of \( \btZ{N} \).

\begin{dfn}
  \label{dfn:subsets}
  Let \(N \in \Z_{\geq 1} \) and \(i,j \in \Z_2\).
  \begin{enumerate}
  \item The subset \(\setC{N}{i}{j} \subset \btZ{N} \) is given by
    \[
      \setC{N}{i}{j} = \{ \bs \in \btZ{N} \, |\,  \bs(1) = i, \bs(N) = j \}.
    \]
  \item For \( 3 \leq  k \leq N \) the subset \(\setA{k,N}{i}{j} \subset \btZ{N}\) is given by
    \[
      \setA{k,N}{i}{j} = \{ \bs \in \btZ{N} \, |\, \bs(1) = i , \bs(k-1) = 1 - j, \bs(n) = j \text{ for } k \leq n \leq N \},
    \]
  \item We have
    \begin{align*}
      \setA{1,N}{0}{0} &= \{  (0,0,0,0,\cdots,0 ) \}, &\setA{2,N}{0}{1} &= \{  (0,1,1,1,\cdots,1 )\}, \\
      \setA{1,N}{1}{1} &= \{  (1,1,1,,1,\cdots,1) \}, &\setA{2,N}{1}{0} &= \{  (1,0,0,0,\cdots,0 ) \},
    \end{align*}
    and \(\setA{k,N}{i}{j} = \emptyset \) for \(k =1,2 \) if it is not one of the four sets above.
  \end{enumerate}
\end{dfn}

For \(N \geq 2 \), the sets \(\setA{k,N}{i}{j} \subset \btZ{N}\) form a partition of \(\btZ{N}\), that is,
\begin{equation}
  \label{eq:partiZ}
   \btZ{N} = \bigsqcup_{ 1 \leq k \leq N} \, \bigsqcup_{i,j \in \Z_2}  \setA{k,N}{i}{j},
\end{equation}
 from where it is clear that for \(i,j \in \Z_2\), we have \(\# \setA{k,N}{i}{j} = 2^{k-3} \) for \(k \geq 3\) and \(\# \setA{n,N}{i}{j} = 1\) if \(n=1,2\).

We frequently use the constant elements \( \bZ_k = (0,0,\cdots,0), \bO_k = (1,1,\cdots,1) \in \btZ{k}\) for \(k \in \Z_{\geq 1} \).
For \(\br \in \btZ{k}\) and \(\bs \in \btZ{\ell}\) with \(k,\ell \in \Z_{\geq 1}\), we denote by \(\br \oplus \bs \in \btZ{k + \ell} \) the element obtained by
concatenation in the natural way.

We note that any element \(\bs \in \setA{k,N}{0}{0}\) for \(k \geq 3 \) can be expressed as
\begin{equation}
  \label{eq:oplusvec}
  \bs = \bar{s} \oplus \bZ_{N-k+1}
\end{equation}
with \(\bar{s} \in \setC{k-1}{0}{1}\). Similar expressions hold for elements of \(\setA{k,N}{0}{1}\), \(\setA{k,N}{1}{0}\) and \(\setA{k,N}{1}{1}\).

Additionally, by Lemma \ref{lem:matword}, the matrix  \(\bM_k(\bs) \) only depends on the first and last entry of \(\bs \in \btZ{k} \)
and thus it is fixed over any subset \(\setC{k}{i}{j} \in \btZ{k}\) for \(i,j =0,1\) (c.f. Definition \ref{dfn:subsets}). In practice, it is convenient to work with the scalar part of the function \(G_k(u,\bs) \) above.

\begin{dfn} 
  For \(k \geq 1 \), the function \(g_k(u,\bs) \) is given by
  \begin{equation*}
    g_k(u,\bs) = \frac{ \prod^{k-1}_{i=1}(1+(-1)^{\bs(i) -\bs(i+1)} u^{2\Delta})}{u^{(k-1)\Delta} 2^k}.
  \end{equation*}
\end{dfn}

The result of Proposition \ref{prop:matrices} is then written as 
\[
  G_k(u,\bs) = g_{k}(u,\bs) \bM_k(\bs) \matrixU{\Delta},
\]
where we note that the degree of \(g_k(u,\bs)\) as a polynomial in \(u^{\Delta}\) is \(2(k-1)\).

With the foregoing notation, the heat kernel $\KRabi(x,y,t)$ of the QRM, given by 
\begin{align*}
  \KRabi(x,y,t; g, \Delta) &= \lim_{N \to \infty} D_N(x,y,u^{\frac1N}) = \lim_{N \to \infty} \sum_{\bs \in \btZ{N}} G_N(u^{\frac1N},\bs) I_N(x,y, u^{\frac1N},\bs),
\end{align*}
is, by  \eqref{eq:IN} and \eqref{eq:partiZ}, equal to
\begin{align} \label{eq:lim1}
   \KRabi(x,y,t; g, \Delta) &= K_0(x,y,u) \lim_{N \to \infty} \sum_{k=1}^N \sum^1_{i,j =0}  \sum_{\bs \in \setA{k,N}{i}{j}}  G_N(u^{\frac1N},\bs) \bar{I}_N(x,y, u^{\frac1N},\bs),
\end{align}
with
\begin{align*} 
  \bar{I}_N(x,y,u,& \,\bs) = \exp  \left( \frac{\sqrt{2} g(1-u)}{1-u^{2 N}} \sum_{j=1}^{N} (-1)^{\bs(j)}\left(  x \Lambda^{(j)}(u) + y \Lambda^{(N-j+1)}(u) \right)   \right)   \nonumber \\
    & \times  \exp \Bigg(  \frac{g^2 (1-u)^2}{2 (1+u)^2 (1-u^{2 N})} \bigg(  \sum_{i=1}^{N-1} \eta_i(s)^2 \Omega^{(i,i)}(u) + 2 \sum_{i<j} \eta_i(s) \eta_j(s) \Omega^{(i,j)}(u)  \bigg) -   \frac{2 N g^2 (1-u)}{1+u} \Bigg).
\end{align*}

Note that for \(k \geq 2 \), by \eqref{eq:oplusvec}, the expression inside the limit in \eqref{eq:lim1} is given by
\begin{align*}
  \sum_{k \geq 2}^N &\bigg( \sum_{\substack{\bs = \bar{\bs} \oplus \bZ_{N-k+1} \\ \bar{\bs} \in \setC{k-1}{0}{1} } }
  + \sum_{\substack{\bs = \bar{\bs} \oplus \bZ_{N-k+1} \\ \bar{\bs} \in \setC{k-1}{1}{1}}}
  + \sum_{\substack{\bs = \bar{\bs} \oplus \bO_{N-k+1} \\ \bar{\bs} \in \setC{k-1}{0}{0}}}
  + \sum_{\substack{\bs = \bar{\bs} \oplus \bO_{N-k+1} \\ \bar{\bs} \in \setC{k-1}{1}{0}}}  G_{N}(u^{\frac1N},\bs) \bar{I}_N(x,y, u^{\frac1N},\bs) \bigg),
\end{align*}
and we remark that \(\setC{1}{i}{j} = \emptyset\) with \( i \neq j \).

Next, we describe how the term \(\bar{I}_N(x,y, u^{\frac1N},\bs)\) factors in each of the sums. For \(\bar{\bs} \in \btZ{k-1}\) with \(k \geq 1\), write
\begin{align*} 
  \bar{I}_N(x,y,u, \,\bar{\bs} \oplus \bZ_{N-k+1} ) &= J^{(k,N)}_0(x,y,u) R_0^{(k,N)}(u,\bar{\bs}),  \nonumber\\
  \bar{I}_N(x,y,u, \,\bar{\bs} \oplus \bO_{N-k+1} ) &= J^{(k,N)}_1(x,y,u) R_1^{(k,N)}(u,\bar{\bs}),
\end{align*}
with functions \(J^{(k,N)}_\mu(x,y,u)\) and \( R_\mu^{(k,N)}(u,\bar{\bs}) \) for \(\mu \in \{0,1\}\) given in Definition \ref{dfn:JR} below.
Notice that in the first line \( \bar{\bs} \in \setC{k-1}{i}{1} \) and in the second line \( \bar{\bs} \in \setC{k-1}{i}{0} \) for \(i=0,1\).

\begin{dfn} \label{dfn:JR}
  For \( k \geq 1\), the function \(J^{(k,N)}_{\mu}(x,y,u)\) is given by
  \begin{align*}
    J^{(k,N)}_{\mu}(x,y,u) &= \exp\left((-1)^\mu \frac{\sqrt{2} g(1-u)}{1-u^{2 N}} \left(  \sum_{j=k}^{N}\left( x \Lambda^{(j)}(u) + y \Lambda^{(N-j+1)}(u)  \right) \right) \right)  \nonumber \\
                           & \times  \exp \left(  \frac{ 2 g^2 (1-u)^2}{ (1+u)^2 (1-u^{2 N})} \bigg(  \sum_{i=k}^{N-1} \Omega^{(i,i)}(u)  \,   + 2 \sum_{i=k}^{N-2} \sum_{j=i+1}^{N-1} \Omega^{(i,j)}(u) \bigg)
                             -  2 N g^2 \frac{  (1-u)}{1+u} \right),
  \end{align*}
  while  \(R_{\mu}^{(k,N)}(u,\bar{\bs}) \) is given, for \(\bar{\bs} \in \btZ{k-1} \), by
  \begin{align*}
    R_{\mu}^{(k,N)}(u,\bar{\bs}) &= \exp \left( \frac{\sqrt{2} g(1-u)}{1-u^{2 N}}  \sum_{j=1}^{k-1} (-1)^{\bar{\bs}(j)} \left( x \Lambda^{(j)}(u) + y \Lambda^{(N-j+1)}(u)  \right)    \right)   \nonumber \\
                                 &\times  \exp \Bigg( \frac{g^2 (1-u)^2}{2 (1+u)^2 (1-u^{2 N})} \Bigg[  \sum_{i=1}^{k-2} \eta_i(\bar{\bs})^2 \Omega^{(i,i)}(u)  + 2 \sum_{i=1}^{k-2} \sum_{j=i+1}^{k-2} \eta_i(\bar{\bs})\eta_j(\bar{\bs}) \Omega^{(i,j)}(u) \\
                                 & \qquad \qquad \qquad \qquad \qquad \qquad \qquad + 4 (-1)^{\mu} \sum_{i=1}^{k-2} \sum_{j=k}^{N-1} \eta_i(\bar{\bs}) \Omega^{(i,j)}(u)  \Bigg] \Bigg).
  \end{align*}  
\end{dfn}

Suppose that \( \bs = \bs_1 \oplus \bZ_{N-k+1} \) with  \(  \bs_1  \in \setC{k-1}{v}{1}\) and \(v \in \{0,1\} \), then it is easy to see that
\[
  G_N(u^{\frac1N},\bs) =  \left( \frac{1-u^{\frac{2\Delta}N}}{2 u^{\frac{\Delta}N}} \right) \left( \frac{1+u^{\frac{2\Delta}N}}{2 u^{\frac{\Delta}N}} \right)^{N-k}  g_{k-1}(u^{\frac1N},\bs_1) \bM_N(\bs)  \matrixU{\frac{\Delta}N},
\]
with similar expressions for other cases. Therefore, the sum inside the limit (starting from \(k=2 \)) is given by
\begin{align} \label{eq:InnerSum1}
   \left( \frac{1-u^{\frac{2\Delta}N}}{2 u^{\frac{\Delta}N}} \right)& \sum_{k \geq 2}^N  \left( \frac{1+u^{\frac{2\Delta}N}}{2 u^{\frac{\Delta}N}} \right)^{N-k} \Bigg[ J_0^{(k,N)}(x,y,u^{\frac1N}) \Bigg( \sum_{v=0}^1 \bM_{v 0} \sum_{ \bs \in \setC{k-1}{v}{1}} g_{k-1}(u^{\frac1N},\bs) R_0^{(k,N)}(u^{\frac1N},\bs) \Bigg) \nonumber \\
  &+  J_1^{(k,N)}(x,y,u^{\frac1N})\Bigg( \sum_{v=0}^1 \bM_{v 1} \sum_{ \bs \in \setC{k-1}{v}{0}} g_{k-1}(u^{\frac1N}, \bs) R_1^{(k,N)}(u^{\frac1N},\bs)  \Bigg) \Bigg]  \matrixU{\frac{\Delta}N}. 
\end{align}

Next, we make some considerations to further simplify the expression of the heat kernel.
First, we notice that the matrix factor
\[
\matrixU{\frac{\Delta}N}
\]
is the identity matrix at the limit $N\to\infty$, so we omit it in the subsequent discussion. Similarly, without loss of generality, we drop the term corresponding to \(k=2\), since it vanishes due to the presence of the factor \((1-u^{2\Delta/N})\). This is analogous to removing a finite number of terms from a Riemann sum.

Summing up, the expression for the heat kernel \(\KRabi(x,y,t; g ,\Delta)\) is given by
\begin{align*}
  & K_0(x,y,u) \lim_{N \to \infty} \Bigg( \frac12 \left(\frac{1+u^{\frac{2\Delta}N}}{2 u^{\frac{\Delta}N}} \right)^{N-1}\left( J^{(1,N)}_0(x,y,u^{\frac1N}) \matrixZZ  + J^{(1,N)}_1(x,y,u^{\frac1N}) \matrixOO  \right) \\
             & + \left( \frac{1-u^{\frac{2\Delta}N}}{2 u^{\frac{\Delta}N}} \right) \sum_{k \geq 3}^N   \left( \frac{1+u^{\frac{2\Delta}N}}{2 u^{\frac{\Delta}N}} \right)^{N-k}   \\
   & \times \Bigg[ J_0^{(k,N)}(x,y,u^{\frac{1}N}) \Bigg( \bM_{0 0} \sum_{ \bs \in \setC{k-1}{0}{1}} g_{k-1}(u^{\frac1N},\bs) R_0^{(k,N)}(u^{\frac{1}N},\bs)  + \bM_{1 0} \sum_{ \bs \in \setC{k-1}{1}{1}} g_{k-1}(u^{\frac{1}N},\bs) R_0^{(k,N)}(u^{\frac1N},\bs) \Bigg) \\
        & + J_1^{(k,N)}(x,y,u^{\frac1N}) \Bigg( \bM_{0 1} \sum_{ s \in \setC{k-1}{0}{0}} g_{k-1}(u^{\frac1N}, \bs) R_1^{(k,N)}(u^{\frac1N},\bs)  + \bM_{1 1} \sum_{ s \in \setC{k-1}{1}{0}} g_{k-1}(u^{\frac1N},  \bs) R_1^{(k,N)}(u^{\frac1N},\bs) \Bigg) \Bigg] \Bigg).
\end{align*}

Notice that the limit in the expression \eqref{eq:InnerSum1} resembles a Riemann sum of the type
\[
  \lim_{N\to \infty} \sinh\left(\frac{t}{N}\right) \sum_{k=1}^{N} f\left( \frac{k t}{N} \right) = \int_0^t f(x) d x,
\]
for a Riemann integrable function \(f:[0,t] \to \R\). However, due to the presence of alternating sums
depending of \(k\) in \( R_\mu^{(k,N)}(u,\bar{\bs})\) and in \(g_{k-1}(u^{\frac{1}N},\bs)\) it is not possible to
interpret the limit directly as a Riemann sum.

\section{Harmonic analysis on \(\Z_2^{k} \)}
\label{sec:four-transf}

Denote by \(\C[\btZ{k}]\) the group algebra of the abelian group \(\Z_2^{k}\). For \(f, h \in \C[\btZ{k}] \)
the elementary identity (Parseval's identity)
\begin{equation}
  \label{eq:sumC}
    \sum_{s \in \btZ{k}} f(s) h(s) = (f * h)(0) = \frac{1}{2^k}\widehat{\left(\widehat{f} \cdot \widehat{h}\right)}(0) = \frac{1}{2^k} \sum_{\rho \in \btZ{k}} \widehat{f}(\rho) \widehat{h}(\rho),
\end{equation} holds, where \(\widehat{f} \) (resp. \(\widehat{h}\) ) is the Fourier transform of \(f\) (resp. \(h\)) defined below (see \eqref{eq:defFour}).

In this section we use the identity \eqref{eq:sumC} to transform the sum appearing \eqref{eq:InnerSum1} into an
expression that can be evaluated as a Riemann sum. First, we compute the Fourier transform of \(g_{k}(u^{\frac1N},  \bs)\), then in \S \ref{sec:four-transf-R} we describe the Fourier transform of \(R_\eta^{(k,N)}(u^{\frac1N},\bs)\). In \S \ref{sec:graph-theoretical} we collect a number of combinatorial results to simplify the  expression of the Fourier transform of \(R_\eta^{(k,N)}(u^{\frac1N},\bs)\). In \S \ref{sec:transformation-sum}, we use identity \eqref{eq:sumC} to simplify the expression \eqref{eq:InnerSum1} and in \S \ref{sec:transf-summ-into} we transform finite sums into definite integrals using the standard method with Riemann-Stieltjes integrations and estimate the order of the residual terms.

We begin by setting the notation and recalling the basic properties of the Fourier transform in \(\Z_2^{k} \), we refer the reader to
\cite{Cecc2008} for more details.
For \(\rho \in \btZ{k} \ \), define the character \(\chi_{\rho}(\bs) \in \btZdual{k} \) by
\begin{equation*} 
  \chi_{\rho}(\bs) := (-1)^{ (\bs | \rho)},
\end{equation*}
where \((\cdot|\cdot) \) is the standard inner product in \(\btZ{k}\). It is known that  all the characters in the dual
group \(\btZdual{k}\) are obtained in this way. Then, for \(f \in \C[\btZ{k}] \), the Fourier transform \(\widehat{f}(\rho)\) is given by
\begin{equation}
  \label{eq:defFour}
  \widehat{f}(\rho) = \mathcal{F}(f) :=  \sum_{\bs \in \btZ{k} } f(\bs) \chi_{\rho}(\bs),
\end{equation}
for \(\rho \in \btZ{k}\). Since \( \widehat{f} \in \C[\btZ{k}]\), the Fourier inversion formula is given by
\[
  f = \frac{1}{2^k} \hat{\hat{f}}.
\]

Next, we equip the set  \( \setC{k+2}{v}{w} \) with a abelian group structure such
that \( \setC{k+2}{v}{w} \simeq \btZ{k}\). We naturally identify an element  \( \bs  \in \setC{k+2}{v}{w}  \) via the
projection \( \bar{\bs} \in \btZ{k}\) given by
\begin{equation}
  \label{eq:projectionVect}
  \bs = (v,s_1,s_2,\cdots,s_k,w) \longmapsto \bar{\bs} = (s_1,s_2,\cdots,s_k).
\end{equation}

Clearly,  the sum \eqref{eq:sumC} may be regarded as a sum over \( \setC{k+2}{v}{w} \) by lifting an element
\( \bs \in \btZ{k} \) to \( \setC{k+2}{v}{w} \) by using the inverse of the projection \eqref{eq:projectionVect}.

In the case of the function \(g_k(u,\bs) \) we define a special notation. 
\begin{dfn} 
  Let \(v,w \in \{0,1\} \). Then, for \(\bs \in \btZ{k}\) with \(k \geq 1 \), define the function \(g^{(v,w)}_k(u,\bs)\) by
  \begin{equation*} 
    g^{(v,w)}_k(u,\bs) := \frac{1}{2^k} (1+(-1)^{v+\bs(1)}u^{2 \Delta}) (1+(-1)^{w+\bs(k)}u^{2\Delta}) \prod^{k-1}_{i=1}(1+(-1)^{\bs(i) +\bs(i+1)} u^{2 \Delta} ).
  \end{equation*}
  In addition, for \(\rho \in \btZ{0}\), define
  \[
    g_0^{(v,w)}(u,\bs) = 1+(-1)^{v+w} u^{2 \Delta}.
  \]
\end{dfn}

For \(\bs \in \setC{k+2}{v}{w}\), we have
\begin{equation*} 
  4 u^{(k+1)\Delta} g_{k+2}(u,\bs) =  g^{(v,w)}_k(u,\bar{\bs}),
\end{equation*}
and in addition, we note that the degree of \(g^{(v,w)}_k(u,\bar{\bs})\) as a polynomial in \(u^{\Delta}\) is  \(2(k+1)\).

For fixed \(u,\Delta \in \R \), the function \(g^{(v,w)}_k(u,\bs)\) is an element of the group algebra \(\C[\btZ{k}]\) of
the abelian group \(\btZ{k}\). Since the parameters \(g,\Delta >0 \) and \(u \in \{0,1\} \) are assumed to be fixed, in the remainder of this
section as it is obvious we may omit the dependence of \(g\),\(\Delta\) and \(u \) from certain functions.

Next, we give an explicit expression for the Fourier transform \(\widehat{g^{(v,w)}_k}(\rho)\) for arbitrary character \( \rho \in \btZ{k}\).

\begin{dfn} \label{dfn:varphi}
  Let \(\rho = (\rho_1,\rho_2,\cdots,\rho_k) \in \btZ{k} \). The function \(\vert \cdot \vert : \btZ{k} \to \C \) is
  given by
  \begin{equation*} 
    \vert \rho \vert = \| \rho \|_1 := \sum_{i=1}^{k} \rho_i.
  \end{equation*}  
  Let \(j_1 < j_2 < \cdots < j_{\vert \rho \vert} \) the position of the ones in \(\rho\), that is, \(\rho_{j_i} = 1\) for
  all \(i \in \{ 1,2,\cdots,\vert \rho \vert \}\) and if \(\rho_i =1\) then \(i \in \{j_1,j_2,\cdots,j_{\vert \rho \vert}\} \).
  The function  \( \varphi_k : \btZ{k} \to \C \) is given by
  \begin{equation} \label{eq:phi}
    \varphi_k(\rho)  := \sum_{i=1}^{\vert \rho \vert} (-1)^{i-1} j_{\vert \rho \vert + 1 - i}
    = j_{\vert \rho \vert} - j_{\vert \rho \vert-1} +  \cdots + (-1)^{\vert \rho \vert-1} j_1,
  \end{equation}
  and \( \varphi_k(\bZ) =0 \) where \(\bZ \) is the identity element in \(\btZ{k}\). For \(k=0\), define
  \(\varphi_k(\rho) = |\rho| = 0 \) where \(\rho\) is the unique element of \(\btZ{0}\).
\end{dfn}

Let \(\rho = (\rho_1,\rho_2,\cdots,\rho_{k}) \in \btZ{k} \) and \(\delta = (\rho_1,\rho_2,\cdots,\rho_{k-1}) \in \btZ{k-1}\). From the definition 
we obtain 
\begin{align}
  \label{eq:idenvarphi}
  \varphi_k(\rho) = (-1)^{\rho_{k}} \varphi_{k-1}(\delta)  + \rho_{k} k.
\end{align}
 
\begin{prop} \label{prop:fourier_g}
  For \( \rho \in \btZ{k}\), we have
  \begin{equation*}
    \widehat{g^{(v,w)}_k}(u,\rho) =  (-1)^{v \vert\rho\vert} \left( u^{2 \varphi_k(\rho) \Delta } + 
    (-1)^{v + w} u^{2(k+1 - \varphi_k(\rho))\Delta} \right)
  \end{equation*}
\end{prop}

\begin{proof}
  The identity is immediately verified for the cases \(k=0,1\). Next, suppose
  that  \(\rho = (\rho_1,\rho_2,\cdots,\rho_{k+1}) \in \btZ{k+1} \) and
  let \(\delta = (\rho_1,\rho_2,\ldots,\rho_{k} ) \in \btZ{k} \).  Then we have
  \begin{align*}
    \widehat{g^{(v,w)}_{k+1}}(u,\rho) &= \sum_{\bs \in \btZ{k+1}} g_{k+1}^{(v,w)}(u,\bs) \chi_{\rho} (\bs)= \sum_{i=0}^{1} \sum_{\substack{s \in \btZ{k+1} \\ s(k+1)= i}} g_{k+1}^{(v,w)}(u,\bs) \chi_{\rho} (\bs) \\
     &= \frac12 \sum_{i=0}^{1} (-1)^{\rho_{k+1} \cdot i} \left( 1 + (-1)^{w + i} u^{2\Delta} \right) \sum_{\bs \in \btZ{k}} g_{k}^{(v,i)}(u,\bs) \chi_{\delta} (\bs) \\
     &= \frac12 \sum_{i=0}^{1} (-1)^{\rho_{k+1} \cdot i} \left( 1 + (-1)^{w + i} u^{2\Delta} \right) \widehat{g^{(v,i)}_{k}}(u,\delta) \\
      &=\frac12 \sum_{i=0}^{1} (-1)^{\rho_{k+1} \cdot i} \left( 1 + (-1)^{w + i} u^{2\Delta} \right) (-1)^{v |\delta|} \left( u^{2 \varphi_k(\delta) \Delta}
                                       + (-1)^{v + i} u^{2(k+1 - \varphi_k(\delta))\Delta} \right),
  \end{align*}
  the equality in the last line holding by the induction hypothesis. The expression above is equal to
  \begin{align*}
    \frac12 (-1)^{v |\delta|} \Bigg[ & \left( 1+(-1)^{w} u^{2 \Delta} \right) \left(u^{2 \varphi_k(\delta) \Delta} +(-1)^v u^{2(k+1 - \varphi_k(\delta)) \Delta} \right)  \\
                                     & + (-1)^{\rho_{k+1}} \left( 1-(-1)^{w} u^{2 \Delta} \right) \left(u^{2 \varphi_k(\delta) \Delta} -(-1)^v u^{2(k+1 - \varphi_k(\delta)) \Delta} \right)  \Bigg],
  \end{align*}
  the result then follows by considering the cases \(\rho_{k+1} \in \{ 0,1\} \) by the identity \eqref{eq:idenvarphi}.
  and the fact that   \(|\rho| = |\delta| + \rho_{k+1}\).
\end{proof}

In the subsequent discussion of the heat kernel it is necessary to consider a generalization of the function
\(\varphi_k \). We motivate the definition via the Fourier transform of  \(\varphi_k \in \C[\btZ{k}]\).

\begin{prop} 
  Let \(\varphi_k : \btZ{k} \to \Z  \) be the function of Definition \ref{dfn:varphi}. We have
  \[
    \widehat{\varphi_k}(\rho) =
    \begin{cases}
      k 2^{k-1}  &\mbox{if } \rho = \bZ_k \\
      - 2^{k-1}  &\mbox{if } \rho = \bZ_i \oplus \bO_{k-i} \, (1 \leq i \leq k)  \\
      0  &\mbox{in any other case .} 
    \end{cases}
  \]
\end{prop}

\begin{proof}
  The case \(k=0 \) is trivial. For \(k \geq 1 \), let \( \rho = (\rho_1,\rho_2,\cdots,\rho_k) \in \btZ{k} \)
  and \( \delta = (\rho_1,\rho_2,\cdots,\rho_{k-1}) \in \btZ{k-1}\), then we have
  \begin{align*}
    \widehat{\varphi_k}(\rho) &= \sum_{\bs \in \btZ{k}} \varphi_k(\bs) (-1)^{(\bs|\rho)} = \sum_{\substack{\bs \in \btZ{k}\\ s_k = 0}} \varphi_k(\bs) (-1)^{(\bs|\delta)} +  \sum_{\substack{\bs \in \btZ{k}\\ s_k = 1}} \varphi_k(\bs) (-1)^{(\bs|\delta)}  \\
    &= \sum_{\br \in \btZ{k-1}} \varphi_{k-1}(\br) (-1)^{(\br|\delta)} + (-1)^{\rho_k} \sum_{\br \in \btZ{k-1}} (k -\varphi_{k-1}(\br)) (-1)^{(\br|\delta)} \\
  &= (1+ (-1)^{\rho_k + 1})\widehat{\varphi_{k-1}}(\delta) + (-1)^{\rho_k} k \sum_{\br \in \btZ{k-1}} (-1)^{(\br|\delta)},
  \end{align*}
  where the equality in the second line follows by \eqref{eq:idenvarphi}.
  Next, suppose that \(\rho_k= 0\), then
  \[
    \widehat{\varphi_k}(\rho) = k \sum_{\bs \in \btZ{k-1}} (-1)^{(\bs|\delta)} =
    \begin{cases}
      k 2^{k-1}  &\mbox{if } \delta = \bZ_{k-1} \\
      0  &\mbox{if } \delta \neq \bZ_{k-1}
    \end{cases}.
  \]
  On the other hand, if \(\rho_k= 1\) we have
  \[
    \widehat{\varphi_k}(\rho) = 2 \widehat{\varphi_{k-1}}(\delta) - k \sum_{\bs \in \btZ{k-1}} (-1)^{(\bs|\delta)} =
    \begin{cases}
      (k-1) 2^{k-1} - k 2^{k-1}  &\mbox{if } \delta = \bZ_{k-1} \\
      2 \widehat{\varphi_{k-1}}(\delta) &\mbox{if } \delta \neq \bZ_{k-1} 
    \end{cases},
  \]
  and the result follows by induction.
\end{proof}

By virtue of the proposition above, for \( \rho = (\rho_1,\rho_2,\cdots,\rho_k) \) we can write
\begin{align*} 
  \varphi_k(\rho) &=  \frac{k}{2} - \frac{1}{2}\left( \sum_{i=1}^{k} (-1)^{\sum_{j=i}^k \rho_j} \right).
\end{align*}

\begin{dfn} \label{dfn:varphi2}
  For \(k \geq 1 \)  and \(t \in \C \), the function \(\varphi_k(\rho;t): \btZ{k} \to \C\) is defined by
  \begin{equation*} 
    \varphi_k(\rho;t) := \frac{1}2 \sum_{i=1}^k\left( 1 - (-1)^{\sum_{j=i}^{k} \rho_j} \right) t^{i-1}.
  \end{equation*}  
\end{dfn}

In the following theorem we collect some properties and transformation formulas for \(\varphi_k(\rho;t)\). For an
integer \(i \in \Z_{\geq0} \) and \( t \in \C \), we write \([i]_t = \frac{1-t^i}{1-t}\). Notice that since
\[
  \lim_{t \to 1} \varphi_k(\rho,t) = \varphi_k(\rho),
\]
the identities of the following theorem also apply to \( \varphi_k(\rho)\).

\begin{thm} \label{thm:varphi} 
  Let \(\rho = (\rho_1,\rho_2,\cdots,\rho_k) \in \btZ{k} \), \(\check{\rho} = (\rho_k,\rho_{k-1},\cdots,\rho_1) \in \btZ{k} \) and \( v \in \{0,1\} \).
  Recall that for \(\br \in \btZ{k}, \bs \in \btZ{\ell} \), the vector \(\br \oplus \bs \in \btZ{k + \ell} \) denotes
  the concatenation of \(\br \) and \(\bs \). Then
  \begin{enumerate}
  \item \(\varphi_{k+1}(\rho \oplus (v);t) = v [k+1]_t + (-1)^v \varphi_k(\rho,t)\),
  \item \(\varphi_{k+1}((v) \oplus \rho;t) =   \varphi_k(\rho,t)t + \left( \frac{1-(-1)^{v + |\rho|}}{2} \right) \),
  \item \( \varphi_k(\check{\rho};t) = (-1)^{|\rho|} t^{k} \varphi_k(\rho;t^{-1}) + \left(\frac{1-(-1)^{|\rho|}}{2} \right) [k+1]_t\),
  \item \(\sum_{i=1}^{k} (-1)^{\sum_{j=i}^{k} \rho_j} t^{i-1}  = [k]_t- 2 \varphi_k(\rho;t)\).
  \end{enumerate}
\end{thm}

\begin{proof}
  The first claim is just the analog of \eqref{eq:idenvarphi}, the second follows immediately from the expression of \(\varphi_k\) in Definition \eqref{dfn:varphi2}. For the third one, we have
  \begin{align*}
    \varphi_k(\check{\rho};t) &= \frac{[k]_t}2  - \frac12 \left( \sum_{i=1}^k (-1)^{\sum_{j=1}^{k+1-i} \rho_j} t^{i-1} \right) =\frac{[k]_t}2  - \frac12 \left( \sum_{i=2}^k (-1)^{\sum_{j=i}^{k} \rho_j} t^{k+1-i} \right) (-1)^{|\rho|} - \frac12 (-1)^{|\rho|}\\
                 &=  \frac{(-1)^{|\rho|} t^k}2  [k]_{t^{-1}}  - \frac{ (-1)^{|\rho|} t^k}2 \left( \sum_{i=1}^k (-1)^{\sum_{j=i}^{k} \rho_j} t^{-i+1}  \right)  + \frac{[k]_t}{2} - \frac{(-1)^{|\rho|} t^k}2  [k]_{t^{-1}} + \frac{t^k}2  - \frac12 (-1)^{|\rho|} \\
                     &= (-1)^{|\rho|} t^k \varphi_k(\rho; t^{-1})  + [k+1]_{t} \left(\frac{1 - (-1)^{|\rho|}}2 \right)
  \end{align*}
  as desired. The last claim is obtained directly from the definition.
\end{proof}

In addition, it is not difficult to see from the formulas in Theorem \ref{thm:varphi} that if \(0< j_1 < j_2 < \cdots < j_{|\rho|} \leq k\) are the position of the ones in
\( \rho  \in \btZ{k}\), we have
\begin{equation}
  \label{eq:varphit}
  \varphi_k(\rho;t) = \sum_{i=1}^{\vert \rho \vert} (-1)^{i-1} [j_{\vert \rho \vert + 1 - i}]_t,
\end{equation}
so that \(\varphi_k(\rho;t) \)  is seen to be a $t$-analogue of the function \(\varphi_k(\rho) \) of  Definition \ref{dfn:varphi}.

To close the discussion of the function $\varphi_n$, let us describe with more detail the relation between the two
representations of the function $\varphi_n$. The main point is the underlying bijection
\begin{equation}
  \label{eq:bij1}
  \mathcal{S}_k^{(n)}:= \left\{ \rho \in \Z_2^k \, : \, |\rho|= n \right\} \longleftrightarrow \left\{ (j_1,j_2,\cdots,j_n) \in \Z_{\geq1}^{n} \,;\, j_1 <j_2<\cdots<j_n \leq k \right\} =: \mathcal{J}_n^{(k)}                                                            
\end{equation}
given by the position of the ones in $\rho \in \Z_2^k$ for $|\rho|=n \in \Z_{\geq 0}$. For \(n \geq 1 \), we define the function
$\phi^{(n)} :  \C^{n+1} \to \C$ by
\begin{equation}
  \label{eq:phidef}
  \phi^{(n)}(\bm{x},t) := \sum_{i=1}^{n} (-1)^{i-1} [x_{n + 1 - i}]_t,
\end{equation}
and for \(n=0\) we set $\phi^{(0)}= 0 $. Then, for \(\rho \in \mathcal{S}_k^{(n)}\) corresponding to $\bm{j}=(j_1,j_2,\cdots,j_n) \in \mathcal{J}^{(k)}_n$, we have
\begin{equation}
  \label{eq:phiPhirel}
  \varphi_k(\rho ;t) = \phi^{(n)}(\bm{j},t).
\end{equation}
We remark that, as a function on the variables $x_1,x_2,\cdots,x_n$, the right hand side of the equality does not
depend on $k$. This is the key property that we use in the sequel to evaluate the sums appearing in the heat kernel.

To get a better understanding of equation \eqref{eq:phiPhirel}, we introduce the inductive limit
\[
  \Z_2^{\infty} =  \varinjlim_{n} \Z_2^{n},
\]
where, for $ i \leq j$, the injective homomorphisms $f_{i j} : \Z_2^i \to \Z_2^j$ are given by
\[
  f_{i j}(\rho) = (\rho_1,\rho_2,\cdots,\rho_i,0,\cdots,0) \in \Z_2^{j}
\]
for $\rho = (\rho_1,\rho_2,\cdots,\rho_i) \in \Z_{2}^{i}$. 
Clearly, the functions \(\varphi_{k}\) for \(k \geq 1\) induce naturally a function $\varphi: \Z_2^{\infty}\times \C \to \C$.

\begin{lem}[Universality] \label{lem:bij1}
  Let \( n \in \Z_{\geq 0}\). There is a bijection
  \[
    \mathcal{S}^{(n)}:= \left\{ \rho \in \Z_2^{\infty} \, : \, |\rho|= n \right\}  \longleftrightarrow  \left\{ j_1,j_2,\cdots,j_n \in \Z_{\geq1} \,;\, j_1 <j_2<\cdots<j_n  \right\} =: \mathcal{J}_n,
  \]
  Let \( \rho \in \mathcal{S}^{(n)}\),  corresponding to \(\bm{j} \in \mathcal{J}_n \), then we have
  \[
    \varphi(\rho ; t) = \phi^{(n)}(\bm{j},t). \qed
  \]
\end{lem}

The lemma above means, in practice, that while the function \(\varphi\) (or any of the individual functions $\varphi_k$ for $k\geq 0$) is,
in general, a complicated function, when restricted to elements of fixed norm $|\rho| =n$, it has a simple
representation given by $\phi^{(n)}$, that is, it is essentially a $q$-polynomial in the variable $t$. 

\begin{rem}
  The function \( \varphi_k(\rho;t) \) admits the following characterization. Denote by \( E(\mathbf{x};t) \)  the generating
  function for the elementary symmetric functions (see e.g. \cite{Macdonald1991})
  \[
    E(\mathbf{x};t) = \prod_{i=1}^{\infty} (1 + x_i t).
  \]
  Let \(F(\mathbf{x};t)\) be a (formal) function defined in infinite vectors
  \( \mathbf{x} = (x_1,x_2,x_3,\cdots) \) given by
  \[
    F( \mathbf{x}; t ) := E( \mathbf{x};-2) \sum_{i=1}^{\infty}  \frac{[i]_t x_i}{\prod_{j=1}^{i}(1-2 x_j)},
  \]
  then we have the equality
  \[
    \varphi_k(\rho;t)  = F((\rho_1,\rho_2,\cdots,\rho_k,0,0,0,\cdots);t).
  \]
  Indeed, by successive application of the first transformation formula, we obtain
  \begin{equation}
    \label{eq:remvarphi}
    \varphi_k(\rho;t) = \sum_{i=1}^k [i]_t \rho_i \prod_{j=i+1}^k (1- 2 \rho_j)
  \end{equation}
  since \(1-2\rho_j = (-1)^{\rho_j}\). 
\end{rem}

  \begin{rem}
    For \(k\geq 1 \), the function \(\varphi_k(\rho,t)\), with a small modification, may be interpreted as a morphism of
    abelian groups. To see this, we notice that by \eqref{eq:remvarphi} we have
    \begin{equation}
      \label{eq:varphihom}
      \varphi_k(\rho + \theta; t) = \varphi_k(\rho;t) + \varphi_k(\theta;t) \pmod{2},
    \end{equation}
    for  \(\rho, \theta \in \btZ{k}\). Next, by using equation \eqref{eq:varphit} as the definition of \(\varphi_k(\rho;t)\) we
    can consider \(\Z_2[t]_k\), the vector space of polynomials of degree less than \(k\)
    over the ring \(\Z_2\), as the codomain of \(\varphi_k(\rho;t)\), that is, \(\varphi_k(\cdot,t) : \btZ{k} \to \Z_2[t]_k\).
    Thus, the identity \eqref{eq:varphihom} exhibits \(\varphi_k(\rho;t) \) as an isomorphism of abelian groups and by linear extension, an isomorphism
    of vector spaces over \(\Z_2\).
  \end{rem}
  
\subsection{Fourier transform  of $R_{\mu}^{(k,N)}$} 
\label{sec:four-transf-R}

In this section we describe the Fourier transform of the function \(R_{\mu}^{(k,N)}(u,\bs)\). For convenience, we recall
the definition
\begin{align*}
  R^{(k,N)}_{\mu}(u,\bs) &= \exp \left(  \frac{\sqrt{2} g(1-u)}{1-u^{2 N}}  \sum_{j=1}^{k-1} (-1)^{\bs(j)} \left( x \Lambda^{(j)}(u) + y \Lambda^{(N-j+1)}(u)  \right)    \right)   \nonumber \\
       &\qquad \times  \exp \Bigg( \frac{g^2 (1-u)^2}{2 (1+u)^2 (1-u^{2 N})} \bigg(  \sum_{i=1}^{k-2} \eta_i(s)^2 \Omega^{(i,i)}(u)  + 2 \sum_{i=1}^{k-2} \sum_{j=i+1}^{k-2} \eta_i(s)\eta_j(s) \Omega^{(i,j)}(u)  \\
  & \qquad  \qquad \qquad  \qquad \qquad  \qquad \qquad \qquad + 4 (-1)^{\mu} \sum_{i=1}^{k-2} \sum_{j=k}^{N-1} \eta_i(s) \Omega^{(i,j)}(u)  \bigg) \Bigg),
\end{align*}
from where is it clear that  \(R^{(k,N)}_{\mu} \in \C[\btZ{k-1}]\). As in the case of the function \(g_k(u,\bs)\), the Fourier
transform is computed in the abelian group \( \setC{k-1}{v}{w} \simeq \btZ{k-3} \), with \(v,w \in \{0,1\}\), and we denote by
\(R^{(v,w)}_{\mu} \in \C[\btZ{k-3}]\) the function resulting by applying the projection \eqref{eq:projectionVect} to
\(R^{(k,N)}_{\mu}(u,\bs)\).
We note that \(R^{(v,w,k,N)}_{\mu}(u,\bs)\) would be a more appropriate notation for \(R^{(v,w)}_{\mu}(u,\bs)\), but since \(k ,N \in \Z_{\geq 1} \)
remain fixed in the computations of this section and there is no risk of confusion we have dropped the dependence of \(k,N\) from the
notation of \(R^{(v,w)}_{\mu}(u,\bs)\).

We start with some general considerations. First,  suppose \(S\) is subset of characters \( S \subset \btZdual{k-3}\) and \( f \in \C[\btZ{k-3}] \) is given by
\begin{equation*} 
  f(\bs) := \exp\left( \sum_{\chi \in S} a_\chi \chi(\bs)  \right) = \sum_{\xi \in \btZdual{k-3}} C_\xi \xi(\bs),
\end{equation*}
for arbitrary \(a_{\chi} \in \C \) with \(\chi \in S \), and where \(C_\xi \in \C\) is the Fourier coefficient corresponding
to \(\xi \in \btZdual{k-3} \). The Fourier transform \(\widehat{f} \) is then given by
\[
  \widehat{f}(\rho) =  2^{k-3 }\sum_{\xi \in \btZdual{k-3}} C_\xi \delta_{\xi,\chi_{\rho}}.
\]

Therefore, in order to get the expression for the Fourier transform of \(f(\bs) \), it is enough to describe
the Fourier coefficients \(C_\xi \in \C\) in terms of \(a_{\chi} \in \C \). Let us consider the
case \(|S| =1 \), that is, \(S = \{ \chi \} \). In this case
\[
  f(s) = \cosh(a_\chi)  + \sinh(a_\chi)\chi(\bs),
\]
since any character \(\chi  \in \btZdual{k-3}\) is real.

To describe the general case, we introduce an ordering in \(S = \{\chi_1, \chi_2, \cdots,\chi_\ell \} \) with \(\ell = |S|\).
Then, for \(\mathbf{a} \in \C^{\ell} \) and an index vector \(\br \in \{0,1\}^\ell \) we define
\begin{align*} 
  T^{(\br)}(\mathbf{a}) &:= \prod_{i=1}^{\ell} \left[ \cosh(\mathbf{a}_i)^{(1-\br_i)}\sinh(\mathbf{a}_i)^{\br_i} \right], 
\end{align*}
where \(\mathbf{a}_i \) (resp. \(\br_i\)) denotes the \(i\)-th component of \( \mathbf{a}\) (resp. \(\br \)). 

The Fourier coefficients of \(f\) are given by
\[
  C_{\chi_\rho } \left(= C_{\rho } \right) = \sum_{\substack{\br \in \{0,1\}^{\ell} \\ \chi_\rho = \prod_{i=1}^{\ell} (\chi_i)^{\br_i}}} T^{(\br)}(\mathbf{a}),
\]
where \(\mathbf{a} = \{a_1,a_2,\cdots,a_\ell \} \in \C^{\ell} \) is the vector of coefficients. 
In particular, note that \(C_{\chi} \ne 0\) if and only if \(\chi\) is generated by elements in the set \(S\).


Next, we specialise these considerations for the case of the function \(R^{(v,w)}_{\mu}(\bs) \in \C[\btZ{k-3}]\).
In this case, the set \(S_{k-3}\) (corresponding to the set $S$ in the discussion above) is given by
\[
  S_{k-3} = \left\{ \chi = \chi_\rho \in \btZdual{k-3} \, | \, \rho \in \btZ{k-3} \, , \, 0 < |\rho| \leq 2 \right\}.
\]
 In particular \(|S_{k-3}| = \frac{(k-3)(k-2)}2\), and if \(\chi_\rho \in S_{k-3} \), we have
\[
  \rho \in \{ \bm{e}_i + \bm{e}_j \, | \, 0 \leq i < j \leq k-3 \}
\]
where  \(\bm{e}_0 := \bZ\) is the zero vector. For \(\rho = \bm{e}_i + \bm{e}_j\)
we denote  \(\chi_\rho \in S_{k-3} \) by \(\chi_{i,j} \). Similarly,  we denote by \(a_{i,j}\) (resp. \(r_{i,j}\)) the entries of the
coefficient vector \(\mathbf{a}\) (resp. the vector \(\br \in \{0,1\}^{\ell} \)) in lexicographical ordering.

Note that the trivial character is omitted from the set, since
\[
  R^{(v,w)}_{\mu}(u,\bs) = \exp\left( a_0^{(\mu)} + \sum_{\chi \in S_{k-3}} a_\chi^{(\mu)} \chi(\bs) \right) = \exp( a_0^{(\mu)})  \exp\left( \sum_{\chi \in S_{k-3}} a_\chi^{(\mu)} \chi(\bs) \right) ,
\]
thus
\begin{equation*} 
  \widehat{R^{(v,w)}_{\mu}}(u,\bs) =  2^{k-3 }\exp(a_0^{(\mu)} ) \left(\sum_{\xi \in \btZdual{k-3}} C_\xi^{(\mu)} \delta_{\xi,\chi_{\rho}}\right).
\end{equation*}
We note here that in the case \(k=3\), \(R^{(v,w)}_{\mu}(u,\bs)=\widehat{R^{(v,w)}_{\mu}}(u,\rho) = \exp\left( a_0^{(\mu)} \right) \)
for \( \rho \in \btZ{0} \).

The next lemma describes the coefficients \(a_{i,j} \) for the case of the function \(R^{(v,w)}_{\mu}(u,\bs)\).
The proof is by direct computation from the definitions and we omit it.

\begin{lem}
  \label{lem:fourier4}
  The trivial character coefficient \(a_0^{(\mu)}\) is given by
  \begin{align*}
    a_0(u^{\frac{1}N}) &= O(\tfrac{1}{N}).
  \end{align*}
  For \(1 \leq i \leq k-3 \), the Fourier coefficient \(a^{(\mu)}_{0,i} \) is given by
  \begin{align*}
    a_{0,i}^{(\mu)}(u^{\frac{1}N}) &= \frac{\sqrt{2}g(1-u^{\frac{1}N})}{1-u^2}
                   \left( x u^{\frac{i}N} (1-u^{2-\frac{2i+1}{N}}) + y u^{1-\frac{i+1}{N}}(1-u^{\frac{2 i+1}N}) \right) \\
         &+ (-1)^{\mu} \frac{2g^2 (1-u^{\frac1N})}{(1-u^2)(1+u^{\frac{1}{N}})}
              u^{\frac{k}N - \frac{i+1}{N}}(1-u^{1-\frac{k}{N}})(1-u^{1-\frac{k-1}{N}})(1-u^{\frac{2i+1}{N}}) + O(\tfrac{1}{N^2}).
  \end{align*}
  For \(1\leq i < j \), the Fourier  coefficients \(a_{i,j} \) are given by
  \[
    a_{i,j}(u^{\frac1N}) = \frac{g^2 (1-u^{\frac{1}{N}})^2 }{1-u^{2}} u^{\frac{j}N - \frac{i+1}N}(1-u^{\frac{2i +1}N})(1-u^{2 - \frac{2j+1}{N}}). \qed
  \] 
\end{lem}

\subsection{Graph theoretical considerations} 
\label{sec:graph-theoretical} 

For the case of Fourier transform $\widehat{R_{\mu}^{(k,N)}}$ of $R_{\mu}^{(k,N)}$, we have seen in \S\ref{sec:four-transf-R} that the elements $\chi_\rho \in S_{k-3}$ correspond to $\rho = \bm{e}_i + \bm{e}_j$ with $0 \leq i < j \leq k-3$. In addition, note that for \(\rho \in \btZ{k-3}\), defining the set
\[
  V^{(k-3)}_{\rho} = \left\{\br \in \{0,1\}^{|S_{k-3}|} \, \Big| \, \chi_\rho = \prod_{i=1}^{k-3} (\chi_{0,i})^{r_{0,i}} \prod_{1\leq i<j} (\chi_{i,j})^{r_{i,j}} \right\},
\]
the Fourier coefficients \(C_{\chi_\rho }\) are given by
\[
  C_{\chi_\rho } = \sum_{ \br \in V^{(k-3)}_{\rho} } T^{(\br)}(\mathbf{a}),
\]
where, for simplicity, and only in this subsection we drop the dependency of $\mu$ from the notation of $a_{i,j}^{(\mu)}$.
The structure of the sets \(V^{(k-3)}_{\rho}\) allows a graph-theoretical (combinatorial) description as we see below in Definition \ref{def:graph}. Using this description, in this subsection we prove several properties of the sets \(V^{(k-3)}_{\rho}\) used in the Section \ref{sec:transformation-sum} below. In particular, by the reduction procedure described in Proposition \ref{prop:V0_1}
we see that for our purposes it is enough to consider the case of \(\rho = \bZ\) (see Example \ref{example:V_0} for case of \(V^{(3)}_{\bZ}\)).  

Notice that, by the definition of \(S_{k-3}\), it is clear that for any \(\rho \in \btZ{k-3} \), the set
\(V^{(k-3)}_{\rho} \) is not empty. In fact, we see in Lemma \ref{lem:v0} that the set \(V_\rho^{(k-3)} \) has the same cardinality as \(V^{(k-3)}_{\bZ}\). 

Next, we give an alternative description of the elements of the set \(V_\rho^{(k-3)}\) as simple undirected graph allowing loops.

\begin{dfn} \label{def:graph}
  For \(\br \in \{0,1 \}^{|S_{k-3}|} = \btZ{\frac{1}2(k-2)(k-3)} \), the graph \(\mathcal{G}(\br)\) is the
  undirected graph with the vertex set
  \[
    V(\mathcal{G}(\br)) = \{1,\ldots, k-3 \},
  \]
  and edges determined by
    \begin{equation*}
    \begin{cases}
      (i,i) \in E(\mathcal{G}(\br)) & \text{ if } r_{0,i} = 1, \text{ for } 0<i \\
      (i,j) \in E(\mathcal{G}(\br)) & \text{ if } r_{i,j} = 1, \text{ for } 0<i<j,  
    \end{cases}
  \end{equation*}
  where \(E(\mathcal{G}(\br))\) is the edge set of \(\mathcal{G}(\br)\).

  We denote by \(\deg(\mathcal{G}({\br}))\) the (ordered) list of degree of the vertices of
  \(\mathcal{G}({\br})\).
\end{dfn}

Note that different to usual convention,  when the graph \(\mathcal{G}(\br)\) has a
loop \((i,i) \in E(\mathcal{G}(\br))\) we consider the loop to contribute \(1\) to the degree
of the vertex \(i\). 

\begin{ex}
  Let \(k=7\). For \(\br = (0, 0, 1, 1, 1, 1, 1, 0, 0, 1) \in \btZ{10}\) the graph \(\mathcal{G}({\br})\) is shown
  in Figure \ref{fig:graph1}.
  \begin{figure}[!ht]
    \centering
    \includegraphics[height=3cm]{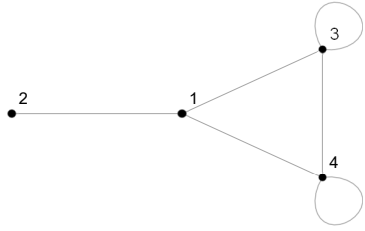}
    \caption{Graph \(\mathcal{G}({\br})\)  associated to the vector \(\br = (0, 0, 1, 1, 1, 1, 1, 0, 0, 1)\)}
    \label{fig:graph1}
  \end{figure}
  Actually, we easily verify that \(\br \in V^{(4)}_{\rho}\) with \(\rho = (1,1,1,1)\).
  Notice also that \(\deg(\mathcal{G}({\br})) = ((3,1,3,3) \equiv (1,1,1,1) \pmod{2}\).
\end{ex}

In fact, the last property of the example determines the set \(V^{(k-3)}_{\rho}\), as we can easily verify and
state in the following lemma.

\begin{lem}
  For \(k\ge 3\), we have
  \begin{equation} \label{eq:setVgrahp}
  V^{(k-3)}_{\rho} = \{\br \in \{0,1\}^{|S_{k-3}|} \,| \, \deg(\mathcal{G}({\br})) \equiv \rho \pmod{2} \}. \qed
\end{equation}
\end{lem}

\begin{prop}
  \label{prop:V0_1}
  For \(\rho \in \btZ{k-3} \), we have
  \[
    |V^{(k-3)}_\rho| = |V^{(k-3)}_{\bZ}| = 2^{(k-3)(k-4)/2},
  \]
  and the bijection \(\sigma_\rho :  V^{(k-3)}_{\bZ} \to V^{(k-3)}_{\rho} \) is given explicitly by the map 
  \[
    \br \in V^{(k-3)}_{\bZ}  \mapsto  \br + (\rho \oplus \bZ_{(k-3)(k-4)/2}) \pmod{2}  \in V^{(k-3)}_{\rho}.
  \]
  Furthermore, the map \(\sigma_{\rho}\) induces the relation
  \[
    T^{(\sigma_\rho(\br))}(\mathbf{a}) = \prod_{i=1}^{k-3}\left(\tanh(a_i)^{1-2 r_{0,i}}\right)^{\rho_i} T^{(\br)}(\mathbf{a}).
  \]
\end{prop}

\begin{proof}
  From \eqref{eq:setVgrahp}, we see that \( |V^{(k-3)}_{\bZ}|\)  is equal to the number of even graphs with \(k-3\)
  vertices. By \S 1.4 of \cite{Haray1973}, the number of such graphs is equal to \(2^{(k-3)(k-4)/2}\).

  Next, let us consider the effect of the map \( \sigma_\rho: V^{(k-3)}_{\bZ} \to \{0,1 \}^{|S_{k-3}|}\) on the associated graphs
  \(\mathcal{G}({\br}) \), in particular on the degree of a given vertex \(i \in \{1,2,3,\cdots,k-3\}\).
  First, it is clear that any edge \((i,j) \), with \(i \neq j\), in \(\mathcal{G}({\br}) \) is invariant under
  \(\sigma_\rho\), that is, if \((i,j) \) is an edge of \(\mathcal{G}(\br) \) then it is also an edge
    of \(\mathcal{G}(\sigma_{\rho}(\br)) \). Now, suppose that \(\rho_i = 1\) and the vertex \(i\) does not have a loop
  in \(\mathcal{G}({\br}) \) (i.e. \(\br_{0 i} =0\)), then the vertex \(i\) has a loop in
  \(\mathcal{G}({\sigma_{\rho}(\br)}) \). On the other hand, if the vertex \(i\) has a loop in \( \mathcal{G}({\br}) \),
  then \(i\) does not have a loop in \(\mathcal{G}({\sigma_{\rho}(\br)}) \).
  Thus the degree of \(i\) in \(\mathcal{G}({\sigma_{\rho}(\br)}) \) is \(\pm 1\) the degree of \(i\) in
  \( \mathcal{G}({\br}) \) (see Figure \ref{fig:graph2} for an example).

  \begin{figure}[!ht]
    \centering
    \begin{subfigure}{0.4\textwidth}
      \includegraphics[height=2cm]{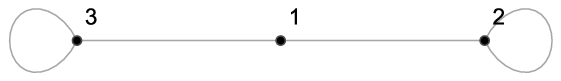}
      \caption{\(\br = (0,1,1,1,1,0)\)}
    \end{subfigure}
    ~
    \begin{subfigure}{0.4\textwidth}
      \includegraphics[height=2.2cm]{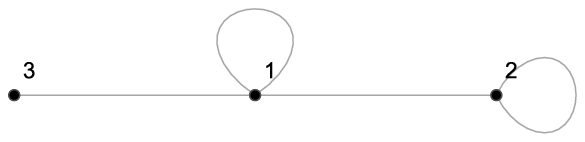}
      \caption{\( \sigma_{\rho}(\br) = (1,1,0,1,1,0)\)}
    \end{subfigure}
    \caption{Graphs corresponding to a vector \(\br\) and its image under \(\sigma_{\rho}\) for \(\rho = (1,0,1) \) }
    \label{fig:graph2}
  \end{figure}

  If \(\rho_i = 0\), there is no change in the degree of the vertex \(i\). Consequently, we have
  \[
    \deg(\mathcal{G}({\sigma_{\rho}(\br)})) \equiv  \deg(\mathcal{G}({\br})) + \rho \equiv \rho \pmod{2},
  \]
  and thus \(\sigma_\rho(V^{(k-3)}_{\bZ}) \subset V^{(k-3)}_{\rho}\). It is clear that the map \(\sigma_{\rho} \) is an involution, whence
  the second claim is proved.
  The third claim follows directly by the definition of \(T^{(\br)}(\mathbf{a}) \).
\end{proof}

\begin{ex}
  Suppose \(k=6\), thus \(|S_{3}|=6 \). Let \(\br = (0,1,1,1,1,0) \in V^{(3)}_{\bZ}\) and \(\rho = (1,0,1 ) \).
  Then, letting \(C = \sinh(a_4)\sinh(a_5) \cosh(a_6)  \), we have
  \begin{align*}
    T^{(\sigma_\rho(\br))}(\mathbf{a}) &= \sinh(a_1)\sinh(a_2)\cosh(a_3) C \\
                                      &= \left(\tanh(a_1) \coth(a_3)\right) \cosh(a_1)\sinh(a_2)\sinh(a_3) C \\
                                            &=\left(\tanh(a_1) \coth(a_3)\right) T^{(\br)}(\mathbf{a}).
  \end{align*}
\end{ex}

By the foregoing discussion,  the Fourier coefficients  are given by
\begin{equation}
  \label{eq:fouriercoeff7}
  C_{\rho }  = \sum_{\br \in V_{\bZ}^{(k-3)}} T^{(\sigma_{\rho}(\br))}(\mathbf{a}).
\end{equation}

Next, we define several projections on set \(V_{\bZ}^{(k-3)}\). 

\begin{dfn}
  Let \(k \geq 4 \), then
  \begin{itemize}
  \item \(p_1: \btZ{\frac{1}2(k-2)(k-3)} \to \btZ{k-3} \) is the projection of the first \(k-3\) components,
  \item \(p_2: \btZ{\frac{1}2(k-2)(k-3)} \to \btZ{\frac{1}2(k-3)(k-4)} \) is the projection of the last \((k-3)(k-4)/2 \) components.
  \end{itemize}
  Let \(k \geq 5 \), then:
  \begin{itemize}
  \item \(q_1 : \btZ{\frac{1}2(k-2)(k-3)} \to \btZ{k-4}  \) is the projection of the \(k-4\) components starting from \(k-2 \),
  \item \(q_2 : \btZ{\frac{1}2(k-2)(k-3)} \to \btZ{\frac{1}2(k-4)(k-5)}\) is the projection of the last \( (k-4)(k-5)/2 \) components.
  \end{itemize}
\end{dfn}

Note that if the appropriate domain of the functions are considered, the relations 
\[
  q_1 = p_1 \circ p_2,  \qquad q_2 = p_2 \circ p_2.
\]
hold. 

\begin{ex}
  Let \(k=7\) and \(\br = (0, 0, 1, 1, 1, 1, 1, 0, 0, 1) \in \br \in V^{(4)}_{\rho}\), then
  \[
    p_1(\br) = (0,0,1,1), \qquad p_2(\br) = (1,1,1,0,0,1),
  \]
  and \(\br = p_1(\br) \oplus p_2(\br)\). Also.
  \[
    q_1(\br) = (1,1,1), \qquad q_2(\br) = (0,0,1).
  \]
\end{ex}

The next results describes the structure of the set \(V_{\bZ}^{(k-3)}\) used in \S \ref{sec:transformation-sum} to
evaluate the sums over the Fourier transforms of the functions \(R_{\mu}^{(k,N)} \) and \(g_k(u,\bs)\) (cf. Lemma \ref{lem:sumv0}).

\begin{lem}
  \label{lem:v0}
  Let \(\br ,\bs \in V^{(k-3)}_{\bZ}\). 
  \begin{enumerate}
  \item If \(\br \neq \bs \), then \(p_2(\br) \neq p_2(\bs) \). In other
    words, \(p_2 \) is a bijection of \(V^{(k-3)}_{\bZ}\) onto \(\btZ{(k-3)(k-4)/2}\).
  \item If \(\br \in V^{(k-3)}_{\bZ}\), then \(|p_1(\br)| \equiv 0 \pmod{2}\). Moreover,
    \[
      p_1(V^{(k-3)}_{\bZ}) = \{ \rho \in \btZ{k-3}\, | \, |\rho| \equiv 0 \pmod{2}  \}.
    \]
  \item We have
    \[
      V_\bZ^{(k-4)} = p_2 ( \{ v \in V_0^{(k-3)} \,:\, p_1(v) = \bZ_{k-3} \} ),
    \]
  \item For \(\bv \in \btZ{(k-4)(k-5)/2}\),
    \[
      p_1( \{ \sigma \in V_\bZ^{(k-3)} \, : \, q_2(\sigma) = \bv \}) = \{ \sigma \in \btZ{k-3} \, : \, |\sigma| \equiv 0 \pmod{2} \},
    \]
    and 
    \[
      q_1( \{ \sigma \in V_\bZ^{(k-3)} \, : \, q_2(\sigma) = \bv \}) = \btZ{k-4}.
    \]
    Moreover, the restriction of \(p_1\) and \(q_1\) to the above sets are bijections.
  \item For \(\bv \in \btZ{(k-4)(k-5)/2}\), let  \(\br \in V_{\bZ}^{(k-3)}\) such that \(q_2(\br) = \bv \). Let \( \mathbf{r_0} \in V_{\bZ}^{(k-3)} \) be the
    unique element such that \(p_1(\br_0) = \bZ_{k-3} \) and \(q_2(\br_0) = \bv\). Then,
    \[
      q_1(\br) =  q_1(\br_0) + (r_{0,2},r_{0,3}, \cdots, r_{0,k-3}) \pmod{2}.
    \]
  \end{enumerate}
\end{lem}

\begin{ex}\label{example:V_0}
  We illustrate the statements of Lemma \ref{lem:v0} with an example. For \(k=6\), the set \(V_\bZ^{(k-3)} \) is given by
  \begin{align*}
    V_\bZ^{(3)} = \{&
                    (0,0,0,0,0,0),(0,0,0,1,1,1),(0,1,1,0,0,1),(0,1,1,1,1,0), \\
                  &(1,0,1,0,1,0),(1,0,1,1,0,1),(1,1,0,0,1,1),(1,1,0,1,0,0) \}.
  \end{align*}
  Then, we see directly that 
  \[
    p_2(V_\bZ^{(3)}) = \btZ{3},
  \]
  and
  \[
    p_1(V_\bZ^{(3)}) = \{(0,0,0),(0,1,1),(1,0,1),(1,1,0)\}.
  \]
  Moreover,
  \[
    p_2 ( \{ v \in V_\bZ^{(3)} \,:\, p_1(v) = \bZ_{3} \} ) = \{(0,0,0),(1,1,1)\} = V_\bZ^{(2)},
  \]
  and if \(\bv = (0)\)
  \[
    p_1( \{ \sigma \in V_\bZ^{(3)} \, : \, q_2(\sigma) = \bv \}) = \{(0,0,0),(0,1,1),(1,0,1),(1,1,0)\},
  \]
  and
  \[
    q_1( \{ \sigma \in V_\bZ^{(3)} \, : \, q_2(\sigma)) = \bv \}) = \{(0,0),(1,1),(0,1),(1,0)\} = \btZ{2}.
  \]
  Finally, with the notation of (4) in the lemma, let \(\bv = (1)\) and \(\br = (1,0,1,1,0,1)\). Then
  \( \br_0 = (0,0,0,1,1,1)\), \(q_1(\br_0) = (1,1) \),\((r_{0,2},r_{0,3})= (0,1)\) and
  \[
    q_1(\br) = (1,0) \equiv (1,1) + (0,1) \pmod{2}.
  \]
\end{ex}

\begin{proof}[Proof of Lemma \ref{lem:v0}]
 
  In this proof we use repeatedly the well-known (elementary) fact from graph theory that the number of vertices in an undirected simple
  graph with odd degree is even.
  
  Suppose \(\bv = (v_1,v_2,\cdots,v_{\ell}) \in  \btZ{\ell} \), where \( \ell = \frac{(k-3)(k-4)}{2} \) and
  \(\br = (0,0,\cdots,0,v_1,v_2,\cdots,v_\ell)  \in \btZ{(k-3)(k-2)/2} \). The associated graph
  \(\mathcal{G}({\br})\) is a undirected simple graph on \(k-3\) vertices.
  Then, there is a unique element \(\bar{\br} \in V^{(k-3)}_{\bZ}\) such that \(p_2(\bar{\br}) = \bv \), that is,
  the one corresponding to the graph obtained by adding loops to the vertices of \(\mathcal{G}({\br})\) with
  odd degree.
  The correspondence establishes an injection of \(\btZ{\ell}\) into \(V^{(k-3)}_{\bZ}\), which is
  actually seen to be a bijection by comparing the cardinality of the sets (cf. Proposition \ref{prop:V0_1}),
  proving (1). 
  Moreover, this argument also shows that for \(\br \in V^{(k-3)}_{\bZ}\) we have \(|p_1(\br)| \equiv 0 \pmod{2}\).

  Conversely,  let \( \bv \in \btZ{k-3}\) with \(|\bv| \equiv 0 \pmod{2}\). Set \(\br = (v_1,v_2,\cdots,v_{k-3},0,0,\cdots,0) \in \btZ{(k-3)(k-2)/2} \),
  then the graph \(\mathcal{G}(\br)\) is a graph with exactly an even number of loops. Moreover, the graph \(\mathcal{G}_1\) with all
  vertices of even degree obtained by joining pair of vertices with loops with exactly one edge corresponds to a vector
  \(\br_1 \in \btZ{(k-3)(k-2)/2}\) such that \(p_1(\br_1) = \bv \), proving (2).

  Let \(\bv \in V_{\bZ}^{(k-4)}\), and \(\bar{\bv} = \bZ_{k-3} \oplus \bv \), we are to prove that \(\bar{\bv} \in V_{\bZ}^{(k-3)}\), in
  other words, that
  \[
     \chi_\bZ = \prod_{i=1}^{k-3} (\chi_{0,i})^{\bar{v}_{0,i}} \prod_{i<j} (\chi_{i,j})^{\bar{v}_{i,j}},
   \]
   equivalently,
   \[
     \bZ_{k-3} = \sum_{i=1}^{k-3} \bar{v}_{0,i} \, \bm{e}_i + \sum_{i<j} \bar{v}_{i,j} \, (\bm{e}_i + \bm{e}_j ).
   \]
   Notice that by the definitions, we have
   \[
     \bar{v}_{0,i} = 0, \qquad \bar{v}_{1,j} = v_{0,j-1}, \qquad \bar{v}_{n,m} = v_{n-1,m-1},
   \]
   where \(i \in \{ 1,2,\cdots, k-3\}\), \(j \in \{2,3,\cdots,k-3\}\), and \(2\leq n< m \leq k-3 \). Thus
   \begin{align*}
     \sum_{i=1}^{k-3} \bar{v}_{0,i} \, \bm{e}_i + \sum_{i<j} \bar{v}_{i,j} \, (\bm{e}_i + \bm{e}_j )
     &= \sum_{j=2}^{k-3} \bar{v}_{1,j} (\bm{e}_1 + \bm{e}_j) + \sum_{2\leq i < j} \bar{v}_{i,j} \, (\bm{e}_i + \bm{e}_j ) \\
     &= |p_1(\bv)| \bm{e}_1 +  \sum_{i=1}^{k-4} v_{0,i} \bm{e}_{i+1} + \sum_{i< j} v_{i,j} (\bm{e}_{i+1} + \bm{e}_{j+1} ) 
   \end{align*}
   and this is equal to \(\bZ_{k-3}\) since \(\bv \in V_{\bZ}^{(k-4)}  \)  (and  \(|p_1(\bv)| \equiv 0 \pmod{2} \) by (2) ).
   The converse follows in the same way, proving (3).
   
   Next, let  \(\bv \in \btZ{(k-4)(k-5)/2} \). By (2), we have
   \[
     p_1(\{ \sigma \in V_0^{(k-3)} \, : \, q_2(\sigma) = \bv \}) \subset \{ \sigma \in \btZ{k-3} \, : \, |\sigma| \equiv 0 \pmod{2} \}.
   \]
   and by (1), we have
   \[
     |\{ \sigma \in V_0^{(k-3)} \, : \, q_2(\sigma) = \bv \}| = 2^{k-4},
   \]
   it suffices to show that for \(\sigma \in \btZ{k-3} \) with  \(|\sigma| \equiv 0 \pmod{2}\), there is an
   element \(\bv \in \{ \sigma \in V_0^{(k-3)} \, : \, q_2(\sigma) = \bv \}  \) such that \(p_1(\bv) = \sigma \).
   Let \(\sigma \in \btZ{k-3} \) with  \(|\sigma| \equiv 0 \pmod{2}\). Let \(\bar{\bv} = \bv \oplus \bZ_{k-4} \oplus \sigma\),
   the associated graph \(\mathcal{G}(\bar{\bv})\), it is a graph on \(k-3\) vertices with an even number \(|\bv|\) of loops
   where the vertex \(1\) has degree \(1\) (if it has a loop) or \(0\). We consider the two cases separatedly.

   Suppose that degree of the vertex \(1\) is \(0\). In this case, the subgraph \(\mathcal{G}_1 \) of \(\mathcal{G}(\bar{\bv}) \) obtained by
   removing the vertex \(1\) is a graph on \(k-4\) vertices with \(n = |\bv| \equiv 0 \pmod{2}\) loops, let \(\mathcal{G}_0 \) be \(\mathcal{G}_1 \)
   without the loops. As in (1), we know that \(\mathcal{G}_0 \) has an even number of vertices with odd degree. Let
   \(a\) (resp. \(b\)) be the number of vertices with odd degree (resp. even degree) in \(\mathcal{G}_0 \) that have a loop in \(\mathcal{G}_1 \).
   Let \(m_1\) (resp. \(m_0 \) ) be the number of vertices of odd degree in \(\mathcal{G}_1\) (resp. \(\mathcal{G}_0 \)), then we have
   \( m_1 = m_0 - a + b \). Since \(a + b = n \equiv 0 \pmod{2} \) and \(m_0 \equiv 0 \pmod{2} \), then \(a\) and \(b\) have the same parity and therefore \(m_1 \equiv 0 \pmod{2}\).
   Let \(\mathcal{G} \) be the graph obtained from \(\mathcal{G}(\bar{\bv})\) by adding edges from \(1\) to each of the (even number of) vertices 
   with odd degree. Then, \(\mathcal{G} \) is a graph where all vertices have even degree. It corresponds to a vector \(\br \in V_{\bZ}^{(k-3)} \) with \(p_1(\br) = \bv \)
   and \(q_2(\bv) = \sigma \). The case where the vertex \(1\) has a loop is dealt in a similar way. This proves the first part of (4). The second
   part follows directly from (1).

   Finally, let \(\bv \in \btZ{(k-4)(k-5)/2} \). By (1) and (3) the existence of a unique \(\br_0 \in V_{\bZ}^{(k-3)} \) with \(p_1(\br_0) = \bZ_{k-3}\) and
   \(q_2(\br_0) = \bv \) is guaranteed. First, we consider the case \(\bv = \bZ_{(k-4)(k-5)/2}\), where we have \( \br_0 = \bZ_{(k-3)(k-2)/2} \).
   Let \(\br \in V_{\bZ}^{(k-3)} \) with \(q_2(\br)= \bZ_{(k-4)(k-5)/2}\). Since \( p_1(\br) =  (r_{0,1},r_{0,2},\cdots,r_{0,k-3}) \), we are to prove
   \((r_{1,2},r_{1,3},\cdots,r_{1,k-3}) = q_1(\br) = (r_{0,2},r_{0,3},\cdots,r_{0,k-3}) \). The graph \(\mathcal{G}(\br)\) is a graph with \(k-3\) vertices of even degree,
   with an even number of loops and edges only of the form \((1,j)\) for \(r_{1,j} = 1\). If \(r_{0,j} =1 \) for \(j\geq 2 \) there is a loop in the vertex
   \(j\) and there must be a vertex \((1,j) \) to make the degree of \(j\) even, thus \(r_{1,j} = r_{0,j}=1 \). Similarly, if there is no loop in $j$, then
   there is no vertex \((1,j) \) in the graph. This proves (5) for the case \(\br = \bZ_{(k-4)(k-5)/2} \).

   Next, for general \(\bv \in \btZ{(k-4)(k-5)/2} \), let \(\bar{\bv} \in V_0^{(k-3)} \) the unique vector with \(p_1(\bar{\bv}) = p_1(\br) \) and
   \(q_2(\bar{\bv}) = \bZ_{(k-4)(k-5)/2}\). By our argument above, we have
   \((\bar{r}_{1,2},\bar{r}_{1,3},\cdots, \bar{r}_{1,k-3}) = (\bar{r}_{0,2},\bar{r}_{0,3},\cdots,\bar{r}_{0,k-3}) \).
     The graph \(\mathcal{G}(\br_0) \) is a simple graph with no loops and the graph \(\mathcal{G}(\bar{\bv}) \) is a graph where the edges that are not loops
     are of the form \((1,j)\) for \(j \geq 2 \)
     and where if such an edge appear then there is loop in \(j\). Both graphs have all vertices with even degree. From this, it is easy to see that
   the graph  \(\mathcal{G}(\bs)\) corresponding to \(\bs = \br_0 + \bar{\br} \pmod{2} \) has all even vertices and therefore \(\bs \in V_{\bZ}^{(k-3)} \).
   Moreover, \(p_1(\bs) = p_1(\br) \) and \(q_2(\bs) = q_2(\br) \), therefore, by (4), we have \(\bs = \br \), proving (5). This completes the proof of Lemma \ref{lem:v0}.
\end{proof}

\subsection{Summation via Fourier transforms}
\label{sec:transformation-sum}

With the preparations of the previous sections, we proceed to compute the innermost sum appearing
in \eqref{eq:InnerSum1}. By \eqref{eq:fouriercoeff7}, we have
\begin{align*}
  \sum_{\bs \in \setC{k-1}{v}{w}} g_{k-1}(u,\bs)  R^{(k,N)}_\mu(u,\bs)&= \frac{1}{4 u^{(k-2)\Delta}} \sum_{\bs \in \btZ{k-3}} g^{(v,w)}_{k-3}(u,\bs)
                                                                 R^{(v,w)}_\mu(u,\bs)
              = \frac{1}{2^{k-1} u^{(k-2)\Delta }} \sum_{\rho \in \btZ{k-3}} \widehat{g^{(v,w)}_{k-3}}(u,\rho) \widehat{R_\mu^{(v,w)}}(u,\rho) \\
         &= \frac{1}{4 u^{(k-2)\Delta}} \exp(a_0^{(\mu)}) \sum_{\rho \in \btZ{k-3}} \widehat{g^{(v,w)}_{k-3}}(u,\rho) \sum_{\xi \in \btZ{k-3}} C_\xi^{(\mu)} \delta_{\xi,\rho} \\
                       &= \frac{1}{4 u^{(k-2)\Delta}} \exp(a_0^{(\mu)}) \sum_{\rho \in \btZ{k-3}} \widehat{g^{(v,w)}_{k-3}}(u,\rho) \sum_{\br \in V^{(k-3)}_{\bZ}}  T^{(\sigma_{\rho}(\br))}(\mathbf{a}^{(\mu)}).
\end{align*}

By Proposition \ref{prop:fourier_g}, the sum in the last line can be written as 
\begin{align*}
  &\sum_{\rho \in \btZ{k-3}} \widehat{g^{(v,w)}_{k-3}}(u,\rho) \sum_{\br \in V^{(k-3)}_{\bZ}} T^{(\sigma_{\rho}(\br))}(\mathbf{a}^{(\mu)})  \nonumber\\ 
  &\quad = \sum_{\br \in V^{(k-3)}_{\bZ}} T^{(\br)}(\mathbf{a}^{(\mu)}) \sum_{\rho \in \btZ{k-3}}  (-1)^{v |\rho|}\left(u^{2 \varphi_{k-3}(\rho)\Delta}+ (-1)^{v+w}u^{2(k-2-\varphi_{k-3}(\rho))\Delta}\right) \prod_{i=1}^{k-3} \left( \tanh(a_i^{(\mu)})^{1-r_{0 i}}\right)^{\rho_i}.
\end{align*}
Setting \( A_i^{(\br)} = (-1)^v \tanh(a_i^{(\mu)})^{1-r_{0 i}} \), we obtain
\begin{equation*} 
  \sum_{\br \in V^{(k-3)}_{\bZ}} T^{(\br)}(\mathbf{a}^{(\mu)}) \sum_{\rho \in \btZ{k-3}} \left(u^{2 \varphi_{k-3}(\rho)\Delta}+ (-1)^{v+w}u^{2(k-2-\varphi_{k-3}(\rho))\Delta}\right)  \prod_{i=1}^{k-3} \left(A_i^{(\br)}\right) ^{\rho_i},
\end{equation*}
or equivalently
\begin{equation}
  \label{eq:expsumB}
  \sum_{\br \in V^{(k-3)}_{\bZ}} T^{(\br)}(\mathbf{a}^{(\mu)}) \left(f^{(\br)}_{k-3}(u^{2 \Delta}) + (-1)^{v+w} g^{(\br)}_{k-3} (u^{2 \Delta})\right),
\end{equation}
where the functions \(f^{(\br)}_k(\tau) \) and \(g^{(\br)}_k(\tau)\) are given by
\[
  f^{(\br)}_{k}(\tau) = \sum_{\rho \in \btZ{k}} \tau^{\varphi_k(\rho)} \prod_{i=1}^{k} (A_i^{(\br)})^{\rho_i}, \qquad
  g^{(\br)}_{k}(\tau) = \sum_{\rho \in \btZ{k}} \tau^{k+1-\varphi_k(\rho)} \prod_{i=1}^{k} (A_i^{(\br)})^{\rho_i}.
\]

Next, we compute explicitly the functions \(f^{(\br)}_{k}(\tau)\) and \(g^{(\br)}_{k}(\tau)\). For simplicity, we consider
the general case
\[
  f_{k}(\tau) = \sum_{\rho \in \btZ{k}} \tau^{\varphi_k(\rho)} \prod_{i=1}^{k} A_i^{\rho_i}, \qquad g_{k}(\tau) = \sum_{\rho \in \btZ{k}} \tau^{k+1-\varphi_k(\rho)} \prod_{i=1}^{k} A_i^{\rho_i},
\]
where \(A_i \in \C \) for \(i \in \{1,2,\cdots,k\}\). Note also that \(g_k(\tau) = \tau^{k+1} f_k(\tau^{-1})\).

\begin{prop}
  \label{prop:sumfg}
  For \(k \in \Z_{\geq 1} \), we have
  \begin{align*}  
    & f_k(\tau) + (-1)^{v+w} g_k(\tau) \\
    = & \frac{1}{2^k} \sum_{\ell=0}^k (1+\tau)^{k-\ell} (1-\tau)^{\ell} (1 + (-1)^{v+ w+\ell}\tau) \sum_{j_1 < j_2 < \ldots < j_{\ell}}  \prod_{i=0}^{\ell} \left( \prod_{n= j_i +1}^{j_{i+1}} (1+(-1)^{v+w+\ell-i} A_n) \right),
  \end{align*}
  where in the innermost product we have \(j_0 = 0 \) and \(j_{\ell+1} = k \). 
\end{prop}

\begin{proof}
  By property (1) of Theorem \ref{thm:varphi}, we obtain the system of simultaneous recurrence relations
  \begin{gather}
    \label{eq:recurrfg}
    f_k(\tau) = f_{k-1}(\tau) + A_k g_{k-1}(\tau), \qquad g_k(\tau) = \tau \left( g_{k-1}(\tau) + A_k f_{k-1}(\tau) \right) 
  \end{gather}
  with initial conditions \(f_0(\tau) = 1 \) and \(g_{0}(\tau) = \tau \).
  The recurrence \eqref{eq:recurrfg} can be written as 
  \begin{align*}
    \begin{bmatrix}f_k\\
      g_k
    \end{bmatrix}
    = \overleftarrow{\prod}_{j=1}^{k}
    \begin{bmatrix}
      1& A_j\\
      \tau A_j& \tau
    \end{bmatrix}
             \begin{bmatrix}f_0\\
               g_0
             \end{bmatrix},
  \end{align*}
  where \(f_0(\tau) = 1 \) and \(g_0(\tau) = \tau \). Notice that 
  \[
    \begin{bmatrix}
      1& A_j\\
      \tau A_j& \tau
    \end{bmatrix}
    =\begin{bmatrix}1&0\\
      0&\tau
    \end{bmatrix}
    \begin{bmatrix}1&A_j\\
      A_j&1
    \end{bmatrix}.
  \]

  Actually, we have
  \begin{align} \label{eq:recfg}
    \begin{bmatrix}f_k\\g_k
    \end{bmatrix}
    = C \overleftarrow{\prod}_{j=1}^{k}[a\bI+b\bJ]D(A_j) C
    \begin{bmatrix}f_0\\g_0
    \end{bmatrix},
  \end{align}
  where we $a=\frac{1+ \tau}2, b=\frac{1-\tau}2$, \(C\) is the Cayley transform
  \[
    C= \frac1{\sqrt2}\begin{bmatrix}1&1\\
      1&-1
    \end{bmatrix},
  \]
  and $D(x)$ is a two-by-two matrix-valued function given by
  \[
    D(x)=
    \begin{bmatrix}1+x&0\\
      0&1-x
    \end{bmatrix}.
  \]
  Indeed, \eqref{eq:recfg} follows immediately from the facts
  \[
    C
    \begin{bmatrix}1&A_j\\
      A_j&1
    \end{bmatrix}C = D(A_j), \quad
    C
    \begin{bmatrix}1&0\\
      0& \tau
    \end{bmatrix}C = \frac12 \Big((1+\tau)\bI+(1-\tau)\bJ\Big).
  \]
  Obviously, we have 
  \[
    \overleftarrow{\prod}_{j=1}^{k}[a \bI+b \bJ]D(A_j)
    =\sum_{\ell=0}^k a^{k-\ell}b^\ell \sum_{ \substack{\delta \in \btZ{k} \\ |\delta| = \ell } } \overleftarrow{\prod}_{j=1}^{k} D^{\delta_j}(A_j),
  \]
  where $D^{\delta_j}(A_j):=D(A_j)$ when $\delta_j=0$ and $D^{\delta_j}(A_j) := \bJ D(A_j)$ when $\delta_j=1$.

  For \(\delta \in \btZ{k}  \) with \(|\delta| = \ell \), define $j_i$ by enumerating 
  \[
    \{ j \, |\, \delta_j=1 \} = \{ (1\leq) j_1 < j_2 < \ldots < j_\ell (\leq k)\},
  \]
  then
  \begin{align*}
    \overleftarrow{\prod}_{j=1}^{k} D^{\delta_j}(A_j)
    = & D(A_k)\cdots \bJ D(A_{j_\ell})\cdots D(A_{j_{\ell-1}+1})\bJ D(A_{j_{\ell-1}})\cdots \bJ D(A_{j_1})\cdots D(A_1)\\
    = & D(A_k)\cdots D(-A_{j_\ell})\cdots D(-A_{j_{\ell-1}+1})D((-1)^2A_{j_{\ell-1}})\cdots D((-1)^\ell A_{j_1})
        \cdots D((-1)^\ell A_1)\bJ^\ell.
  \end{align*}
  It follows that
  \begin{align*}
    &\overleftarrow{\prod}_{j=1}^{k}[a \bI+b \bJ]D(A_j)\\
    =&\sum_{\ell=0}^k a^{k-\ell}b^\ell \sum_{j_1 < j_2 < \ldots < j_\ell}
       D(A_k)\cdots D(-A_{j_\ell})\cdots D(-A_{j_{\ell-1}+1})D((-1)^2A_{j_{\ell-1}})\cdots D((-1)^\ell A_{j_1})
       \cdots D((-1)^\ell A_1)\\
    =& \sum_{\ell=0}^k a^{k-\ell}b^\ell \sum_{j_1 < j_2 < \ldots < j_\ell}
       \prod_{i=j_\ell+1}^k D(A_i)\cdots \prod_{i={j_{\ell-1}+1}}^{j_\ell} D(-A_i)\cdots \prod_{i=1}^{j_1} D((-1)^\ell A_i) \bJ^\ell.
  \end{align*}

  Define, for a vector \( \mathbf{j} =\{j_1,j_2,\cdots,j_\ell \} \in \Z_{\geq 1}^{\ell} \) with \(1 \leq j_1 < j_2< \cdots < j_\ell \leq k \), the expressions
  \[
    S(\mathbf{j})  = \prod_{i=0}^{\ell} \left( \prod_{n= j_i +1}^{j_{i+1}} (1+(-1)^{\ell-i} A_n) \right), \qquad \bar{S}(\mathbf{j})  = \prod_{i=0}^{\ell} \left( \prod_{n= j_i +1}^{j_{i+1}} (1-(-1)^{\ell-i} A_n) \right),
  \]
  where \(j_0 := 0 \) and \(j_{\ell+1} := k \). Then, for \(\mathbf{j} \) as above we can write
  \[
    \prod_{i=j_\ell+1}^k D(A_i)\cdots \prod_{i={j_{\ell-1}+1}}^{j_\ell} D(-A_i)\cdots \prod_{i=1}^{j_1} D((-1)^\ell A_i)  =
    \begin{pmatrix}
      S(\mathbf{j}) & 0\\
      0 & \bar{S}(\mathbf{j}) 
    \end{pmatrix}.
  \]
  Noticing that the factor \(J^{\ell}\) depends only on the parity of \(\ell\), we obtain
  \begin{align*}
    f_k(\tau) =& \frac{1}{2^{k+1}}\sum_{\ell= 0}^{[k/2]} (1+\tau)^{k-2\ell}(1-\tau)^{2\ell} \sum_{j_1 < j_2 < \ldots < j_{2\ell}} \left( (1+\tau)S(\mathbf{j}) + (1-\tau) \bar{S}(\mathbf{j}) \right) \\
            &+ \frac{1}{2^{k+1}}\sum_{\ell= 1}^{[(k+1)/2]} (1+\tau)^{k-(2\ell-1)}(1-\tau)^{2\ell-1} \sum_{j_1 < j_2 < \ldots < j_{2\ell-1}} \left( (1-\tau)S(\mathbf{j}) + (1+t) \bar{S}(\mathbf{j}) \right),
  \end{align*}
  and
  \begin{align*}
    g_k(\tau) =& \frac{1}{2^{k+1}}\sum_{\ell= 0}^{[k/2]} (1+\tau)^{k-2\ell}(1-\tau)^{2\ell} \sum_{j_1 < j_2 < \ldots < j_{2\ell}} \left( (1+t)S(\mathbf{j}) - (1-\tau) \bar{S}(\mathbf{j}) \right) \\
            &+ \frac{1}{2^{k+1}}\sum_{\ell= 1}^{[(k+1)/2]} (1+\tau)^{k-(2\ell-1)}(1-\tau)^{2\ell-1} \sum_{j_1 < j_2 < \ldots < j_{2\ell-1}} \left( (1-\tau )S(\mathbf{j}) - (1+\tau) \bar{S}(\mathbf{j}) \right).
  \end{align*}
Hence the results follows.
\end{proof}

We remark here that for \(1 \leq \ell \leq k\), each set of \(\ell \) numbers \(j_i\) ( \(1 \leq j_1 < j_2 < \cdots < j_\ell \leq k \)) determine a unique vector \( \rho \in \btZ{k}\) such that \(|\rho| = \ell \) and where \(j_i \) is the position of the \(i\)-th one in \(\rho\). Likewise, each vector \( \rho \in \btZ{k}\) determines a unique set of \(\ell = |\rho| \) integers such that \(1 \leq j_1 < j_2 < \cdots < j_\ell \leq k \) by setting \(j_i \) as the position of the \(i\)-th one in \(\rho\).

By Proposition \ref{prop:sumfg} applied to \(A_i^{(\bm{r})}\), \eqref{eq:expsumB} is equal to
\begin{align} \label{eq:expsumC}
  &\frac{1}{2^{k-1} u^{(k-2) \Delta}}  \sum_{\ell=0}^{k-3} (1+u^{2\Delta})^{k-\ell} (1-u^{2\Delta})^{\ell} (1 + (-1)^{v+ w+\ell} u^{2\Delta}) \nonumber \\
  & \qquad \qquad \qquad \times \sum_{j_1 < j_2 < \ldots < j_{\ell}} \sum_{\br \in V_{\bZ}^{(k-3)}} T^{(\br)}(\mathbf{a}^{(\mu)}) \prod_{i=0}^{\ell} \left( \prod_{n= j_i +1}^{j_{i+1}} (1+(-1)^{w+\ell-i} \tanh(a_n^{(\mu)})^{1-2 r_{0 n}} ) \right).
\end{align}
Next we deal with the innermost sum over the set \(V_0^{(k-3)}\).
Let \(x \in \Z\),  \(\rho = (\rho_1,\rho_2, \cdots,\rho_{k-3})   \in \btZ{k-3} \) with \(|\rho|=\ell \), and \(1 \leq j_1 < j_2 < \cdots < j_{\ell} \) be the position of the ones in \(\rho\). We have
\begin{align*}
  &\sum_{\br \in V_{\bZ}^{(k-3)}} T^{(\br)}(\mathbf{a}^{(\mu)}) \prod_{i=0}^{\ell} \left( \prod_{n= j_i +1}^{j_{i+1}} \left(1 +(-1)^{x+\ell-i} \tanh(a_n^{(\mu)})^{1-2 r_{0 n}} \right) \right) \\
  & = \sum_{\br \in V_{\bZ}^{(k-3)}} T^{(\br)}(\mathbf{a}^{(\mu)}) \prod_{i=1}^{k-3} \left(1 + (-1)^{v_0 + v_i} \tanh(a_i^{(\mu)})^{1-2 r_{0 i}} \right).  
\end{align*}
with \(v_0 = x\) and  \(v_i = v_i (\rho) = \sum_{j=i}^k \rho_j\), for \(i = 1,2,\ldots,k-3\).

\begin{lem}
  \label{lem:sumv0}
  Let \(v_i \in \C \) for \(i \in \{0,1,2,\cdots,k-3\}\). We have
  \begin{equation}
    \label{eq:lemsumV0}
    \sum_{\br \in V^{(k-3)}_{\bZ}} T^{(\br)}(\mathbf{a}^{(\mu)}) \prod_{i=1}^{k-3} \left(1 + (-1)^{v_0 + v_i} \tanh(a_i^{(\mu)})^{1-2 r_{0 i}} \right) = \exp\left( \sum_{m=0}^{k-4}  \sum_{j=1}^{k-3-m} (-1)^{v_m + v_{m+j}} a_{m, m+j}^{(\mu)}\right).
  \end{equation}
\end{lem}

\begin{proof}
  The proof is by induction.  For simplicity, in this proof we drop the dependency of $\mu$ from the notation of the
  coefficients \(a_{i,j}^{(\mu)}\). It is immediate to verify the result for the cases \(k-3=1,2\).
  For \(\br \in V_0^{(k-3)}\), the single summand of \eqref{eq:lemsumV0} corresponding to
  \(\br\) is 
  \begin{align*}
   \prod_{i=1}^{k-3} \left( \cosh(a_i)^{1-r_{0 i}} \sinh(a_i)^{r_{0 n}} \right) \left(1 + (-1)^{v_0 + v_i} \tanh(a_i)^{1-2 r_{0 i}} \right) T^{(p_2(\br))}(p_2(\mathbf{a})),
  \end{align*}
  it is not difficult to see that it can be written as
  \begin{align*}
    & (-1)^{|p_1(\br)|v_0 + \sum_{i=1}^{k-3} \br_{0 i} v_i} \prod_{i=1}^{k-3}  \cosh(a_i) \left(1 + (-1)^{v_0 + v_i} \tanh(a_i) \right) T^{(p_2(\br))}(p_2(\mathbf{a}))  \\
    & = (-1)^{\sum_{i=1}^{k-3} \br_{0 i} v_i} \prod_{i=1}^{k-3}  \cosh(a_i) \left(1 + (-1)^{v_0 + v_i} \tanh(a_i) \right) T^{(p_2(\br))}(p_2(\mathbf{a})),
  \end{align*}
  since \(|p_1(\br)| \equiv 0 \pmod{2}\). Next, observe that
  \[
    \prod_{i=1}^{k-3}  \cosh(a_i) \left(1 + (-1)^{v_0 + v_i} \tanh(a_n) \right)  = \exp\left( \sum_{i=1}^{k-3} (-1)^{v_0 + v_i} a_{0 i} \right),
  \]
  thus, the expression above is given by 
  \[
    (-1)^{\sum_{i=1}^{k-3} \br_{0 i} v_i} \exp\left( \sum_{i=1}^{k-3} (-1)^{v_0 + v_i} a_{0 i}  \right) T^{(p_2(\br))}(p_2(\mathbf{a})).
  \]

  Next, for \(\bv \in \btZ{(k-4)(k-5)/2}\), we define the set \(S(\bv) \subset V_0^{(k-3)}\) as  \( S(\bv) =  \{ \sigma \in V_0^{(k-3)} \, : \, q_2(\sigma) = v \}\).
  By Lemma \ref{lem:v0}(4), we have  \(|S(\bv)| = 2^{k-4}\). For \(\bv \in \btZ{(k-4)(k-5)/2}\), we have
  \begin{align*}
    &\sum_{\br \in S(\bv)} T^{(\br)}(\bm{a}) \prod_{i=1}^{k-3} \left(1 + (-1)^{v_0 + v_i} \tanh(a_i)^{1-2 r_{0 i}} \right)  \\
    &\quad = \exp\left( \sum_{i=1}^{k-3} (-1)^{v_0 + v_i} a_{0 i}  \right) T^{(\bv)}(q_2(\bm{a})) \sum_{\br \in S(\bv)} (-1)^{\sum_{i=1}^{k-3} \br_{0 i} v_i} T^{(q_1(\br))}(q_1(\bm{a}),
  \end{align*}
  Let \(\bar{\br} \in V_0^{(k-3)}\) be the unique element such that \(p_1(\bar{\br})= \bZ_{k-3} \) and \(q_2(\bar{\br}) = \bv \). Then, by the proof of Lemma \ref{lem:v0}(5), we can write
  \(\br \in S(\bv) \) as 
  \[
    \br = \bar{\br} + \hat{\br},
  \]
  where \(p_1(\hat{\br}) = p_1(\br) \) and \(q_2(\hat{\br}) = \bZ_{(k-4)(k-5)/2}\). Moreover, also by Lemma \ref{lem:v0}(5), we have
    \[
      \hat{\br} = (\br_{0,1},\br_{0,2},\cdots, \br_{0,k-3},\br_{0,2},\cdots,\br_{0,k-3}) \oplus \bZ_{(k-4)(k-5)/2}.
    \]
  Therefore,  by \ref{lem:v0}(4) we have
  \begin{align*}
    \sum_{\br \in S(\bv)} (-1)^{\sum_{i=1}^{k-3} \br_{0 i} v_i} T^{(q_1(\br))}(q_1(\bm{a})) &= \sum_{\substack{ \br \in \btZ{k-3} \\ |\br|\equiv 0 \pmod{2}}} (-1)^{\sum_{i=1}^{k-3} \br_{0 i} v_i}
   T^{ ((\br_{0,2},\br_{0,3},\cdots,\br_{0,k-3}) + q_1(\bar{\br}))}(q_1(\bm{a})) \\
   &= \sum_{\substack{ \br \in \btZ{k-3} \\ |\br|\equiv 0 \pmod{2}}} (-1)^{\sum_{i=1}^{k-3} \br_{0 i} v_i}  \prod_{i=2}^{k-3} (\tanh(a_{1,i}^{1-2 \bar{r}_{1 i}}))^{\br_{0 i}} T^{(q_1(\bar{\br}))}(q_1(\bm{a})) \\
   &= \prod_{i=2}^{k-3}  \cosh(a_{1 i})^{1-\bar{r}_{1 i}} \sinh(a_{1 i})^{\bar{r}_{1 i}} \left(1 + (-1)^{v_1 + v_i} \tanh(a_{1 i})\right), 
  \end{align*}
  the last equality holding since \(|\br| \equiv 0 \pmod{2} \).
  
  The sum in \eqref{eq:lemsumV0} is given by the sums \(S(\bv)\) over all vectors \( \bv \in \btZ{(k-4)(k-5)/2}\). Namely, we have, by Lemma \ref{lem:v0}(3),
  that \eqref{eq:lemsumV0} is given by
  \begin{align*}
    &\exp\left( \sum_{i=1}^{k-3} (-1)^{v_0 + v_i} a_{0 i}  \right) \sum_{\br \in V_0^{(k-4)} }  \prod_{i=2}^{k-3} \cosh(a_{1 i})^{1-r_{1 i}} \sinh(a_{1 i})^{r_{1 i}} \left(1 + (-1)^{v_1 + v_i} \tanh(a_{1 i})\right) T^{(p_2(\br))}(q_2(\bm{a})) \\
    &\quad = \exp\left( \sum_{i=1}^{k-3} (-1)^{v_0 + v_i} a_{0 i}  \right) \sum_{\br \in V_0^{(k-4)} } T^{(\br)}(p_2(\bm{a})) \prod_{i=2}^{k}  \left(1 + (-1)^{v_1 + v_i} \tanh(a_{1 i})\right) 
  \end{align*}
  where each element \(\br \in V_0^{(k-4)}\) has entries \(\br = (r_{1,2},r_{1,3}, \cdots,r_{1,k}, r_{2,3}, r_{2,4}, \cdots, r_{k-4,k-3}) \).  The result follows by induction.
\end{proof}

Next, using Lemma \ref{lem:sumv0} with \(x = w \) in \eqref{eq:expsumC} we see that the main limit in the expression of the heat kernel is given by
\begin{align} \label{eq:sumheat1}
  &\lim_{N \to \infty}\left( \frac{1-u^{\frac{2\Delta}{N}}}{2 u^{\frac{\Delta}{N}}} \right) \sum_{\mu=0}^1 \sum_{k \geq 3}^N   \left( \frac{1+u^{\frac{2\Delta}{N}}}{2 u^{\frac{\Delta}{N}}} \right)^{N-3}  \nonumber  \\
  &\times  \frac12 J_\mu^{(k,N)}(x,y,u^{\frac1N})
    \Bigg\{
    \sum_{\rho \in \btZ{k-3}} 
    \begin{bmatrix}
      (-1)^{|\rho|+1} u^{\frac{\Delta}{N}} &  (-1)^{|\rho|+\mu} u^{\frac{\Delta}{N}} \\
       (-1)^{\mu+1}u^{-\frac{\Delta}{N}} &  u^{-\frac{\Delta}{N}}
    \end{bmatrix}
    \left(\frac{1-u^{\frac{2\Delta}N}}{1+u^{\frac{2 \Delta}N}}\right)^{|\rho|}  \nonumber  \\
  & \qquad \qquad \qquad
    \times \exp\left(a_0^{(\mu)}(u^{\frac1N}) + \sum_{m=0}^{k-4}  \sum_{j=1}^{k-3-m} (-1)^{(|\rho|+\mu-1)\delta_0(m) + \sum_{i=m}^{m+j-1} \rho_i} a_{m, m+j}^{(\mu)}(u^{\frac1N}) \right)  \Bigg\} ,
\end{align}
where we wrote the sum appearing at the right hand side of Lemma \ref{lem:sumv0} in terms of the vector \(\rho\) and
where \(\delta_y(x)\) is the Kronecker delta function.

Let us further simplify the factors appearing as arguments  in the exponential function in the above limit.
By Lemma \ref{lem:fourier4} we have \(a_0^{(\mu)}(u^{\frac1N}) = O\left(\frac1N \right) \). For \(\rho \in \btZ{k-3}\), we have
\begin{align*} 
  (-1)^{|\rho|}&\sum_{j=1}^{k-3} (-1)^{\sum_{i=1}^{j-1} \rho_i} a_{0, j}^{(\mu)}(u^{\frac1N)})  =   \\
  & \frac{ (1-u^{\frac1N})}{1-u^2} \Bigg[  \sqrt{2}g  x \sum_{j=1}^{k-3} (-1)^{\sum_{i=j}^{k-3} \rho_i} \left( u^{\frac{j}N} - u^{2-\frac{j+1}N}\right)  +  \sqrt{2}g y u \sum_{j=1}^{k-3} (-1)^{\sum_{i=j}^{k-3} \rho_i} \left( u^{-\frac{j+1}N} - u^{\frac{j}N}\right) \\
  &+(-1)^\mu \frac{2g^2}{(1+u^{\frac1N})} u^{\frac{k}N}(1-u^{1-\frac{k}N})(1-u^{1-\frac{k-1}N})
      \sum_{j=1}^{k-3} (-1)^{\sum_{i=j}^{k-3} \rho_i} \left( u^{-\frac{j+1}N} - u^{\frac{j}N}\right) \Bigg] + O(\tfrac{1}N).
\end{align*}
Using identity (4) of Theorem \ref{thm:varphi}, the above is equal to
{\footnotesize
\begin{align*} 
& \frac{ u^{\frac1N}}{1-u^{2}} \left(  \sqrt{2} g \left[ x  (1-u^{\frac{k-3}{N}})(1-u^{2-\frac{k-1}{N}})
                                   -  u y (1-u^{\frac{k-3}{N}})(1-u^{-\frac{k-1}{N}}) \right]
                                   - (-1)^\mu \frac{2g^2 u^{\frac{k}N}(1-u^{1-\frac{k}N})(1-u^{1-\frac{k-1}N})(1-u^{\frac{k-3}{N}})(1-u^{-\frac{k-1}{N}})}{(1+u^{\frac1N})}  \right)
  \nonumber  \\
    &\quad + \frac{2(1-u^{\frac1N})}{1-u^2}  \Bigg( \sqrt{2} g \left[ x \left( u^{2-\frac2N} \varphi_{k-3}(\rho;u^{-\frac1N}) - u^{\frac1N} \varphi_{k-3}(\rho;u^{\frac1N}) \right) +
                 u y \left( u^{\frac1N} \varphi_{k-3}(\rho;u^{\frac1N}) - u^{-\frac2N}\varphi_{k-3}(\rho;u^{-\frac1N}) \right)\right] \nonumber \\
               &\qquad \qquad \qquad \qquad + (-1)^\mu \frac{2g^2 u^{\frac{k}N}(1-u^{1-\frac{k}N})(1-u^{1-\frac{k-1}N})}{(1+u^{\frac1N})}\left( u^{\frac1N} \varphi_{k-3}(\rho;u^{\frac1N}) - u^{-\frac2N}\varphi_{k-3}(\rho;u^{-\frac1N}) \right) \Bigg)  + O(\tfrac{1}N) ,
\end{align*}}

On the other hand, the sum of the Fourier coefficients \(a_{i,j}^{(\mu)}(u^{\frac1N})  \) with \(1\leq i < j \) is given by
\begin{gather*}
  \sum_{m=1}^{k-4}  \sum_{j=1}^{k-3-m} (-1)^{\sum_{i=m}^{m+j-1} \rho_i} a_{m, m+j}^{(\mu)}(u^{\frac1N}) = \sum_{m=1}^{k-4}  \sum_{j=m}^{k-4} (-1)^{\sum_{i=m}^{j} \rho_i} a_{m,j+1}^{(\mu)} (u^{\frac1N}) \nonumber \\
  \qquad \qquad = \frac{g^2 (1-u^{\frac1N})^2 }{1-u^{2}} \sum_{m=1}^{k-4} (u^{-\frac{m}{N}}-u^{\frac{m +1}{N}})  \sum_{j=m}^{k-4} (u^{\frac{j}{N}}-u^{2 -\frac{j+3}{N}}) (-1)^{\sum_{i=m}^{j} \rho_i}  \nonumber \\
  \qquad \qquad = \frac{g^2 (1-u^{\frac1N})^2 }{1-u^{2}} \sum_{j=1}^{k-4} (u^{\frac{j}{N}}-u^{2 -\frac{j+3}{N}})  \sum_{m=1}^{j} (u^{-\frac{m}{N}}-u^{\frac{m +1}{N}})   (-1)^{\sum_{i=m}^{j} \rho_i},
\end{gather*}
and by using Theorem \ref{thm:varphi}(4) once more, we see that this is equal to
\begin{gather*} 
  \frac{g^2}{1-u^2}\left( (1-u^{\frac1N})(1-u^{2-\frac1N}) (k-4)  -  u^{\frac1N}(1+u^{\frac2N}) (1-u^{\frac{k-4}N})(1-u^{2 - \frac{k}N}) + \frac{u^{\frac4N}(1-u^{\frac{2k-8}N})(1-u^{2 - \frac{2k-1}N})}{1+u^{\frac1N}}   \right) \nonumber \\
  + \frac{2g^2(1-u^{\frac1N})^2}{(1-u^2)} \sum_{j=1}^{k-4} (u^{\frac{j}{N}}-u^{2 -\frac{j+3}{N}}) \left( u^{\frac2N} \varphi_j(\pref_j(\rho);u^{\frac1N}) - u^{-\frac1N} \varphi_j(\pref_j(\rho);u^{-\frac1N}) \right),
\end{gather*}
where \(\pref_j(\rho)\) is the prefix of lenght \(j\) of \(\rho\). Concretely, if \(\rho \in \btZ{k}\) and \(j \leq k\), then  \(\pref_j(\rho) : \btZ{k} \to \btZ{j}\) is  the projection into \(\btZ{j}\) of the first \(j\) elements of \(\rho\).

Next, we define auxiliary functions $H_{\eta}^{(k,N)}$,$P^{(k,N)}(u^{\frac1N},\rho)$ and $Q^{(k,N)}_{\eta}(x,y,u^{\frac1N})$  containing the expressions appearing above. In \S \ref{sec:transf-summ-into}, we describe how to evaluatate certain sums containing $H_{\eta}^{(k,N)}$ and $P^{(k,N)}(u^{\frac1N},\rho)$, restricted to a fixed value \(|\rho|= \lambda \in \Z_{\geq 0}\), as iterated Riemann integrals.

\begin{dfn} \label{dfn:HPQ}
    Let \( \rho \in \btZ{k-4}\) and \(\eta \in \{0,1\} \). The functions $H_{\eta}^{(k,N)}(x,y,u^{\frac1N},\rho)$,$P^{(k,N)}(u^{\frac1N},\rho)$ and  $Q^{(k,N)}_{\eta}(x,y,u^{\frac1N})$ are given by
  \begin{align*}
    H_{\eta}^{(k,N)}(x,y,u^{\frac1N},\rho) &:= \exp\Bigg( (-1)^{\eta} \frac{2\sqrt{2} g(1-u^{\frac1N})}{1-u^2}\Big[ x \left( u^{2-\frac2N} \varphi_{k-3}(\rho;u^{-\frac1N}) - u^{\frac1N} \varphi_{k-3}(\rho;u^{\frac1N}) \right) \\
                               &\qquad \qquad \qquad \qquad \qquad \qquad+ u y \left( u^{\frac1N} \varphi_{k-3}(\rho;u^{\frac1N}) - u^{-\frac2N}\varphi_{k-3}(\rho;u^{-\frac1N}) \right)\Big] \Bigg),
  \end{align*}

  \begin{align*}
    P^{(k,N)}(u^{\frac1N},\rho) &:= \exp\left( \frac{2g^2(1-u^{\frac1N})^2}{(1-u^2)} \sum_{j=1}^{k-4} (u^{\frac{j}{N}}-u^{2 -\frac{j+3}{N}}) \left( u^{\frac2N} \varphi_j(\pref_j(\rho);u^{\frac1N}) - u^{-\frac1N} \varphi_j(\pref_j(\rho);u^{-\frac1N}) \right) \right) \\
    & \quad \times \exp\left( - \frac{2g^2 u^{\frac{k}N}(1-u^{\frac1N})}{(1-u^2)} (1-u^{1-\frac{k}N})^2\left( u^{\frac1N} \varphi_{k-3}(\rho;u^{\frac1N}) - u^{-\frac2N}\varphi_{k-3}(\rho;u^{-\frac1N}) \right) \right),
  \end{align*}

  \begin{align*}
    Q^{(k,N)}_{\eta}&(x,y,u^{\frac1N})  :=   \exp\left( \frac{(-1)^{1-\mu}\sqrt{2} g}{1-u^{2}} \left[ x (1-u^{\frac{k}{N}})(1-u^{2-\frac{k}{N}}) -  u y (1-u^{\frac{k}{N}})(1-u^{-\frac{k}{N}}) \right] \right)     \\
         &\, \times \exp\left( g^2\frac{t k}{N} + \frac{g^2}{1-u^2} \left(  -2(1-u^{\frac{k}N})(1-u^{2-\frac{k}N}) + \frac12(1-u^{\frac{2 k}{N}})(1-u^{2-\frac{2 k}{N}}) - (1-u^{1-\frac{k}N})^2(1-u^{\frac{k}N})^2 \right) \right).
  \end{align*}
\end{dfn}

We note here that in the limit \eqref{eq:sumheat1} the summands of the innermost sum consists of a sum of a radial
function on \(\rho\) multiplied by an exponential factor.Moreover, by Lemma \ref{lem:bij1}, the exponential factor
is also determined by fixing the norm $|\rho|$. This is an essential fact for the evaluation of the limit appearing
in the heat kernel of the QRM as a Riemann sum in \S \ref{sec:limit}.

\begin{rem}
  In \cite{TK1989}, in the context of a test system interacting with a heat bath consisting of harmonic oscillators,
  the Laplace transform of the reduced density matrix (given as a series of iterated integrals) is introduced in order to
  recover the evolution equation of the stochastic model in a generalized form. 
  It may be interesting to compare the discussion in \cite{TK1989} with our method for obtaining a reasonable evaluation of
  the sum involving $g_{k-1}(\bs) R_\mu(\bs)$ over $\bs$ using the Fourier transform on $\Z_2^k$ described in this section.
\end{rem}

\begin{rem}\label{q-Fourier}
  The computations using Fourier transform for the group \(\Z_2^k\) in this section can be interpreted directly in
  terms of the quantum Fourier transform for \(k\)-qubits (see e.g. \cite{BGS2004}). Thus, it may be interesting
  to describe the technique developed in this section in the setting of quantum computation (complex Hilbert spaces
  of dimension \(2^k\) and quantum Fourier transform) in place of the finite group setting
  (\(\Z_2^k\) and discrete Fourier transform).
\end{rem}

\begin{rem}\label{non-commutativity}
  The Hamiltonian \(\HRabi^{\e} \) of the asymmetric quantum Rabi model (AQRM) is given by
  \[
    \HRabi^{\e} := a^{\dagger}a + \Delta \sigma_z + g (a + a^{\dagger}) \sigma_x + \e \sigma_x,
  \]
  for \(\e \in \R\) (\cite{B2011PRL,bcbs2016} by the name of generalized quantum Rabi model). Here, we have again
  assumed the frequency $\omega$ of the bosonic mode to be $1$. This model is also important experimentally (see
  e.g. \cite{YS2018}). Even though the spectrum of \(\HRabi^{\e}\) is known to have degeneracies of multiplicity
  two for \(\e \in \frac12 \Z\), no \(\Z_2\)-symmetry (parity) has been observed
  in \(\HRabi^{\e}\) (and is hard to expect) for any nonzero value of \(\e \in \R\) (\cite{KRW2017}).
  
  Nonetheless, it is possible to follow the discussion in this section for the AQRM
  by using the Trotter-Kato formula with the operators $b^\dag b-g^2$ and $\Delta \sigma_z + \e \sigma_x$. The computation
  largely remains the same, with more complicated expressions, and, in particular, the finite groups
  $\Z_2^k\, (k=0,1,2,\ldots)$ appear as well.
  In fact, the appearance of the finite groups in the computation of this section is due to the
  decomposition of the matrix terms in \eqref{eq:matterms} (containing $\sigma_x, \sigma_z$) and is not related to the
  existence of a \(\Z_2\)-symmetry in the Hamiltonian. For instance,  in the case of the AQRM we may use 
  \[
    G_N^{(\e)}(u,\bs) :=  \frac{1}{2^N} \overrightarrow{\prod}_{i=1}^{N}[\bI+(-1)^{1-\bs(j)} \bJ ]u^{\Delta\sigma_z+ \e \sigma_x}.
  \]
  in place of $G_N(u,\bs)$ in \eqref{eq:non-com-term}.
  For more complicated or generalized models (i.e. Dicke model), other finite groups may appear in the computation
  depending on the definition of the simplified self-adjoint operator (the analog of $b^\dag b$), the
  non-commutativity among the terms in the objective Hamiltonian. 
  
  We also remark that the choice of a pair of self-adjoint operators to be used in the Trotter-Kato product formula
  is non canonical. In fact, even in the case of the QRM, we have several possibilities. For instance, we can consider
  the pair of self-adjoint operators \(a^{\dagger}a + \Delta \sigma_z \) and \(g (a + a^{\dagger}) \sigma_x\). Since \(a^{\dagger}a\) and \(\Delta \sigma_z \) obviously
  commute, the heat kernel of \(a^{\dagger}a + \Delta \sigma_z \) can be obtained without difficulty. In this choice, the discussion using
  the Trotter-Kato product formula will be, however, highly complicated though the associated finite groups are still the
  family $\Z_2^k\,(k=0,1,2,\ldots)$.    
  Another option is to note that the Hamiltonian
  \[
    H_1 := a^{\dagger}a + \Delta \sigma_x + g (a + a^{\dagger}) \sigma_z,
  \]
  is unitarily equivalent to \(\HRabi\) (given by the finite dimensional Cayley transform
  $C= \frac1{\sqrt2}\begin{bmatrix}1&1\\ 1&-1 \end{bmatrix}$), thus by defining \(d := a + g \sigma_z \) we can consider the
  pair \(d^\dag d-g^2\) and \(\Delta \sigma_x\). In this cases, the discussion of this section should follow with appropriate
  changes in the computations.
\end{rem}

\subsection{Riemann sums and residual terms}
\label{sec:transf-summ-into}

In this subsection we compute the sums given in the previous section \S \ref{sec:spvl} by changing sums to integrals
with residual terms with explicitly given order. Concretely, for \(\lambda \geq 1\) we proceed to rewrite the sum
\begin{align}
  \label{eq:sumhp}
  \sum_{\substack{\rho \in \btZ{k-3} \\ |\rho| = \lambda  } }   H_{\eta}^{(k,N)}(x,y,u^{\frac1N},\rho) P^{(k,N)}(u^{\frac1N},\rho)
\end{align}
into an expression that can be interpreted as multiple iterated integrals over the $\lambda$-th simplex.

We start with a lemma used to deal with the sums including terms \(\varphi_j(\pref_j(\rho);s) \).
The reader may find useful to interpret the lemma in light of the bijection \eqref{eq:bij1} (cf. equation \eqref{eq:phidef} and
Lemma \ref{lem:bij1}).

\begin{lem}
  \label{lem:sumexp}
  For \(k \geq 1\) and \( 1 \leq \lambda \leq k \), for the indeterminates \(t,s \) we have
  \[
    \sum_{\substack{ \rho \in \btZ{k} \\ |\rho|=\lambda }}\exp\left( \sum_{j=1}^k t^j  \varphi_j(\pref_j(\rho);s) \right) =
    \sum_{i_1< i_2<\cdots< i_{\lambda}}^k \exp\left( \sum_{\substack{0 \leq \alpha<\beta \\ \beta - \alpha \equiv 1 \pmod{2}  }}^{\lambda} t^{i_\beta} s^{i_{\alpha}} [i_{\alpha+1}-i_{\alpha}]_s [i_{\beta+1}-i_{\beta}]_t \right),
  \]
  where \( i_0:=0 \), \( i_{\lambda+1} := k+1\).
  Moreover, for \(\rho  \in \btZ{k} \) with \(|\rho|=\lambda\), we have
  \[
     \exp\left( \sum_{j=1}^k t^j  \varphi_j(\pref_j(\rho);s) \right) = \exp\left( \sum_{\substack{0 \leq \alpha<\beta \\ \beta - \alpha \equiv 1 \pmod{2}  }}^{\lambda} t^{i_\beta} s^{i_{\alpha}} [i_{\alpha+1}-i_{\alpha}]_s [i_{\beta+1}-i_{\beta}]_t \right).
   \]
   where \( i_0:=0 \), \( i_{\lambda+1} := k+1\) and  \(i_j \), for \(1 \leq j \leq \lambda \), is the position of
   the $j$-th one in \(\rho\).
\end{lem}

\begin{proof}
  Let us first consider the case \(|\rho| = \lambda =1\). In this case \( \rho = \bm{e}_i\) for \( 1 \leq i \leq k \).
  From the definition of \(\varphi_j\), we verify that
  \[
    \varphi_j(\pref_j(\rho);s) =
    \begin{cases}
      0 & \text{ if } j < i \\
      [i]_s & \text{ if } i \leq j
    \end{cases},
  \]
  thus
  \[
    \sum_{j=1}^k t^j  \varphi_j(\pref_j(\rho),\;s) = [i]_s \sum_{j=i}^k t^j = t^i [k+1-i]_t [i]_s.
  \]
  Next, lets assume the result holds for all \( |\rho|= \lambda -1\) and consider the case \(|\rho|=\lambda\).
  Set \( \omega := \pref_{i_{\lambda}-1}(\rho) \), then we have
  \[
    \sum_{j=1}^k t^j \varphi_j(\pref_j(\rho);s) = \sum_{j=1}^{i_{\lambda-1}} t^j  \varphi_j(\pref_j(\omega);s) + \sum^k_{j=i_{\lambda}} t^j  \varphi_k(\rho;s),
  \]
  since \(\varphi_j(\pref_j(\rho);s) = \varphi_k(\rho;s)  \) for \(j \geq i_{\lambda}\).
  On the one hand, we have
  \[
    \sum_{j=1}^{i_{\lambda-1}} t^j  \varphi_j(\pref_j(\omega);s) = \sum_{\substack{0=\alpha<\beta \\ \beta - \alpha \equiv 1 \pmod{2}  }}^{\lambda-1} t^{i_\beta} s^{i_{\alpha}} [i_{\alpha+1}-i_{\alpha}]_s [i_{\beta+1}-i_{\beta}]_t,
  \]
  by induction since \(|\omega| = \lambda-1 \). On the other hand, we have
  \[
    \varphi_k(\rho;s) = 
    \begin{cases}
      \sum_{n=1}^{\frac{\lambda}2} [i_{2 n}]_s - [i_{2 n - 1}]_s     &    \text{ if }  \lambda \equiv 0 \pmod{2} \\
      \sum_{n=0}^{\frac{(\lambda-1)}2} [i_{2 n+1}]_s - [i_{2 n}]_s &    \text{ if }  \lambda \equiv 1 \pmod{2}
    \end{cases}.
  \]
  Let us consider the case \(\lambda \equiv 0 \pmod{2} \) since the alternative case is completely analogous. We immediately
  verify that 
  \[
    \sum_{n=1}^{\frac{\lambda}2} [i_{2 n}]_s - [i_{2 n - 1}]_s  =  \sum_{n=1}^{\frac{\lambda}2} s^{i_{2 n - 1}} [i_{2 n} - i_{2 n - 1}]_s,
  \]
  and substituting in the second sum of the right-hand side we obtain
  \[
    \sum^k_{j=i_{\lambda}} t^j  \varphi_k(\rho;s)  = t^{i_{\lambda}} \sum_{n=1}^{\frac{\lambda}2} s^{i_{2 n - 1}} [i_{2 n} - i_{2 n - 1}]_s \sum^{k-t_i}_{j=0} t^{j} =   t^{i_{\lambda}} [k+1-i_{\lambda}]_t \sum_{n=1}^{\frac{\lambda}2} s^{i_{2 n - 1}} [i_{2 n} - i_{2 n - 1}]_s ,
  \]
  finally, notice that since \(\lambda\) is even and \(2 n -1 \) for \( 1 \leq n \leq \frac{\lambda}2 \) runs over all odd integers
  smaller than \(\lambda\) we see that the above is equal to
  \[
    \sum_{\substack{0\leq\alpha< \lambda \\ \lambda - \alpha \equiv 1 \pmod{2}  }} t^{i_{\lambda}} s^{t_{\alpha}} [i_{\alpha+1} - i_{\alpha}]_s [i_{\lambda+1}-i_{\lambda}]_t,
  \]
  with \(i_{\lambda+1} := k+1 \), as desired.
\end{proof}

Let us consider a fixed \(\lambda \geq 1 \) and \(\rho \in \btZ{k-3} \) with \(|\rho|= \lambda \). As usual, we denote by 
\(1 \leq i_1 <i_2< \cdots< i_\lambda \) the position of \(1\) in \(\rho\). By Lemma \ref{lem:sumexp} and \eqref{eq:varphit}, we see
that \(H_{\eta}^{(k,N)}(u^{\frac1N},\rho) P^{(k,N)}(u^{\frac1N},\rho)\) is given by
\begin{align*}
  &\exp\left( (-1)^{\eta} \frac{2\sqrt{2} g(1-u^{\frac1N})}{1-u^2} \sum_{\gamma=1}^{\lambda} (-1)^{\gamma-1} \left[ x \left( u^{2-\frac2N} [i_\gamma]_{u^{-\frac1N}} - u^{\frac1N} [i_\gamma]_{u^{\frac1N}}  \right)
    + u y \left( u^{\frac1N} [i_\gamma]_{u^{\frac1N}} - u^{-\frac2N} [i_\gamma]_{u^{-\frac1N}} \right) \right] \right) \\
  & \times \exp\left( - \frac{2g^2(1-u^{\frac1N})}{(1-u^2)} u^{\frac{k}N}(1-u^{1-\frac{k}N})^2\left( u^{\frac1N} [i_\gamma]_{u^{\frac1N}} - u^{-\frac2N} [i_\gamma]_{u^{-\frac1N}} \right) \right)  \\
  & \times \exp\Bigg(  \frac{2g^2(1-u^{\frac1N})^2}{(1-u^2)} \sum_{\substack{0=\alpha<\beta \\ \beta - \alpha \equiv 1 \pmod{2}  }}^{\lambda} \bigg[ u^{\frac{\alpha}N} [i_{\alpha+1}-i_{\alpha}]_{u^{\frac1N}}
  \left( u^{\frac{2+i_{\beta}}N} [i_{\beta+1}-i_{\beta}]_{u^{\frac1N}} - u^{2-\frac{1+i_{\beta}}{N}} [i_{\beta+1}-i_{\beta}]_{u^{-\frac1N} } \right)  \\
  & \qquad \qquad \qquad \qquad \qquad \qquad \qquad  -  u^{-\frac{\alpha}N}[i_{\alpha+1}-i_{\alpha}]_{u^{-\frac1N}} \left( u^{\frac{i_{\beta}-1}N} [i_{\beta+1}-i_{\beta}]_{u^{\frac1N}} - u^{2 -\frac{4+i_\beta}{N}} [i_{\beta+1}-i_{\beta}]_{u^{\frac{-1}N}} \right)  \bigg]  \Bigg).
\end{align*}

Next, for \(1\leq \gamma \leq \lambda \), we immediately see that
\begin{align*}
  u^{2-\frac2N} [i_\gamma]_{u^{-\frac1N}} - u^{\frac1N} [i_\gamma]_{u^{\frac1N}} &=  - \frac{u^{2-\frac1N}}{1-u^{\frac1N}} (1-u^{-2 + \frac{i_\gamma+1}{N}})(1-u^{-\frac{i_\gamma}N}) \\
  u^{\frac1N} [i_\gamma]_{u^{\frac1N}} - u^{-\frac2N}[i_\gamma]_{u^{-\frac1N}} &= \frac{u^{-1/N}}{1-u^{\frac1N}}(1-u^{\frac{-i_\gamma}N})(1-u^{\frac{i_\gamma+2}N}),
\end{align*}
and similarly, for \( 1 \leq \beta \leq \lambda \), we have
\begin{align*}
   u^{\frac{2+i_{\beta}}N} [i_{\beta+1}-i_{\beta}]_{u^{\frac1N}} - u^{2-\frac{1+i_{\beta}}{N}} [i_{\beta+1}-i_{\beta}]_{u^{-\frac1N} }  &= - \frac{u^{\frac{2+i_{\beta+1}}{N}}(1-u^{\frac{i_\beta - i_{\beta+1}}{N}})(1-u^{2-\frac{i_\beta + i_{\beta+1}+2}{N}})}{1-u^{\frac1N}}  \\
  u^{\frac{i_{\beta}-1}N} [i_{\beta+1}-i_{\beta}]_{u^{\frac1N}} - u^{2 -\frac{4+i_\beta}{N}} [i_{\beta+1}-i_{\beta}]_{u^{\frac{-1}N}} &= - \frac{u^{\frac{i_{\beta+1}-1}N}(1-u^{\frac{i_\beta - i_{\beta+1}}{N}})(1-u^{2 - \frac{i_\beta + i_{\beta+1}}{N}+2 })}{1-u^{\frac1N}}.
\end{align*}

For \(\lambda \geq 1 \), define the function
\begin{align*}
    f^{(\eta)}_\lambda(&z_1,z_2,\cdots,z_\lambda;u^{\frac1N}) =  (-1)^{\eta+1}  \frac{2\sqrt{2} g}{1-u^{2}} \sum_{\gamma=1}^{\lambda} (-1)^{\gamma-1} \Big[ x u^2 (1-u^{-2+\frac{z_{\lambda+1-\gamma}}N})(1-u^{-\frac{z_{\lambda+1-\gamma}}N}) \\
              & \qquad \qquad \qquad \qquad \qquad \qquad \qquad \qquad \qquad \qquad \qquad \qquad - y u (1-u^{\frac{z_{\lambda+1-\gamma}}N})(1-u^{-\frac{z_{\lambda+1-\gamma}}N})  \Big] \\
              & \qquad - \frac{2 g^2 }{1-u^2} u^{\frac{k}N}(1-u^{1-\frac{k}N})^2 \sum_{\gamma=1}^{\lambda} (-1)^{\gamma-1} (1-u^{\frac{z_{\lambda+1-\gamma}}N})(1-u^{-\frac{z_{\lambda+1-\gamma}}N}) \\
              & \qquad- \frac{2 g^2 }{1-u^2}\sum_{\substack{0\leq\alpha<\beta \\ \beta - \alpha \equiv 1 \pmod{2}  }}^{\lambda}  u^{\frac{z_{\beta+1}-z_{\alpha}}{N}}(1-u^{2 - \frac{z_{\beta+1}+z_{\beta}}N })(1-u^{\frac{z_\beta - z_{\beta+1}}N})(1-u^{\frac{z_\alpha - z_{\alpha+1}}{N}})(1-u^{\frac{z_\alpha+z_{\alpha+1}}N}),
\end{align*}
where as before, we set \(z_0 := 0 \) and  \(z_{\lambda+1} := k-2 \). Notice that for fixed \( \lambda\),  \(f_\lambda^{(\eta)}(\bm{z};u^{\frac1N})\) is a smooth function on \( z_i\), with \(i=1,2,\cdots,\lambda \), for any \( u \in (0,1) \).

With this notation, equation \eqref{eq:sumhp} is given by
\begin{align*} 
     \sum_{\substack{\rho \in \btZ{k-3} \\ |\rho| = \lambda  } }   H_{\eta}^{(k,N)}(x,y,u^{\frac1N},\rho) P^{(k,N)}(u^{\frac1N},\rho) = \sum_{i_{1}< i_{2}<\cdots<i_{\lambda}}^{k-3}  e^{{f_\lambda^{(\eta)}(i_1,i_2,\cdots,i_\lambda;u^{\frac1N})}+ O\left(\tfrac1N\right)}.
\end{align*}

We are now in the position to write this sum as an iterated integral and a residual term with explicit order. We first describe the behavior of the function with respect to \(u \in (0,1)\). We note that since the terms of order $O(\tfrac1N)$ ultimately vanish due to the limit involved in the final computation of the heat kernel, from this point we omit them to improve the clarity of the exposition.

\begin{lem}\label{lem:uniform_bdd}
  Let \(\lambda \geq 1\) be fixed. The (real valued) function $e^{f_\lambda(\bm{z},u^{\frac1N})}$, where
  \(\bm{z}= \{z_1,z_2,\cdots,z_\lambda \}\), is uniformly bounded with respect to $u\, (0<u<1)$ for
  $0\leq z_1 \leq z_2 \leq \cdots \leq z_\lambda \leq k-3$. 
\end{lem}
\begin{proof}
  It is enough to observe the behavior when $u\to 0$ and  $u \to 1$.  Clearly, when $u\to 1$ there is a limit for
  $f_\lambda(\bm{z},u^{\frac1N})$ which is bounded for any $0\leq z_1 \leq z_2 \leq \cdots \leq z_\lambda \leq k-3$.

  When $u=e^{-t}$ approaches $0$, let us observe the leading contribution in $f_\lambda(\bm{z},u^{\frac1N})$.
  Let \(1 \leq  j \leq \lambda \), then it is easy to see the leading part in $f_\lambda(z,u^{\frac1N})$ as $u\to 0^{+}$ of the term involving
  \(x, y \) and \(z_j\) in the first sum is given by
  \[
    (-1)^{\eta+\gamma-1} u^{-\frac{z_j}{N}}(2\sqrt2 g)(xu^2 - y u).
  \]
  Next, we easily see that the limit of the second sum as $u\to 0^{+}$ is either \(-2g^2\) or \(0\), according to
  $\lambda\equiv 1 \pmod{2}$ or $\lambda\equiv 0 \pmod{2}$, respectively.
  For \(0 \leq \alpha < \beta\) with \( \beta - \alpha \equiv 1 \pmod{2} \), since $0\leq z_{\alpha}\leq z_{\alpha+1}$ and $z_\beta \leq z_{\beta+1}$, the leading part
  in $f_\lambda(z,u^{\frac1N})$ as $u \to 0^{+}$ of the term involving \( g^2 \), \(\alpha \) and \(\beta \) in the last sum is given by
  \[
    - 2 g^2  u^2 u^{-\frac{z_{\alpha+1}}{N}} u^{-{\frac{z_\beta + z_{\beta+1}}{N}}}.
  \]
  Summing up, the leading part of $f_\lambda(z,u^{\frac1N})$ as $u\to 0^{+}$ is given by
  \[
   -2g^2 c -\sum_{j=1}^{\lambda} u^{-\frac{z_j}{N}}\left( (-1)^{\eta+\gamma}(2\sqrt2 g)(xu^2 - y u) 
      + 2 g^2 u^2 \sum_{\substack{ j-1< \beta \\ \beta \equiv j \pmod{2} }}^{\lambda}  u^{-{\frac{z_\beta + z_{\beta+1}}{N}}} \right),
  \]
  for  a constant \(c \in \{0,1\}\). It follows that $e^{f_\lambda(\bm{z},u^{\frac1N})}$ is bounded as $u\to 0^{+}$. 
\end{proof}

In order to deal with the multiple summation over the \(i_1,i_2,\cdots,i_\lambda \), we need the following simple lemma.

\begin{lem}
  \label{lem:sum12}
  For fixed \(\lambda \geq 1\) and \(a \in \Z_{\geq 0} \) with \(a \leq N\), we have
  \begin{equation*}
    \sum_{1\leq i_{1}< i_{2}<\cdots<i_{\lambda}}^{a} e^{{f_\lambda^{(\eta)}(i_1,i_2,\cdots,i_\lambda;u^{\frac1N})}} \quad = \sum_{0\leq i_{1}\leq i_{2}\leq\cdots \leq i_{\lambda}}^{a} e^{{f_\lambda^{(\eta)}(i_1,i_2,\cdots,i_\lambda;u^{\frac1N})}} + O(a^{\lambda-1}).
  \end{equation*}
\end{lem}

\begin{proof}
  Since \(\exp\left({f_\lambda^{(\eta)}(i_1,i_2,\cdots,i_\lambda;u^{\frac1N})} \right)\) is uniformly bounded for \(0 \leq i_j \leq a \)
  and \(0 \leq u \leq 1 \) (this is verified in the same way as Lemma \ref{lem:uniform_bdd}), we see that the difference
  between the number of summands of the two sums is given by
  \[
    \binom{a+\lambda}{\lambda} - \binom{a}{\lambda} = O(a^{\lambda-1}).
  \]
\end{proof}

Finally, we transform the sum into integrals using Riemann-Stieltjes integration. We start by considering the case \(\lambda = 1 \) as it constitutes the basis
for the proof of the general case.

\begin{prop}
  \label{prop:sumintL1}
  Let \(a \in \Z_{\geq 0} \) with \(a \leq N\). We have
  \[
    \sum_{i=0}^{a} e^{{f_\lambda^{(\eta)}(i;u^{\frac1N})}} = \int_{0}^{a} e^{{f_\lambda^{(\eta)}(z;u^{\frac1N})}} d z + O\left(\frac1N\right).
  \]
\end{prop}

\begin{proof}

  We write the sum as a Riemann-Stieltjes integral in the standard way
  \[
    \sum_{i=0}^{a} e^{f_1(i,u^{\frac1N})} = \int_0^{a} e^{f_1(z,u^{\frac1N})} d \, \Xi(z),
  \]
  where \(\Xi(z) =\sum_{1 \leq n < z} 1 = z- \{z\}\).  By partial integration, we see that 
  \begin{align*} 
    \sum_{i=1}^{a} e^{f_1(i,u^{\frac1N})} &=  \int_{0}^{a} e^{f_1(z,u^{\frac1N})}  d z + \int_{0}^{a} \{z\} f_1'(z,u^{\frac1N}) e^{f_1(z,u^{\frac1N})}  d z  \nonumber \\
                                   &= \int_{0}^{a} e^{f_1(z,u^{\frac1N})} + 2 \sum_{n=1}^{\infty} \int_{0}^{a} \cos(2 \pi n z ) e^{f_1(z,u^{\frac1N})} d z + O(1),
  \end{align*}
  the last equality is obtained by using the Fourier series of \(\psi(x) = x - [x] - \frac12 \), that is \( \psi(x) = - \sum_{n=1}^{\infty}\frac{sin(2 n \pi x)}{n \pi} \)
  (see also \cite{Ivic1985}, equation (A26)).
  
  Setting
  \[
    g(z)=\sum_{j=0}^{a-1}e^{f_1(z+j,u^{\frac1N})},
  \]
  we have
  \[
    \sum_{n=1}^{\infty} \int_{0}^{a} \cos(2 \pi n z ) e^{f_1(z,u^{\frac1N})} d z
    =\sum_{n=1}^{\infty} \int_{0}^{1} \cos(2 \pi n z ) g(z) d z.
  \]
  Now, integration by parts twice yields
  \[
    \int_{0}^{1} \cos(2 \pi n z ) g(z) d z
    = \frac1{4\pi^2n^2} (g'(1)-g'(0))- \frac1{4\pi^2n^2}\int_{0}^{1} \cos(2 \pi n z ) g''(z) d z.
  \]
  Hence 
  \begin{align*}
    \big|\sum_{n=1}^{\infty} \int_{0}^{1} \cos(2 \pi n z ) g(z) d z\big| 
    \leq &
        \sum_{n=1}^{\infty} \frac1{4\pi^2n^2} \Big[
        |g'(1)-g'(0)|+\big|\int_{0}^{1} \cos(2 \pi n z ) g''(z) dz \big| \Big] \\
    \leq & \frac1{4\pi^2} \zeta(2) \Big[
        |g'(1)-g'(0)|+\int_0^1 |g''(z)|dz \Big],
  \end{align*}
  where \(\zeta(s)\) is the Riemann zeta function.

  Next, since 
  \[
    g'(z)= \sum_{j=0}^{a-1} f_1'(z+j,u^{\frac1N})e^{f_1(z+j,u^{\frac1N})},
  \]
  we have 
  \begin{align*}
    g'(1)-g'(0) &= \sum_{j=0}^{a-1} f_1'(1+j,u^{\frac1N})e^{f_1(1+j,u^{\frac1N})} -  \sum_{j=0}^{a-1} f_1'(j,u^{\frac1N})e^{f_1(j,u^{\frac1N})}\\
                & = f_1'(a,u^{\frac1N})e^{f_1(a,u^{\frac1N})}- f_1'(0,u^{\frac1N})e^{f_1(0,u^{\frac1N})}.  
  \end{align*}
  Noticing that the summation on $j$ (over \(a\)) disappear and that $\frac{d}{dz} u^{\pm \frac{z}{N}}=\pm\frac1N (\log u) u^{\pm \frac{z}{N}}$,
  we immediately observe that $g'(1)-g'(0) = O(\frac1N)$.  Furthermore, 
  \[
    g''(z)= \sum_{j=0}^{a-1}\{f_1''(z+j,u^{\frac1N})+(f_1'(z+j,u^{\frac1N}))^2\}e^{f_1(z+j,u^{\frac1N})}.
  \]
  By Lemma \ref{lem:uniform_bdd}, there is a positive constant $C$ such that 
  \[
    |g''(z)| \leq C \sum_{j=0}^{a-1}\{|f_1''(z+j,u^{\frac1N})|+(f_1'(z+j,u^{\frac1N}))^2\}.
  \]
  Since again $\frac{d}{dz} u^{\pm \frac{z}{N}}=\pm\frac1N (\log u) u^{\pm \frac{z}{N}}$, there are positive uniform constants $A$ and $B$
   with respect to $u$ such that 
  \[
    |f_1''(z,u^{\frac1N})|\leq \log(u)^2\frac{A}{N^2}, \quad |f_1'(z,u^{\frac1N})| \leq \log(u)^2 \frac{B}{N}.
  \]
  It follows that 
  \[
    |g''(z)| \leq C(A+B^2) \log(u)^2 \frac{a}{N^2}.
  \]
  Therefore we have 
  \[
    \left|\sum_{n=1}^{\infty} \int_{0}^{a} \cos(2 \pi n z ) e^{f_1(z,u^{\frac1N})} d z \right| = O\left(\frac1N\right).
  \]
\end{proof}

\begin{lem}
  \label{lem:sumint}
  For fixed \(\lambda \geq 1\) and \(a \in \Z_{\geq 1}\) with \(a \leq N \), we have
  \begin{align*}
    &  \sum_{1 \leq i_{1} < i_{2} < \cdots < i_{\lambda}}^{a} e^{{f_\lambda^{(\eta)}(i_1,i_2,\cdots,i_\lambda;u^{\frac1N})}} =   \int_{0}^a \int_0^{z_\lambda} \cdots \int_0^{z_{2}} e^{{f_\lambda^{(\eta)}(z_1,z_2,\cdots,z_\lambda;u^{\frac1N})}}  d \bm{z} + O(a^{\lambda-1}).
  \end{align*}
\end{lem}

\begin{proof}
  The proof is by induction. The case \(\lambda = 1\) is given by Lemma \ref{lem:sum12} and Proposition
  \ref{prop:sumintL1}.
  Suppose the result holds for some \(\lambda-1\geq 1\). Then, by Lemma \ref{lem:sum12}, the sum in the left-hand side
  is, up to a factor of order \(O(a^{\lambda-1}) \), given by
  \begin{align*}
    \sum_{0\leq i_{1}\leq i_{2}\leq\cdots \leq i_{\lambda}}^{a} e^{{f_\lambda^{(\eta)}(i_1,i_2,\cdots,i_\lambda;u^{\frac1N})}}
    &= \sum_{i_\lambda=0}^{a} \left( \int_0^{i_\lambda}\int_0^{z_{\lambda-1}} \cdots \int_0^{z_2} e^{{f_\lambda^{(\eta)}(z_1,z_2,\cdots,z_{\lambda-1},i_\lambda;u^{\frac1N})}} d \bm{z} + R_{\lambda-1}^{(i_{\lambda})}(\bm{z}) \right)
  \end{align*}
  where the equality is obtained by applying the induction hypothesis with \(a = i_{\lambda} \) for each \(i_\lambda\).
  The residual terms are of order
  \[
    R_{\lambda-1}^{(i_{\lambda})}(\bm{z}) = O\left(i_{\lambda}^{\lambda-2}\right) = O\left(a^{\lambda-2}\right),
  \]
  and thus
  \[
    \sum_{i_\lambda=0}^{a} R_{\lambda-1}^{(i_{\lambda})}(\bm{z}) = O\left(a^{\lambda-1}\right).
  \]
  On the other hand, we observe as in the case of \(\lambda=1\) that
  \begin{align*}
    \sum_{i_\lambda=0}^{a}  \int_0^{i_\lambda}\int_0^{z_{\lambda-1}} \cdots \int_0^{z_2} &e^{{f_\lambda^{(\eta)}(z_1,z_2,\cdots,z_{\lambda-1},i_\lambda;u^{\frac1N})}} d \bm{z} = \int_{0}^a \int_0^{z_\lambda} \cdots \int_0^{z_{2}} e^{{f_\lambda^{(\eta)}(z_1,z_2,\cdots,z_\lambda;u^{\frac1N})}}  d \bm{z} \\
                                                      & \qquad + 2 \sum_{n=0}^{\infty} \int_{0}^a \left[ \int_0^{z_\lambda} \cdots \int_0^{z_{2}} e^{{f_\lambda^{(\eta)}(z_1,z_2,\cdots,z_\lambda;u^{\frac1N})}} d \bm{z} \right] \cos(2 \pi n z_\lambda) d z_{\lambda}.
  \end{align*}
  It remains to show that
  \[
    \sum_{n=0}^{\infty} \int_{0}^a \left[ \int_0^{z_\lambda} \cdots \int_0^{z_{2}} e^{{f_\lambda^{(\eta)}(z_1,z_2,\cdots,z_\lambda;u^{\frac1N})}} d \bm{z} \right] \cos(2 \pi n z_\lambda) d z_{\lambda} = O(a^{\lambda-1}),
  \]
  the proof follows in the same way as that of that of Proposition \ref{prop:sumintL1} by setting
  \[
    g(z_\lambda) = \sum_{j=0}^{a-1}  \int_0^{z_\lambda+j} \cdots \int_0^{z_{2}} e^{{f_\lambda^{(\eta)}(z_1,z_2,\cdots,z_\lambda+j;u^{\frac1N})}} d \bm{z},
  \]
  and noticing, by Leibniz's rule, that
  \[
    g'(1) - g'(0) = O(a^{\lambda-1}), \qquad g''(z_\lambda) = O(a^{\lambda-3}).
  \]
\end{proof}

\section{Heat kernel of the QRM} \label{sec:limit}

In this section we complete the derivation of the analytical expression of the heat kernel and the partition function of the QRM.
In addition, we give the heat kernel and partition function for each of the parities of the QRM, and,
as an application we describe the spectral determinant of the parity Hamiltonians in terms of the $G$-function.

Recall from \S \ref{sec:matrixG} that the expression of the heat kernel \(\KRabi(x,y,t)\) is the sum of two limits multiplied by a factor \(K_0(x,y,u)\).
The first limit is given by
\begin{equation}
  \label{eq:HKlimit1}
   \frac12 \lim_{N \to \infty}  \left(\frac{1+u^{\frac{2\Delta}{N}}}{2 u^{\frac{\Delta}{N}}} \right)^{N-1} \left( J^{(1,N)}_0(x,y,u^{\frac{1}{N}}) \matrixZZ
    +   J^{(1,N)}_1(x,y,u^{\frac{1}{N}}) \matrixOO  \right),
\end{equation}
while the second limit, by the results of  \S \ref{sec:four-transf}  is given by
\begin{align} \label{eq:HKlimit2}
  &\lim_{N \to \infty}\left( \frac{1-u^{\frac{2\Delta}{N}}}{2 u^{\frac{\Delta}{N}}} \right) \sum_{k \geq 3}^N  \Bigg[ \frac12 J_0^{(k,N)}(x,y,u^{\frac1N}) Q_0^{(k,N)}(x,y,u^{\frac1N}) 
    \Bigg\{
    \sum_{\rho \in \btZ{k-3}} 
    \begin{bmatrix}
      (-1)^{|\rho|+1} u^{\frac{\Delta}{N}} &  (-1)^{|\rho|} u^{\frac{\Delta}{N}} \\
       -u^{-\frac{\Delta}{N}} &  u^{-\frac{\Delta}{N}}
    \end{bmatrix}
    \left(\frac{1-u^{\frac{2\Delta}N}}{1+u^{\frac{2 \Delta}N}}\right)^{|\rho|}  \nonumber \\
  &\qquad \qquad \qquad \qquad \qquad \qquad \qquad \qquad \qquad \qquad \qquad \qquad \times H_{1}^{(k,N)}(x,y,u^{\frac1N},\rho) P^{(k,N)}(u^{\frac1N},\rho) \Bigg\}   \nonumber  \\
  &\qquad \qquad \qquad \qquad + \frac12 J_1^{(k,N)}(x,y,u^{\frac1N}) Q_1^{(k,N)}(x,y,u^{\frac1N}) 
    \Bigg\{
    \sum_{\rho \in \btZ{k-3}}
    \begin{bmatrix}
      (-1)^{|\rho|+1} u^{\frac{\Delta}{N}} & (-1)^{|\rho|+1} u^{\frac{\Delta}{N}} \\
       u^{-\frac{\Delta}{N}} &  u^{-\frac{\Delta}{N}}
    \end{bmatrix}
    \left(\frac{1-u^{\frac{2\Delta}N}}{1+u^{\frac{2 \Delta}N}}\right)^{|\rho|}    \nonumber  \\
& \qquad \qquad \qquad \qquad \qquad \qquad \qquad \qquad \qquad \qquad \qquad \qquad \times H_{0}^{(k,N)}(x,y,u^{\frac1N},\rho) P^{(k,N)}(u^{\frac1N},\rho) \Bigg\} \Bigg].
\end{align}
where the functions $H_{\mu}^{(k,N)}$, $P^{(k,N)}$ and $Q_0^{(k,N)}$ are as in Definition \ref{dfn:HPQ}.

By expanding the geometric series in \(J^{(1,N)}_i\) for \(i=0,1 \), it is easy to verify that the limit \eqref{eq:HKlimit1} is equal to
\begin{align} \label{eq:HK0}
  e^{-2g^2 \frac{1-e^{-t}}{1+e^{-t}}}
                     \begin{bmatrix} \cosh & -\sinh
                       \\ -\sinh & \cosh
                     \end{bmatrix}
                                   \Big( \sqrt2 g(x+y)\frac{1-e^{-t}}{1+e^{-t}}\Big).
\end{align}

Next, we turn our attention to the limit \eqref{eq:HKlimit2}. First, we notice that the matrix factor appearing in the
sums is fixed for all \(\rho \) with \(|\rho| \equiv i \pmod{2} \), with \(i =0,1\). Thus, by partitioning the sum appearing in
\eqref{eq:HKlimit2} according to the norm \(\lambda\) of the vectors \(\rho\) and omitting the matrix factor for now, we obtain a sum of the type 
\begin{align} \label{eq:InfSum1}
  \lim_{N \to \infty}\left( \frac{1-u^{\frac{2\Delta}{N}}}{2 u^{\frac{\Delta}{N}}} \right)  \sum_{ \substack{\lambda \equiv i \\ \pmod{2 }}}^N
  \sum_{k = \lambda + 3}^N  & J_\eta^{(k,N)}(x,y,u^{\frac1N})Q_\eta^{(k,N)}(x,y,u^{\frac1N}) \nonumber \\
  &\times \Bigg\{ \sum_{\substack{\rho \in \btZ{k-3} \\ |\rho| = \lambda  } } \left(\frac{1-u^{\frac{2\Delta}N}}{1+u^{\frac{2 \Delta}N}}\right)^{\lambda} H_{1-\eta}^{(k,N)}(x,y,u^{\frac1N},\rho) P^{(k,N)}(u^{\frac1N},\rho) \Bigg\},
\end{align}
with \(i,\eta \in \{0,1\} \). Notice that if \( \lambda > N \), we have
\[
  \sum_{\substack{\rho \in \btZ{k-3} \\ |\rho| = \lambda  } } \left(\frac{1-u^{\frac{2\Delta}N}}{1+u^{\frac{2 \Delta}N}}\right)^{\lambda} H_{1-\eta}^{(k,N)}(x,y,u^{\frac1N},\rho) P^{(k,N)}(u^{\frac1N},\rho)  = 0,
\]
whence, \eqref{eq:InfSum1} is equal to
\begin{align*}
  \lim_{N \to \infty}\left( \frac{1-u^{\frac{2\Delta}{N}}}{2 u^{\frac{\Delta}{N}}} \right) \sum_{ \substack{\lambda \equiv i \\ \pmod{2 }}}^{\infty}
  \sum_{k = \lambda + 3}^N  & J_{\eta}^{(k,N)}(x,y,u^{\frac1N})Q_{\eta}^{(k,N)}(x,y,u^{\frac1N})  \left(\frac{1-u^{\frac{2\Delta}N}}{1+u^{\frac{2 \Delta}N}}\right)^{\lambda} \\
  &\times\Bigg\{ \sum_{\substack{\rho \in \btZ{k-3} \\ |\rho| = \lambda  } }  H_{1-\eta}^{(k,N)}(x,y,u^{\frac1N},\rho) P^{(k,N)}(u^{\frac1N},\rho) \Bigg\},
\end{align*}
and since the \(H_{\mu}^{(k,N)}(x,y,u^{\tfrac1N},\rho) P^{(k,N)}(u^{\tfrac1N},\rho) \) is  uniformly bounded (cf. the discussion at the beginning of \S \ref{sec:spvl}), the dominated convergence theorem shows that the above expression is equal to 
\begin{align}
  \label{eq:infsum2}
  \sum_{ \substack{\lambda \equiv i \\ \pmod{2 }}}^{\infty} \lim_{N \to \infty}\left( \frac{1-u^{\frac{2\Delta}{N}}}{2 u^{\frac{\Delta}{N}}} \right)
  \sum_{k = \lambda + 3}^N   &J_\eta^{(k,N)}(x,y,u^{\frac1N})Q_\eta^{(k,N)}(x,y,u^{\frac1N}) \left(\frac{1-u^{\frac{2\Delta}N}}{1+u^{\frac{2 \Delta}N}}\right)^{\lambda} \nonumber \\
  &\times \Bigg\{ \sum_{\substack{\rho \in \btZ{k-3} \\ |\rho| = \lambda  } }  H_{1-\eta}^{(k,N)}(x,y,u^{\frac1N},\rho) P^{(k,N)}(u^{\frac1N},\rho) \Bigg\}.
\end{align}
  
Thus, the limit \eqref{eq:HKlimit2} may be computed termwise for each value of \( \lambda  \geq 0 \).

The innermost sum in \eqref{eq:infsum2} is computed as an iterated integral by the results
of \S \ref{sec:transf-summ-into} and the next lemma gives the explicit computation of
\(J_{\mu}^{(k,N)}(x,y,u^{ \frac{1}N })Q_{\mu}^{(k,N)}(x,y,u^{ \frac{1}N }) \).
\begin{lem} \label{lem:Jbar}
  For \(\mu=0,1\), we have
  \begin{align*}
    &J_{\mu}^{(k,N)}(x,y,e^{-\frac{t}N }) Q_\mu^{(k,N)}(x,y,e^{-\frac{t}N }) \\
     &\quad= \exp\Big( (-1)^\mu  \frac{2\sqrt2 g e^{-t} }{1-e^{-2t}}\left( x (e^{-\frac{tk}N+1}+e^{-t + \frac{2 tk}N }) - y (e^{-\frac{tk}N }+e^{\frac{ t k}N }) - \sqrt2 g \frac{1+e^{-t}}{1-e^{-t}}(x-y) \right) \Big)\\
     &\quad \, \times \exp\Bigg( - 4g^2 \frac{1+e^{-2 t}}{1 - e^{-2 t}} + 2 g^2 \frac{e^{-\frac{t k}{N}}(1+ e^{-t+\frac{2 tk}{N}})}{1-e^{- t}} + 2 g^2\frac{e^{-\frac{t k }N}(1-e^{- t + \frac{ t k}{N}})(1-e^{-\frac{ t k}{N}})(1+e^{-t+\frac{2 t k}{N}})}{1-e^{-2 t}}    +O\left(\frac1N\right) \Bigg).
  \end{align*}
\end{lem}

\begin{proof}
  Direct evaluation of the geometric series using the identity
  \[
    (1-u^{\frac1N})(N-k-1)= \left[\frac{t}N + O\left(\frac1{N^2}\right)\right](N-k-1) = t(1-\frac{k}N) +O\left(\frac1N \right) \quad (N\to \infty).
  \]
  gives
  \begin{align*}
    J_\mu^{(k,N)}(x,y,u^{\frac1N})  &= \exp\Big(\frac{(-1)^\mu\sqrt2 g}{1-u^2} (1-u^{1-\frac{k}N}) \big[x u^{\frac{k}N} (1-u^{1-\frac{k}N} ) +y(1-u^{1+\frac{k}N})\big]\Big)\\
    & \times  \exp\Bigg( - g^2\frac{t k}{N} - \frac{2 g^2 (1-u^{1-\frac{k}{N}})(1-u^{1+\frac{k}{N}})}{1-u^2} + \frac{g^2 (1-u^{2-\frac{2 k}{N}})(1-u^{\frac{2 k}N})}{2(1-u^2)}  + O\left(\frac1N\right) \Bigg) .
  \end{align*}
  Then, the result is obtained multiplying by \(Q_{\mu}^{(k,N)}(u)\) and setting $u=e^{-t}$.

\end{proof} 

With these preparations, we proceed to the computation of the limit \eqref{eq:HKlimit2}, thus giving the
analytic formula for the heat kernel of the QRM.

\subsection{Analytical formula of the heat kernel and partition function} \label{sec:analyt-form-heal}

Finally, in this subsection we present the main results of this paper, namely, the analytical expressions for the heat kernel and the partition function of the QRM.
In the next theorem, for \(\lambda = 0 \), we employ the notation
\[
  \idotsint\limits_{0\leq \mu_1 \leq \cdots \leq \mu_\lambda \leq 1} f(x) d \bm{\mu_0} = f(x),
\]
for any function \(f\).


\begin{thm} \label{thm:heat_kernel}
  The heat kernel $\KRabi(x,y,t) $ of the QRM is given by the uniformly convergent series
  \begin{align*}    
    \KRabi(&x,y,t;g,\Delta) =  K_0(x,y,t;g) \Bigg[ \sum_{\lambda=0}^{\infty} (t\Delta)^{\lambda} e^{-2g^2 (\coth(\tfrac{t}2))^{(-1)^\lambda}}
    \\
    &\quad \times \idotsint\limits_{0\leq \mu_1 \leq \cdots \leq \mu_\lambda \leq 1}  e^{4g^2 \frac{\cosh(t(1-\mu_\lambda))}{\sinh(t)}(\frac{1+(-1)^\lambda}{2}) + \xi_{\lambda}(\bm{\mu_{\lambda}},t;g)}  \cdot
          \begin{bmatrix}
            (-1)^{\lambda} \cosh  &  (-1)^{\lambda+1} \sinh  \\
            -\sinh &  \cosh
          \end{bmatrix}
                     \left( \theta_{\lambda}(x,y,\bm{\mu_{\lambda}},t;g) \right) d \bm{\mu_{\lambda}} \Bigg],
  \end{align*}
  with \(\bm{\mu_0} := 0\) and \(\bm{\mu_{\lambda}}= (\mu_1,\mu_2,\cdots,\mu_\lambda)\) and \(d \bm{\mu_{\lambda}} = d \mu_1 d \mu_2 \cdots d \mu_{\lambda} \) for \(\lambda \geq 1\). Here,
  $K_0(x,y,t;g)$ (cf. \eqref{eq:K0} ) is given by
  \begin{align*}
    K_0(x,y,t;g)
    & = \frac{e^{g^2t}}{\sqrt{\pi (1-e^{-2t})}} \exp\left( - \frac{1+e^{-2t}}{2(1-e^{-2t})} (x^2 + y^2) +  \frac{2 e^{-t} x y}{1-e^{-2t}} \right)\\
  \end{align*}
  and the functions \(\theta_{\lambda}(x,y, \bm{\mu_{\lambda}},t)\) and $\xi_\lambda(\bm{\mu_{\lambda}},t)$ are given by
\begin{align*} 
  \theta_{\lambda}(x,y, \bm{\mu_{\lambda}},t;g) &:= \frac{2\sqrt{2} g e^{-t}}{1-e^{-2t}}\left( x (e^{t}+e^{- t}) - 2 y \right) \left( \frac{1-(-1)^{\lambda}}{2} \right) - \sqrt{2} g (x-y) \frac{1+e^{-t}}{1-e^{-t}} \\
                           & \quad +   \frac{2\sqrt{2} g e^{-  t}}{1-e^{-2 t}} (-1)^{\lambda} \sum_{\gamma=0}^{\lambda} (-1)^{\gamma} \Big[ x  (e^{t(1 -   \mu_{\gamma}) } + e^{ t( \mu_{\gamma} - 1)})  -  y  (e^{- t \mu_{\gamma} }+ e^{ t \mu_{\gamma}})  \Big] \nonumber \\
  \xi_\lambda(\bm{\mu_{\lambda}},t;g) &:=  -\frac{2g^2 e^{-t}}{1-e^{-2t}} \left(e^{\frac12t(1-\mu_\lambda)}-e^{\frac12 t(\mu_{\lambda}-1)}\right)^2 (-1)^{\lambda}  \sum_{\gamma=0}^{\lambda} (-1)^{\gamma} (e^{- t \mu_{\gamma} }+ e^{ t \mu_{\gamma}})  \\
                     &\qquad  - \frac{2 g^2 e^{-t} }{1-e^{-2 t}} \sum_{\substack{0\leq\alpha<\beta\leq \lambda-1\\ \beta - \alpha \equiv 1 \pmod{2}  }}  \left( (e^{t(1-\mu_{\beta+1})} +  e^{t(\mu_{\beta+1}-1)} )-(e^{t(1-\mu_{\beta})} +  e^{t(\mu_{\beta}-1)}) \right) \nonumber  \\
  &\qquad \qquad \qquad \qquad \qquad \qquad \times ( (e^{t  \mu_{\alpha}} + e^{-t \mu_{\alpha}}) - (e^{t \mu_{\alpha+1}} + e^{-t \mu_{\alpha+1}})), \nonumber 
\end{align*}
where we use the convention \( \mu_0 = 0 \) whenever it appears in the formulas above.
\end{thm}

\begin{rem}
  Note that the term corresponding to \(\lambda=0\) in the series is given explicitly (see \eqref{eq:HK0} and \eqref{eq:functlambda0} below) by
  \[
    e^{-2g^2 \tanh(\tfrac{t}2)} \begin{bmatrix}
           \cosh  &  - \sinh  \\
            -\sinh &  \cosh
          \end{bmatrix}
                     \left( \sqrt2 g(x+y)\frac{1-e^{-t}}{1+e^{-t}}) \right).
  \]
\end{rem}


\begin{proof}

  For clarity, let us first define some notations to be used during the proof.

  For \(\lambda \geq 0 \), the functions \(\phi(s,t)\), \(\alpha_{\lambda}(x,y,t)\) and $\sigma_\lambda(s,t)$ are given by
  \begin{align*} 
    \phi(s,t) &:= - 4g^2 \frac{1+e^{-2 t}}{1 - e^{-2 t}} + 2 g^2 \frac{e^{-s t}(1+ e^{t(2s-1)})}{1-e^{- t}} + 2g^2\frac{e^{-st}(1-e^{t(s-1)})(1-e^{-st})(1+e^{t(2s-1)})}{1-e^{-2 t}}, \nonumber \\
    \alpha_{\lambda}(x,y,t) &:= \frac{2\sqrt{2} g e^{-t}}{1-e^{-2t}}\left( x (e^{t}+e^{- t}) - 2 y \right) \left( \frac{1-(-1)^{\lambda}}{2} \right) - \sqrt{2} g (x-y) \frac{1+e^{-t}}{1-e^{-t}},  \\
    \sigma_{\lambda}(s,t) &:=  -\frac{4g^2 e^{-t s}(1-e^{t(s-1)})^2}{1-e^{-2t}} \left( \frac{1-(-1)^\lambda}{2} \right) +
                 \frac{2g^2 e^{-t s}(1-e^{t(s-1)})^2(e^{t s}+ e^{-t s})}{1-e^{-2 t}}
  \end{align*} 
  for \( \bm{\mu_{\lambda}} =(\mu_1,\mu_2,\cdots,\mu_\lambda) \in \R^{\lambda} \) (where \(\bm{\mu_0} := 0 \in \{0\} = \R^0\) ), we define
  \begin{align*} 
    \vartheta_{\lambda}(x,y,\bm{\mu_{\lambda}},t) &:= \frac{2\sqrt{2} g e^{-  t}}{1-e^{-2 t}} (-1)^{\lambda} \sum_{\gamma=0}^{\lambda} (-1)^{\gamma} \Big[ x  (e^{t(1 -   \mu_{\gamma}) } + e^{ t( \mu_{\gamma} - 1)})  -  y  (e^{- t \mu_{\gamma} }+ e^{ t \mu_{\gamma}})  \Big]. \nonumber 
  \end{align*} 
  We note that these functions correspond to the expressions appearing inside the exponentials in
  \(J_{\mu}^{(k,N)}(x,y,u^{ \frac{1}N })\), \(Q_{\mu}^{(k,N)}(x,y,u^{ \frac{1}N })\) (see Lemma \ref{lem:Jbar}) and in the function \(f^{(\eta)}_\lambda\)
  (defined in \S \ref{sec:transf-summ-into}).

  To complete the computation of the heat kernel it remains to compute the limits in \eqref{eq:infsum2}. We consider the cases
  \(\lambda = 0 \) and \(\lambda \geq 1 \) by separate.

For \(\lambda= 0 \), the limit \eqref{eq:HKlimit2} is given by
\[
  \frac12 \lim_{N \to \infty}\left( \frac{1-u^{\frac{2\Delta}{N}}}{2 u^{\frac{\Delta}{N}}} \right)
  \sum_{k = \lambda + 3}^N   J_\eta^{(k,N)}(x,y,u^{\frac1N})Q_\eta^{(k,N)}(x,y,u^{\frac1N}) 
\]
since \( H_{1}^{(k,N)}(x,y,u^{\frac1N},\rho) = P^{(k,N)}(u^{\frac1N},\rho) = 1 \) for \(\rho = \bm{0}_n \)  with \(n \geq 1 \). By Lemma \ref{lem:Jbar}, the limit is the Riemann sum corresponding to the integral
\[
  \frac{t \Delta}2 \int_{0}^1 e^{(-1)^{\eta} \left( \alpha_1(x,y,t) + \vartheta_1(x,y,\bm{\mu_1} ,t) \right) + \phi(\mu_1,t)}  d \bm{\mu_1}.
\]

Notice that since \(\xi_1(\bm{\mu_1},t;g) + \sigma_1(\mu_1,t) = 0\), we can write
\begin{align} \label{eq:case0}
  \frac{t \Delta}2 \int_{0}^1 e^{(-1)^{\eta} \left(\alpha_1(x,y,t) + \vartheta_1(x,y,\bm{\mu_1} ,t)  \right) + \phi(\mu_1,t) + \sigma_1(\mu_1,t) + \xi_1(\bm{\mu_1},t;g) } d \bm{\mu_1}. 
\end{align}

Next, we consider the case \(\lambda \geq 1\). In this case, since \( H_{\eta}^{(k,N)}(x,y,u^{\frac1N},\rho),P^{(k,N)}(u^{\frac1N},\rho) \) are
non-vanishing, multiple iterated integrals appear in the computation. Let \(h_\lambda(x,y,t) = \frac{2\sqrt{2} g e^{-t}}{1-e^{-2t}}\left( x (e^{t}+e^{- t}) - 2 y \right) \left( \frac{(1-(-1)^{\lambda})}2\right)  \), then, by Lemma \ref{lem:sumint}, the limit \eqref{eq:HKlimit2} is given by
\begin{align*} 
  &\frac12 \lim_{N \to \infty}\left( \frac{1-u^{\frac{2\Delta}{N}}}{2 u^{\frac{\Delta}{N}}} \right)  \sum_{k = \lambda + 3}^N   J_\eta^{(k,N)}(x,y,u^{\frac1N})Q_\eta^{(k,N)}(x,y,u^{\frac1N}) \left(\frac{1-u^{\frac{2\Delta}N}}{1+u^{\frac{2 \Delta}N}}\right)^\lambda  \sum_{1\leq i_1<i_2 < \cdots < i_\lambda}^{k-3} e^{f_\lambda^{(1-\eta)}(i_1,\cdots,i_\lambda,u^{\frac1N})} \\
  = &  \frac{1}2 e^{(-1)^{\eta}  h_{\lambda+1}(x,y,t)  } \lim_{N \to \infty}\left( \frac{1-u^{\frac{2\Delta}{N}}}{2 u^{\frac{\Delta}{N}}} \right)\sum_{k = \lambda + 3}^N   J_\eta^{(k,N)}(x,y,u^{\frac1N})Q_\eta^{(k,N)}(x,y,u^{\frac1N}) \left(\frac{1-u^{\frac{2\Delta}N}}{1+u^{\frac{2 \Delta}N}}\right)^\lambda
       e^{\sigma_{\lambda+1}(\frac{k}N,t) }  \\
  & \qquad\qquad\quad \times \int_{0}^{k-3}\int_0^{\mu_\lambda} \cdots \int_0^{\mu_{2}}  e^{ (-1)^{\eta} \vartheta_{\lambda+1}(\bm{\nu},t) + \xi_{\lambda+1}(\bm{\nu} ,t;g)} d \bm{\mu_{\lambda}}
\end{align*}
with \( \bm{\nu}  = (\frac{t}N \mu_1, \frac{t}N \mu_2,\ldots,\frac{t}N \mu_{\lambda},\frac{t k}N)\) and \(d \bm{\mu_\lambda} = d \mu_1 d \mu_2 \cdots d \mu_\lambda \). The change of variable  \(  \mu_i \mapsto (k-3) \mu_i\)
for \(i \in \{ 1,2,\cdots, \lambda\}\)  yields
\begin{align*}
  \frac{1}2 e^{(-1)^{\eta} h_{\lambda+1}(x,y,t)}  &\lim_{N \to \infty} \left( \frac{1-e^{-t \frac{2\Delta}{N}}}{2 e^{-t \frac{\Delta}{N}}} \right) \sum_{k = \lambda + 3}^N  J_\eta^{(k,N)}\left(x,y,e^{- \frac{t}N}\right) Q_\eta^{(k,N)} \left(x,y,e^{- \frac{t}N}\right)e^{\sigma_{\lambda+1}(\frac{k}N,t) }  \\
  &\qquad \times  \left(\frac{1-e^{-t \frac{2\Delta}N}}{1+e^{-t \frac{2 \Delta}N}}\right)^{\lambda}  k^{\lambda}  \idotsint\limits_{0\leq \mu_1 \leq \cdots \leq \mu_{\lambda} \leq 1} e^{ (-1)^{\eta} \vartheta_{\lambda+1}(\bm{\nu}_2,t) + \xi_{\lambda+1}( \bm{\nu}_2,t;g) }  d \bm{\mu_\lambda},
\end{align*}
  where \(\bm{\nu}_2  = (\frac{t k}N \mu_1, \frac{t k}N \mu_2,\ldots,\frac{t k}N \mu_{\lambda},\frac{t k}N)\) and where, for clarity,
  we omitted terms of order \( O(k^{\lambda-1}) \) that vanish when taking the limit.
  
  The limit is the Riemann sum corresponding to the integral
\begin{align*}
  \frac{ (t\Delta)^{\lambda+1}}{2} e^{(-1)^{\eta}  \alpha_{\lambda+1}(x,y,t)} & \int_{0}^1 (\mu_{\lambda+1})^{\lambda} e^{\phi(\mu_{\lambda+1},t) +  \sigma_{\lambda+1}(\mu_{\lambda+1},t) } \idotsint\limits_{0\leq \mu_1 \leq \cdots \leq \mu_{\lambda} \leq 1} e^{ (-1)^{\eta} \vartheta_{\lambda+1}(\bm{\nu}_3 ,t)  + \xi_{\lambda+1}(\bm{\nu}_3 ,t;g) }  d \bm{\mu_{\lambda+1}},
\end{align*}
where \(\bm{\nu}_3 = (\mu_{\lambda+1} \mu_1, \mu_{\lambda+1} \mu_2,\ldots,\mu_{\lambda+1} \mu_{\lambda},\mu_{\lambda+1})\).
Finally, the change of variable \( \mu_i \mapsto \frac{\mu_i}{\mu_{\lambda+1}}\) for  \(i \in \{ 1,2,\cdots,\lambda\}\), gives
\begin{equation}
  \label{eq:caselambda1}
  \frac{ (t\Delta)^{\lambda+1} }{2}  \idotsint\limits_{0\leq \mu_1 \leq \cdots \leq \mu_{\lambda+1} \leq 1} e^{ (-1)^{\eta} \left( \alpha_{\lambda+1}(t)+\vartheta_\lambda(\bm{\mu_{\lambda+1}},t) \right) + \phi(\mu_{\lambda+1},t) + \sigma_{\lambda+1}(\mu_{\lambda+1},t) + \xi_\lambda(\bm{\mu_{\lambda+1}},t;g)} d \bm{\mu_{\lambda+1}},
\end{equation}
with \(\bm{\mu_{\lambda+1}} = (\mu_1,\mu_2,\cdots,\mu_\lambda,\mu_{\lambda+1}) \).

 From \eqref{eq:case0} and \eqref{eq:caselambda1}, the limit \eqref{eq:HKlimit2} is given by
\begin{align*}
  \sum_{\lambda=1}^{\infty} (t\Delta)^{\lambda}& \idotsint\limits_{0\leq \mu_1 \leq \cdots \leq \mu_{\lambda} \leq 1}  e^{\phi(\mu_{\lambda},t)+  \sigma_{\lambda+1}(\mu_{\lambda+1},t)  + \xi_{\lambda}(\bm{\mu_\lambda},t;g)} 
     \begin{bmatrix}
    (-1)^{\lambda} \cosh  &  (-1)^{\lambda+1} \sinh  \\
    -\sinh &  \cosh
  \end{bmatrix}
       \left( \alpha_{\lambda}(t) + \vartheta_{\lambda}(\bm{\mu_{\lambda}},t)  \right) d \bm{\mu_{\lambda}}.
\end{align*}

Notice that for \(\lambda \ge 1 \), \(\theta_\lambda(x,y,\bm{\mu_\lambda},t;g) = \alpha_{\lambda}(t) + \vartheta_{\lambda}(\bm{\mu_{\lambda}},t)  \) and
\[
  \phi(s,t) + \sigma_{\lambda}(s,t) =  -2g^2 (\coth(\tfrac{t}2))^{(-1)^\lambda} + 4g^2 \frac{\cosh(t(s-1))}{\sinh(t)} \left( \frac{1-(-1)^{\lambda}}{2} \right).
\]
Furthermore
\begin{equation}
  \label{eq:functlambda0}
  \theta_0(x,y,\bm{\mu_0},t;g) = \sqrt2 g(x+y)\frac{1-e^{-t}}{1+e^{-t}}, \qquad \qquad  \phi(0,t) + \sigma_{0}(0,t) + \xi_0(\bm{\mu_{0}},t;g) = -2g^2 \tanh(\tfrac{t}2),
\end{equation}
therefore the expression for the limit \eqref{eq:HKlimit1} can be written in a way consistent with the notation
of the limit \eqref{eq:HKlimit2}. The sum of the two limits multiplied by $K_0(x,y,g,t)$ gives the desired expression.
\end{proof}


Next, we give the explicit expression for the partition function \(\ZRabi(\beta;g,\Delta)\) of the QRM using the expression for the heat kernel of Theorem \ref{thm:heat_kernel}.

First, by Theorem \ref{thm:heat_kernel}, the trace of \(\KRabi(x,y,t;g,\Delta)\) is equal to
\begin{align*}
  2 K_0(x,y,t;g) &\Bigg\{ e^{-2g^2 \tanh(\tfrac{t}2)}\cosh\left( \theta_{0}(x,y, \bm{\mu_0},t; g) \right)  \\
  &+  e^{-2g^2 \coth(\tfrac{t}2)} \sum_{\lambda = 1}^{\infty} (t\Delta)^{2\lambda} \idotsint\limits_{0\leq \mu_1 \leq \cdots \leq \mu_{2 \lambda} \leq 1} e^{4g^2 \frac{\cosh(t(1-\mu_\lambda))}{\sinh(t)} +\xi_{2 \lambda}(\bm{\mu_{2 \lambda}},t;g)} 
                             \cosh\Big( \theta_{ 2 \lambda}(x,y, \bm{\mu_{2\lambda}},t;g)  \Big) d \bm{\mu_{2 \lambda}}   \Bigg\}.
\end{align*}

Furthermore, notice that 
\begin{align*}
  K_0(x,x,t;g) = \frac{e^{g^2t}}{\sqrt{\pi (1-e^{-2t})}} \exp\left( - \frac{1-e^{-t}}{1+e^{-t}}x^2\right)
\end{align*}
and, that for \(\lambda \equiv 0 \pmod{2}  \), we have
\begin{align*}
  \theta_{\lambda}(x,x, \bm{\mu_\lambda},t;g) = \frac{2\sqrt{2} g  x}{1 +e^{- t}} \sum_{\gamma=0}^{\lambda} (-1)^{\gamma}  \left( e^{- t \mu_{\gamma} } - e^{ t( \mu_{\gamma}- 1)}  \right)
\end{align*}
with \( \mu_0 = 0 \). Thus, we observe that $\tr \KRabi(x,x,t;g,\Delta)$ is equal to
\begin{align*}
   &\frac{2 e^{g^2t} e^{-x^2 \frac{1-e^{-t}}{1+e^{-t}} }}{\sqrt{\pi (1-e^{-2t})}} 
  \Bigg\{ e^{-2g^2\frac{1-e^{-t}}{1+e^{-t}}} \cosh\left(2\sqrt2 gx  \frac{1-e^{-t}}{1+e^{-t}}\right) \, +  e^{-2g^2 \coth(\tfrac{t}2)} \\
  &\quad \times \sum_{\lambda =1}^{\infty} (t\Delta)^{2\lambda} \idotsint\limits_{0\leq \mu_1 \leq \cdots \leq \mu_{2 \lambda} \leq 1} e^{4g^2 \frac{\cosh(t(1-\mu_\lambda))}{\sinh(t)}+\xi_{2 \lambda}(\bm{\mu_{2\lambda}},t;g)}   \cosh\left( \frac{2 \sqrt{2} g x}{1+e^{-t}} \sum_{\gamma=0}^{2\lambda} (-1)^{\gamma}  \left( e^{- t \mu_{\gamma} } - e^{ t( \mu_{\gamma}- 1)}  \right) \right) d \bm{\mu_{2\lambda}}  \Bigg\},
\end{align*}
and we proceed to give the analytical expression for the partition function.


\begin{cor} \label{cor:Partition_function}
  The partition function \( \ZRabi(\beta;g,\Delta)\) of the QRM is given by
  \begin{align*}
    \ZRabi(\beta;g,\Delta) &= \frac{2 e^{g^2\beta}}{1-e^{- \beta}} \Bigg[ 1 + e^{-2g^2 \coth(\frac{\beta}2)} \sum_{\lambda=1}^{\infty} (\beta \Delta)^{2\lambda}
     \idotsint\limits_{0\leq \mu_1 \leq \cdots \leq \mu_{2 \lambda} \leq 1}
     e^{ 4g^2\frac{\cosh(\beta(1-\mu_{2\lambda}))}{\sinh(\beta)} +  \xi_{2 \lambda}(\bm{\mu_{2\lambda}},\beta;g) +\psi^-_{2 \lambda}(\bm{\mu_{2 \lambda}},\beta;g)} d \bm{\mu_{2 \lambda}}  \Bigg],
  \end{align*}
  where the function $\psi_\lambda^{\pm}(\bm{\mu_{\lambda}},t;g)$ is given by
  \begin{equation*}
  \psi_\lambda^{\pm}(\bm{\mu_{\lambda}},t;g) :=  \frac{2 g^2 e^{-t}}{1-e^{-2 t}}\left[ \sum_{\gamma=0}^{\lambda} (-1)^{\gamma}  \left( e^{ t\left(\tfrac12 - \mu_{\gamma}\right) } \pm e^{ t\left( \mu_{\gamma}- \tfrac12\right)}  \right)  \right]^2.
\end{equation*}
for  \(\lambda \geq 1\) and \(\bm{\mu_{\lambda}} = (\mu_1,\mu_2,\cdots,\mu_\lambda) \) and where \( \mu_0 = 0 \). 
\end{cor}

\begin{proof}

  Recall that for $\alpha>0$ and $\gamma,\eta \in \R$, we have the elementary identity
  \[
    \int_{-\infty}^\infty e^{-\alpha x^2} \cosh(x \, \eta)dx = \sqrt{\frac{\pi}{\alpha}}e^{\frac{\eta^2}{4\alpha}}.
  \]
  In particular,
  \[
    \int_{-\infty}^\infty e^{-\frac{1-e^{-\beta}}{1+e^{-\beta}}x^2}\cosh\Big(2\sqrt2 g x \frac{1-e^{-\beta}}{1+e^{-\beta}}\Big)dx
    =\pi^{\frac12} \sqrt{\frac{1+e^{-\beta}}{1-e^{-\beta}}} e^{2g^2\frac{1-e^{-\beta}}{1+e^{-\beta}}}
  \]
  and, for \(\lambda \geq 1 \), we have
  \begin{align*}
    &\int_{-\infty}^\infty e^{-\frac{1-e^{-\beta}}{1+e^{-\beta}}x^2} \cosh\left( \frac{2 \sqrt{2} g x}{1+e^{-\beta}} \sum_{\gamma=0}^{\lambda} (-1)^{\gamma}  \left( e^{- \beta \mu_{\gamma} } - e^{ \beta( \mu_{\gamma}- 1)}  \right)  \right) d x \\
    &\qquad = \pi^{\frac12} \sqrt{\frac{1+e^{-\beta}}{1-e^{-\beta}}}  e^{\frac{2 g^2}{1-e^{-2 \beta}}\left[ \sum_{\gamma=0}^{\lambda} (-1)^{\gamma}  \left( e^{- \beta \mu_{\gamma} } - e^{ \beta( \mu_{\gamma}- 1)}  \right)  \right]^2} = \pi^{\frac12} \sqrt{\frac{1+e^{-\beta}}{1-e^{-\beta}}} e^{\psi^-_{{\lambda}}(\bm{\mu},s,\beta;g)}.
  \end{align*}
  The result then follows from
  \begin{equation*}
    \ZRabi(\beta;g,\Delta):=  \int_{-\infty}^\infty  \tr \KRabi(x,x,\beta;g,\Delta) dx,
  \end{equation*} 
  and the expression for \(\tr \KRabi(x,x,t;g,\Delta) \).
\end{proof}

\begin{rem}
  The unitary operator $e^{-it\HRabi}$ (associated with the Schr\"odinger equation of to \(\HRabi\)) is of fundamental importance. In our case, the operator can be obtained from $e^{-\beta H}$ with $\beta > 0$ by meromorphic continuation to imaginary $\beta$ (with a fixed branch for each $\beta \in 2\pi i \Z$). We direct the reader to \cite{RW2020z} for the details.
\end{rem} 

\subsection{Interpretation of discrete paths through the action of $\mathfrak{S}_{\infty}$ on $\Z_2^{\infty}$} 
\label{sec:remark-computations}

  In this subsection we aim to clarify the rearrangement of the sums in the equations
  leading to \eqref{eq:HKlimit2} 
  and the resultant expression of the heat kernel.
    Namely, we now revisit the discussion on the ``discrete path integrals" appearing from the Trotter-Kato
    product formula started in the Introduction.
    
    Let us briefly describe the main points of the computation. First, in \S \ref{sec:matrixG}, by dealing with the
    non-commutative terms in the expression for the $N$-th power kernel $D_N(x,y,t)$ we obtained an expression
    for the product formula that can be naively seen as a discrete path integral.
    Then we employed harmonic analysis to reformulate the sum using Fourier analysis (notably, Parseval's formula) on
    $\Z_2^{m}\; (m \geq0)$. The resulting sum allowed us to ultimately replace the uncontrollable (infinitely many)
    changes of signature with non-trivial coefficients (at $\Z_2^{\infty}$) appearing in the exponents of exponential terms of the initial summands by various hyperbolic functions.
    To complete the final step, we rearranged the infinite sums according to the norms of the elements of
    $\Z_2^{m}\; (m \geq0)$. This rearrangement is consistent with the discussion of discrete paths
    (equivalently, elements of $\Z_2^{\infty}$), as we now explain. 
  
  Recall from \S \ref{sec:four-transf} that the groups $\Z_2^{n}$  for \(n \geq 0 \) may be assumed to be
  embedded into the inductive limit \(\Z_2^{\infty} \). Next, we consider the action on $\Z_2^{\infty}$ of the infinite
  symmetric group $\mathfrak{S}_\infty$, defined by
  \[
    \mathfrak{S}_\infty := \varinjlim_n \mathfrak{S}_n,
  \]
  where, for $i \leq  j$, the injective homomorphims are given by the natural embedding (as a subgroup)
  of $\mathfrak{S}_i$ into $\mathfrak{S}_j$.

  The orbits of the action are exactly the sets
  \[
    \mathcal{O}_\lambda := \left\{ \sigma \in \Z_2^{\infty} \, : \, |\sigma|= \lambda  \right\},
  \]
  for \(\lambda \geq 0 \). Here, $|\cdot| : \Z_2^{\infty} \to \R $ is the function induced by the norms for each $\Z_2^{n}$ for
  $n \geq 0$. Canonical orbits representatives for $\mathcal{O}_\lambda\,\, (\lambda=1,2,3,\cdots)$ are given by the image of the elements
  \[
    \bm{0} \in \Z_2^{0}, \quad  (1) \in \Z_2^{1}, \quad (1,1) \in \Z_2^{2}, \quad \cdots ,\quad (1,1,\cdots,1) \in \Z_2^{n}, \cdots
  \]
  in \(\Z_2^{\infty} \). For instance,  when $\lambda=2$ we have $\mathcal{O}_2=\mathfrak{S}_\infty.[1,1]$, $[1,1]$ being the image of $(1,1) \in \Z_2^2$ in $\Z_2^{\infty}$.
  The orbit decomposition is then given by
  \begin{equation}
    \label{eq:orbitdecomposition}
    \Z_2^{\infty} = \bigsqcup_{n=0}^{\infty} \mathfrak{S}_{\infty} . [\underbrace{1,..,1}_{n}]
    \xleftrightarrow{ \text{label by $|\cdot|$}} \Z_{\geq 0}.
  \end{equation}
  and, from this point of view, the rearrangement of the sums in \eqref{eq:HKlimit2} is done according to the orbit
  decomposition of $\Z_2^{\infty}$ with respect the action of $\mathfrak{S}_\infty$ (through the orbit invariant $|\cdot|$).
  Each summand given by the iterated integral over the $\lambda$-th simplex (obtained by the computations in \S \ref{sec:transf-summ-into}) in the resulting sums is, by virtue of Lemma \ref{lem:bij1}, shown to be an orbit integral
  $\mathcal{O}_\lambda$.

  As discussed in the Introduction, we might also interpret the elements of the groups $\Z_2^{n}$ for $n \geq 0$ as paths
  between two points alternating between two states (represented in Figure \ref{fig:paths} by ``+'' and ``-'').
  In this interpretation, the rearranging of the sum \eqref{eq:HKlimit2} according to the norm $\lambda$ corresponds to
  grouping paths according to the number of times that the path is in the ``+'' state as shown in Figure
  \ref{fig:paths2} (compare with \eqref{eq:orbitdecomposition} above). 
  Therefore, we might say that the sum over the paths in $\Z_2^{n}$ arising from the Trotter-Kato product formula is
  ultimately reduced to a sum over points (labeled by $\Z_{\geq 0}$) which is then computed in an elementary way with
  the method described in this paper.

  To summarize, the infinite series of the resultant expression of the heat kernel is considered as a sum over the orbits $\mathcal{O}_\lambda= \mathfrak{S}_\infty.\rho$ ($\rho\in \Z_2^{\infty}$ with $|\rho|=\lambda\in \Z_{\geq0}$), and each summand, given by an integral over the $\lambda$-th simplex, can be regarded as the integral over the fundamental domain of $\mathfrak{S}_\lambda (\subset \mathfrak{S}_\infty)$ acting on the orbit $\mathcal{O}_\lambda$
in $\Z_2^\infty$ by looking at the formula in Lemma \ref{lem:sumint} and Lemma \ref{lem:sumexp}.

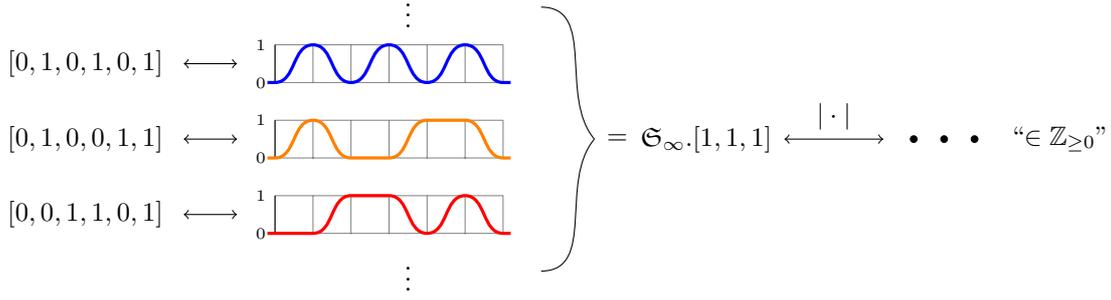
\begin{figure}[!ht]
  \centering
  \begin{tikzpicture}[domain=0:4]

    
    \node at (1.75,3) {$\vdots$};
    \node at (1.75,-0.5) {$\vdots$};
    
    \foreach \y in {0,1,2}
    {
      \draw[step=0.5,very thin,color=gray] (0,\y -0.01) grid (3,\y + 0.5);
      \draw[very thin] (0,\y) node[left] {\tiny{$0$}} -- (0,\y + 0.5) node[left] {\tiny{$1$}};
    }


    \node at (-2.5,2.25) {$[0,1,0,1,0,1]$};
    \draw[<->] (-1.2,2.25) --  (-0.5,2.25);
    \node at (-2.5,1.25) {$[0,1,0,0,1,1]$};
    \draw[<->] (-1.2,1.25) --  (-0.5,1.25);
    \node at (-2.5,0.25) {$[0,0,1,1,0,1]$};
    \draw[<->] (-1.2,0.25) --  (-0.5,0.25);

    

    \draw[color=blue,very thick] (-0.1,2) -- (0,2) to[out=0,in=180] (0.5,2.5) to[out=0,in=180] (1,2)
    to[out=0,in=180] (1.5,2.5) to[out=0,in=180] (2,2) to[out=0,in=180] (2.5,2.5) to[out=0,in=180] (3,2)
    --  (3.1,2); 

    \draw[color=orange,very thick] (-0.1,1) -- (0,1) to[out=0,in=180] (0.5,1.5) to[out=0,in=180] (1,1)
    to[out=0,in=180] (1.5,1) to[out=0,in=180] (2,1.5) to[out=0,in=180] (2.5,1.5) to[out=0,in=180] (3,1)
    --  (3.1,1); 

    \draw[color=red,very thick] (-0.1,0) -- (0,0) to[out=0,in=180] (0.5,0) to[out=0,in=180] (1,0.5)
    to[out=0,in=180] (1.5,0.5) to[out=0,in=180] (2,0) to[out=0,in=180] (2.5,0.5) to[out=0,in=180] (3,0)
    --  (3.1,0); 
    
    
    \draw (3.5,3) to[out=0,in=130] (4.2,1.25) ;
    \draw (3.5,-0.5) to[out=0,in=-130] (4.2,1.25) ;
    \node at (4.5,1.25) {$=$};
    

    \node at (5.67,1.25) {$\mathfrak{S}_\infty.[1,1,1]$};
    

    \draw[<->] (6.7,1.25) -- node [above,midway] {$|\cdot|$}  (8,1.25);

    \draw[fill] (8.4,1.25) circle (1.25pt);
    \draw[fill] (8.8,1.25) circle (1.25pt);  
    \draw[fill] (9.2,1.25) circle (1.25pt) node[right] {\quad``$\in \Z_{\geq 0}$''};
    
  \end{tikzpicture}
  \caption{Paths in the same orbit $\mathcal{O}_3=\mathfrak{S}_\infty.[1,1,1]$ for $\lambda=3$.}
  \label{fig:paths2}
\end{figure}

\begin{rem}
  In the general quantum interaction system case, the computation of the heat kernel using this method for a given
  Hamiltonian may produce the situation where the groups $\Z_2^n$ are replaced with a family of finite groups
  $\{G_n\}_{n\geq 0}$ that constitute a directed set (see also Remark \ref{non-commutativity}).
  In this case, it is necessary to find an appropriate invariant for the orbit decomposition with respect to certain
  group acting on the inductive limit $G_{\infty}$ of the family $\{G_n\}_{n\geq 0}$. We leave the detailed discussion for another occasion.
\end{rem}

\begin{rem}\label{rem:q-gravity}
In quantum gravity theory, we may find an important example where the path integral can be turned into a discrete summation defined over cosets of the modular group $SL_2(\Z)$. Loosely speaking, according to the theory of quantum gravity, in general, the space cannot be divided into infinity, so there may be no uncountable infinite number of paths to sum up in a path integral. In other words, as in our study, the path integral could turn to be discrete or particle-like, i.e. points. Actually, in \cite{ITT2015} (see also \cite{HITT2016}), using the Chern-Simons formulation, the partition function of the three-dimensional pure gravity given by a gravity path integral is exactly calculated by localization techniques developed in recent years. In addition, it is also worth noting that the resultant partition function is modular invariant.
\end{rem}

\begin{rem}
  We may describe the situation in \eqref{eq:orbitdecomposition} by the language of representation theory of $\mathfrak{S}_\infty$,
  that is, the space of the Fourier image also on $\Z_2^\infty$  by the representation induced from the trivial representation of
  its Young subgroups (see \cite{RW2020z}).   
\end{rem}

\subsection{Parity decomposition of the heat kernel} \label{sec:parity-decomp-heat}

As we already mentioned in the Introduction, the Hamiltonian \(\HRabi \) possesses a
\(\Z_2(=\Z/2\Z)\)-symmetry indicated by the existence of a parity operator \(\Pi= - \sigma_{z} e^{-i \pi a^{\dag} a} \) satisfying
\([\Pi,\HRabi]=0\) and with eigenvalues \(p=\pm 1\). Consequently, the direct decomposition of the full
space $L^2(\R)\otimes \C^2$ into the invariant subspaces (corresponding to the positive and negative parity)  is 
\[
  L^2(\R)\otimes \C^2= \mathcal{H}_+\oplus \mathcal{H}_-.
\]
First, we introduce the decomposition of the Hamiltonian of the QRM. We follow the discussion in \cite{B2011PRL-OnlineSupplement} and suggest the reader to consult \cite{B2013MfI, KRW2017, Reyes2018PhD} for more details.

Let $(\hat{T}\psi)(z):= \psi(-z)$ $(\psi \in L^2(\R)$ be the reflection operator acting on \( L^2(\R)\),  $\bU$ be the unitary operator on $L^2(\R)\otimes \C^2$ given by
\[
  \bU:= \frac1{\sqrt2}\begin{bmatrix}1&1\\
    \hat{T}&-\hat{T}
  \end{bmatrix},
\]
and \(\bC\) the Cayley transform
\[
  \bC := \frac1{\sqrt2}\begin{bmatrix}1&1\\ 1&-1 \end{bmatrix}.
\]
The parity decomposition of the \(\HRabi\) is given by
\begin{align*} 
(\bC \bU)^\dag \HRabi \bC \bU = \begin{bmatrix}H_+&0\\
    0&H_-
  \end{bmatrix},
\end{align*}
where the operators \(H_{\pm}\) are given by
\[
  H_{\pm} = a^{\dag} a + g ( a + a^{\dag}) \pm \Delta \hat{T}.
\]

Clearly, the subspaces
\begin{align}
  \label{eq:parityS}
  \mathcal{\bar{H}}_+ =  L^2(\R) \otimes \spn\left\{
  \begin{pmatrix} 
    1 \\
    0
  \end{pmatrix}\right\}
  \quad \text{ and } \quad
  \mathcal{\bar{H}}_- = L^2(\R) \otimes \spn\left\{
  \begin{pmatrix}
    0 \\
    1
  \end{pmatrix}\right\}
\end{align}
are invariant subspaces of the operator \( (\bC \bU)^\dag \HRabi \bC \bU\). Accordingly, we write
\( H_{\pm} = (\bC \bU)^\dag \HRabi \bC \bU |_{\mathcal{\bar{H}}_{\pm}} \).

Now, we proceed to compute the heat kernel of the parity Hamiltonians \(H_{\pm}\).

Recall that $\sigma_x=\begin{bmatrix} 0 & 1 \\ 1 & 0 \end{bmatrix}$ and
$\sigma_z=\begin{bmatrix} 1 & 0 \\ 0 & -1 \end{bmatrix}$. Notice that 
\[
  e^{-t \sigma_x} =
  \begin{bmatrix} \cosh & -\sinh
    \\ -\sinh & \cosh
  \end{bmatrix}(t)  \qquad \text{and} \qquad
  -\sigma_z e^{-t \sigma_x} = 
  \begin{bmatrix} -\cosh & \sinh
    \\ -\sinh & \cosh
  \end{bmatrix}(t). 
\]

For \(\epsilon,\delta \in \{+,-\}\), let us define four operators 
\( K_{\epsilon \delta} =  K_{\epsilon \delta}(x,y,t, \Delta) : L^2(\R) \to L^2(\R) \) by 
\begin{align*}
 (\bC \bU)^\dag \KRabi(x,y,t) \bC \bU 
  = \begin{bmatrix} K_{++} & K_{-+}\\ K_{+-}& K_{--}\end{bmatrix}.
\end{align*}
It is not difficult to see that
\begin{align*}
\frac{\partial}{\partial t} \begin{bmatrix} K_{++} & K_{-+}\\ K_{+-}& K_{--}\end{bmatrix}
= -\begin{bmatrix}H_+ & 0 \\ 0 & H_-\end{bmatrix}
\begin{bmatrix} K_{++} & K_{-+}\\ K_{+-}& K_{--}\end{bmatrix},
\end{align*}
and thus  \((\bC \bU)^\dag   \KRabi(x,y,t) \bC \bU \) is the heat kernel of the operator \((\bC \bU)^\dag \HRabi \bC \bU\).
Similarly, from this we see that $K_{++}$ (resp. $K_{--}$) is the heat kernel of $H_+$ (resp. $H_-$). One knows from the
general discussion for the $G$-function and constraint polynomials  (see e.g. \cite{B2011PRL} and \cite{KRW2017})
that $K_{--}(x,y,t, -\Delta)= K_{++}(x,y,t, \Delta)$. We will see this again below. 

Recall that the action of the (semigroup) operator \(e^{-t \HRabi}\) is given by
\[
  e^{-t \HRabi}\phi(x)= \int_{-\infty}^\infty  \KRabi(x,y,t;g,\Delta) \phi(y)dy
\]
for any compactly supported smooth function $\phi\in C_0^\infty(\R)\otimes \C^2$. From this expression, we have 
\begin{align*}
(\bC \bU)^\dag e^{-t \HRabi} \bC \bU ((\bC \bU)^\dag\phi)(x)
= & \begin{bmatrix} e^{-tH_+}&0\\
    0&e^{-tH_-}
  \end{bmatrix}((\bC \bU)^\dag\phi)(x)\\
= &\int_{-\infty}^\infty \begin{bmatrix} K_{++} & K_{-+}\\ K_{+-}& K_{--}\end{bmatrix} (x,y,t)((\bC \bU)^\dag\phi)(y)dy,
\end{align*}

From this expression, we observe the heat kernel is splitting along the two parities and  in Theorem \ref{Split_Kernel}
we give the explicit expression of the heat kernel by taking \(\phi \in H_{\pm} \) in the expression above.

For $\lambda\geq1$, define
\begin{align*}
  \Phi^\pm_{\lambda}(x,y,t;g) := e^{-2g^2 (\coth(\tfrac{t}2))^{(-1)^\lambda}} \idotsint\limits_{0\leq \mu_1 \leq \cdots \leq \mu_{\lambda} \leq 1}  e^{4g^2 \frac{\cosh(t(1-\mu_\lambda))}{\sinh(t)}(\frac{1+(-1)^\lambda}{2}) +  \xi_{\lambda}(\bm{\mu_{\lambda}},t;g)\pm \theta_{\lambda}(x,y, \bm{\mu_{\lambda}},t;g)} d \bm{\mu_{n}}
\end{align*}
and  
\begin{align*}
  \Phi^\pm_0(x,y,t;g) := e^{-2g^2\tanh\big(\frac{t}2\big) \pm\sqrt2 g(x+y)\tanh\big(\frac{t}2\big)} = e^{-2g^2\tanh\big(\frac{t}2\big) \pm \theta_{0}(x,y, \bm{\mu_{0}},t;g)}.
\end{align*}
Since $\theta_\lambda(x,y,\bm{\mu_{\lambda}},t;g)$, for \(\lambda \geq 0 \), is linear on $x$ and $y$, it is clear that 
\begin{align*}
  \Phi^\mp_\lambda(-x,-y,t;g)= \Phi^\pm_\lambda(x,y,t;g). 
\end{align*}

We now observe that 
\begin{align*}
   (\bC \bU)^\dag K_0 &(x,y,t;g)\begin{bmatrix} \cosh & -\sinh
    \\ -\sinh & \cosh
                     \end{bmatrix}
                     \Big(\theta_{2\lambda}(x,y, \bm{\mu_{2\lambda}},t;g)\Big) \bC \bU\\
  =& \bU^\dag K_0 (x,y,t;g) \bC e^{- \theta_{2\lambda}(x,y, \bm{\mu_{2\lambda}},t;g) \sigma_x} \bC \bU\\
  =& \bU^\dag K_0 (x,y,t;g)
     \begin{bmatrix} e^{- \theta_{2\lambda}(x,y, \bm{\mu_{2\lambda}},t;g)} & 0\\ 0& e^{\theta_{2\lambda}(x,y, \bm{\mu_{2\lambda}},t;g)} \end{bmatrix} \bU\\
  = & \frac12  \begin{bmatrix} 1 & \hat{T}\\ 1& -\hat{T}\end{bmatrix}
                                                K_0(x,y,t; g)
    \begin{bmatrix} e^{-\theta_{2\lambda}(x,y, \bm{\mu_{2\lambda}},t;g)} & 0\\ 0& e^{\theta_{2\lambda}(x,y, \bm{\mu_{2\lambda}},t;g)} \end{bmatrix}
     \begin{bmatrix} 1 & 1 \\ \hat{T} & -\hat{T} \end{bmatrix} \\
  =& \frac12  \begin{bmatrix} 
   K_0 e^{-\theta_{2\lambda}}+ \hat{T} K_0e^{\theta_{2\lambda}}\hat{T}
    & K_0 e^{-\theta_{2\lambda}} - \hat{T}K_0 e^{\theta_{2\lambda}}\hat{T}\\
    K_0 e^{-\theta_{2\lambda}}- \hat{T} K_0e^{\theta_{2\lambda}}\hat{T}
    & K_0 e^{-\theta_{2\lambda}}+ \hat{T} K_0 e^{\theta_{2\lambda}}\hat{T}
  \end{bmatrix}(x,y,t;g).
\end{align*}
Similarly 
\begin{align*}
 (\bC \bU)^\dag K_0 & (x,y,t;g) \begin{bmatrix} -\cosh & \sinh
                       \\ -\sinh & \cosh
                     \end{bmatrix}
                     \Big( \theta_{2\lambda+1}(x,y, \bm{\mu_{2\lambda+1}},t;g)\Big) \bC \bU\\
=& \bU^\dag K_0 (x,y,t;g) \bC (-\sigma_z)e^{- \theta_{2\lambda+1}(x,y, \bm{\mu_{2\lambda+1}},t;g)\sigma_x} \bC \bU\\
=& - \bU^\dag K_0 (x,y,t;g)
\begin{bmatrix} 
0& e^{\theta_{2\lambda+1}(x,y, \bm{\mu_{2\lambda+1}},t;g)} \\ e^{-\theta_{2\lambda+1}(x,y, \bm{\mu_{2\lambda+1}},t;g)} & 0 \end{bmatrix} \bU\\
= & -\frac12  \begin{bmatrix} 1 & \hat{T}\\ 1& -\hat{T}\end{bmatrix}
K_0 (x,y,t;g)
\begin{bmatrix} 0& e^{\theta_{2\lambda+1}(x,y, \bm{\mu_{2\lambda+1}},t;g)} \\ e^{- \theta_{2\lambda+1}(x,y, \bm{\mu_{2\lambda+1}},t;g)} &0 \end{bmatrix}
\begin{bmatrix} 1 & 1 \\ \hat{T} & -\hat{T} \end{bmatrix} \\
=& -\frac12  \begin{bmatrix} 
K_0 e^{\theta_{2\lambda+1}}\hat{T}+ \hat{T} K_0 e^{-\theta_{2\lambda+1}}
& -K_0 e^{\theta_{2\lambda}+1}\hat{T}+ \hat{T} K_0 e^{-\theta_{2\lambda+1}}\\
K_0 e^{\theta_{2\lambda+1}}\hat{T}- \hat{T} K_0 e^{-\theta_{2\lambda+1}}
& -K_0 e^{\theta_{2\lambda+1}}\hat{T}- \hat{T}K _0 e^{-\theta_{2\lambda+1}}
\end{bmatrix}(x,y,t;g).
\end{align*}

From the discussion, we obtain the heat kernel for the parity Hamiltonians \(H_{\pm}\).


\begin{thm} \label{Split_Kernel}
  The heat kernel $K_{\pm}(x,y,t, \Delta) (= K_{\pm \pm}(x,y,t, \Delta))$ of $H_\pm= \HRabi |_{\mathcal{H}_\pm}$ is given by 
  \begin{align*}
    K_{\pm}(x,y,t;g,\Delta)
    = & K_0 (x,y,t;g)\sum_{\lambda=0}^{\infty} (t\Delta)^{2\lambda} \Phi^-_{2\lambda}(x,y,t;g) \mp K_0 (x,-y,t;g) \sum_{\lambda=0}^{\infty}
        (t\Delta)^{2\lambda+1} \Phi^+_{2\lambda+1}(x,-y,t;g)
  \end{align*}
  Moreover, $K_{\pm \mp}(x,y,t, \Delta)=0$. In other words,
  \[
    \KRabi(x,y,t;g,\Delta) = K_{+}(x,y,t;g,\Delta)\oplus  K_{-}(x,y,t;g,\Delta).
  \]
\end{thm}

\begin{proof}
  For \(\epsilon, \delta \in \{+,- \}\), we define operators $k_{\epsilon \delta} = k_{\epsilon \delta}(x,y,t)\in \rm{End}_{\C}(\bar{\mathcal{H}}_\epsilon, \bar{\mathcal{H}}_\delta)$ by 
\[
  (k_{\epsilon \delta} v_{\epsilon})(x)= \int_{-\infty}^\infty  K_{\epsilon \delta}( x,y,t)v_{\epsilon}(y)dy
\]
for $v_{\epsilon}\in \bar{\mathcal{H}}_\epsilon$. Further, we write $k_{\epsilon \delta}$ as 
\[
  (k_{\epsilon \delta}v_{\epsilon})(x)= \sum_{\lambda=0}^{\infty} (\Delta t)^{\lambda} k_{\epsilon \delta}^{\lambda}v_{\epsilon}(x).
\]

By \eqref{eq:parityS}, we see that \(\bar{\mathcal{H}}_{\epsilon} \simeq L^2(\R) \). First,
we verify that \(K_{\pm \mp}(x,y,t;g,\Delta)=0 \). Let \(v \in L^2(\R)\) be a function with appropriate decay at $\pm \infty$
( e.g. \(v\) is a compactly supported function), then
\begin{align*}
  ({k_{+-}^{2\lambda} v})(x)  = & ({k_{- +}^{2\lambda} v})(x)  
  = \frac12\int_{-\infty}^\infty [K_0(x,y,t;g) \Phi^-_{2\lambda}(x,y,t;g) - \hat{T} K_0(x,y,t; g)\Phi^+_{2\lambda}(x,y;g)\hat{T}]v(y)dy \\
  = &\frac12 \int_{-\infty}^\infty K_0(x,y,t;g) \Phi^-_{2\lambda}(x,y,t;g)v(y)dy -\frac12 \int_{-\infty}^\infty K_0(-x,y,t;g) \Phi^+_{2\lambda}(-x,y,t;g)v(-y)dy\\
  = &\frac12 \int_{-\infty}^\infty K_0(x,y,t;g) \Phi^-_{2\lambda}(x,y,t;g)v(y)dy -\frac12 \int_{-\infty}^\infty K_0(-x,-y,t;g) \Phi^+_{2\lambda}(-x,-y,t;g)v(y)dy\\
  = &\frac12 \int_{-\infty}^\infty K_0(x,y,t;g) \Phi^-_{2\lambda}(x,y,t;g)v(y)dy   -\frac12 \int_{-\infty}^\infty K_0(x,y,t;g) \Phi^-_{2\lambda}(x,y,t;g)v(y)dy=0,
\end{align*}
and
\begin{align*}
  ({k_{+-}^{2\lambda+1} v}) (x)  = & -({k_{- +}^{2\lambda+1} v})(x) 
  = \frac12\int_{-\infty}^\infty [K_0(x,y,t;g) \Phi^+_{2\lambda+1}(x,y,t;g)\hat{T}- \hat{T} K_0(x,y,t;g)\Phi^-_{2\lambda+1}(x,y,t;g)]v(y)dy \\
  = &\frac12 \int_{-\infty}^\infty K_0(x,y,t;g) \Phi^+_{2\lambda+1}(x,y,t;g)v(-y)dy -\frac12 \int_{-\infty}^\infty K_0(-x,y,t;g) \Phi^-_{2\lambda+1}(-x,y,t;g)v(y)dy\\
  = &\frac12 \int_{-\infty}^\infty K_0(x,y,t;g) \Phi^+_{2\lambda+1}(x,y,t;g)v(-y)dy -\frac12 \int_{-\infty}^\infty K_0(-x,-y,t;g) \Phi^-_{2\lambda+1}(-x,-y,t;g)v(-y)dy\\
  = &\frac12 \int_{-\infty}^\infty K_0(x,y,t;g) \Phi^+_{2\lambda+1}(x,y,t;g)v(-y)dy -\frac12 \int_{-\infty}^\infty K_0(x,y,t;g) \Phi^+_{2\lambda+1}(x,y,t;g)v(-y)dy = 0.
\end{align*}

Thus, we see that $({k_{\pm \mp}^{\lambda} v})(x)=0$ for \( v \) with appropriate decay and \(\lambda \geq 0 \).

On the other hand, we have
\begin{align*}
({k_{++}^{2\lambda} v})(x)  &= ({k_{- -}^{2\lambda} v})(x) 
  = \frac12\int_{-\infty}^\infty [K_0(x,y,t;g) \Phi^-_{2\lambda}(x,y,t;g) + \hat{T} K_0(x,y,t;g)\Phi^+_{2\lambda}(x,y,t;g)\hat{T}]v(y)dy \\
  = &\frac12 \int_{-\infty}^\infty K_0(x,y,t;g) \Phi^-_{2\lambda}(x,y,t;g)v(y)dy + \frac12 \int_{-\infty}^\infty K_0(-x,y,t;g) \Phi^+_{2\lambda}(-x,y,t;g)v(-y)dy\\
  = &\frac12 \int_{-\infty}^\infty K_0(x,y,t;g) \Phi^-_{2\lambda}(x,y,t;g)v(y)dy  + \frac12 \int_{-\infty}^\infty K_0(-x,-y,t;g) \Phi^+_{2\lambda}(-x,-y,t;g)v(y)dy\\
  = &\frac12 \int_{-\infty}^\infty K_0(x,y,t;g) \Phi^-_{2\lambda}(x,y,t;g)v(y)dy  + \frac12 \int_{-\infty}^\infty K_0(x,y,t;g) \Phi^-_{2\lambda}(x,y,t;g)v(y)dy\\
  = & \int_{-\infty}^\infty K_0(x,y) \Phi^-_{2\lambda}(x,y,t;g)v(y)dy
\end{align*}
and 
\begin{align*}
  -({k_{+ +}^{2\lambda+1} v}) (x)  = & ({k_{- -}^{2\lambda+1} v})(x) 
  = \frac12\int_{-\infty}^\infty [K_0(x,y,t;g) \Phi^+_{2\lambda+1}(x,y,t;g)\hat{T} + \hat{T} K_0(x,y,t;g)\Phi^-_{2\lambda+1}(x,y,t;g)]v(y)dy \\
  = &\frac12 \int_{-\infty}^\infty K_0(x,y,t;g) \Phi^+_{2\lambda+1}(x,y,t;g)v(-y)dy  + \frac12 \int_{-\infty}^\infty K_0(-x,y,t;g) \Phi^-_{2\lambda+1}(-x,y,t;g)v(y)dy\\
  = &\frac12 \int_{-\infty}^\infty K_0(x,y,t;g) \Phi^+_{2\lambda+1}(x,y,t;g)v(-y)dy  + \frac12 \int_{-\infty}^\infty K_0(-x,-y,t;g) \Phi^-_{2\lambda+1}(-x,-y,t;g)v(-y)dy\\
  = &\frac12 \int_{-\infty}^\infty K_0(x,y,t;g) \Phi^+_{2\lambda+1}(x,y,t;g)v(-y)dy  + \frac12 \int_{-\infty}^\infty K_0(x,y,t;g) \Phi^+_{2\lambda+1}(x,y,t;g)v(-y)dy \\
  = & \int_{-\infty}^\infty K_0(x,y,t;g) \Phi^+_{2\lambda+1}(x,y,t;g)v(-y)dy = \int_{-\infty}^\infty K_0(x,-y,t;g) \Phi^+_{2\lambda+1}(x,-y,t;g)v(y)dy
\end{align*}
Hence, we have
\begin{align*}
({k_{\pm\pm} v})(x) 
  = \sum_{\lambda=0}^{\infty} \Delta^{2\lambda}\int_{-\infty}^\infty K_0(x,y,t;g) \Phi^-_{2\lambda}(x,y,t;g)v(y)dy
  \mp  \sum_{\lambda=0}^{\infty} \Delta^{2\lambda+1}\int_{-\infty}^\infty K_0(x,-y,t;g) \Phi^+_{2\lambda+1}(x,-y,t;g)v(y)dy.
\end{align*}
Thus we have the desired conclusion for \(K_{\epsilon\delta}\) as a distribution, whence the result follows as functions in the standard way.

\end{proof}

\begin{rem}
  Note that for  \( v  \in L^2(\R)\), we can write
  \[
    \int_{-\infty}^\infty K_0(x,y,t;g) \Phi^+_{2\lambda+1}(x,y,t;g)v(-y)dy = \int_{-\infty}^\infty K_0(x,y,t;g) \Phi^+_{2\lambda+1}(x,y,t;g) (\hat{T}v)(y)dy.
  \]
\end{rem}


To conclude this section, we compute the partition function \(\ZRabi^{\pm}(\beta;g,\Delta) \) of the parity
Hamiltonian \(H_{\pm} \).

\begin{cor}
  \label{cor:parityPart}
  The partition function \(\ZRabi^{\pm}(\beta;g,\Delta) \) for the parity Hamiltonian \(H_{\pm}\) is given by
  \begin{align*}
    \ZRabi^{\pm}(\beta;g,\Delta) =& \frac{ e^{g^2\beta}}{1-e^{- \beta}} \Bigg[ 1 + e^{-2g^2 \coth(\frac{\beta}2)}  \sum_{\lambda =1}^{\infty} (\beta\Delta)^{2 \lambda} \idotsint\limits_{0\leq \mu_1 \leq \cdots \leq \mu_{2 \lambda} \leq 1} e^{4g^2\frac{\cosh(\beta(1-\mu_{2\lambda}))}{\sinh(t)} +   \xi_{2\lambda}(\bm{\mu_{2 \lambda}},\beta;g) +\psi^-_{2 \lambda} (\bm{\mu_{2\lambda}},\beta;g)} d \bm{\mu_{2\lambda}}  \Bigg] \\
                              & \quad \mp \frac{ e^{g^2 \beta}}{1+e^{- \beta}} e^{- 2g^2 \tanh(\frac{\beta}2)} \sum_{\lambda = 0}^{\infty} (\beta \Delta)^{2\lambda+1} \idotsint\limits_{0\leq \mu_1 \leq \cdots \leq \mu_{2 \lambda+1} \leq 1} e^{\xi_{2\lambda+1}(\bm{\mu_{2 \lambda+1}},\beta;g) +\psi^+_{2 \lambda+1} (\bm{\mu_{2\lambda+1}},\beta;g)} d \bm{\mu_{2\lambda+1}},
  \end{align*}
  where the function $\psi_\lambda^{\pm}(\bm{\mu_{\lambda}},t;g)$ is as in Corollary \ref{cor:Partition_function}.
\end{cor}

\begin{proof}
  The first part is computed in the same way as in the case of \(Z_{\text{Rabi}}(\beta) \) (cf. Corollary \ref{cor:Partition_function} ) by noticing that
  \[
    \int_{-\infty}^\infty e^{-\alpha x^2} \cosh(x \, \eta)dx = \int_{-\infty}^\infty e^{-\alpha x^2 \pm \eta x} dx = \sqrt{\frac{\pi}{\alpha}}e^{\frac{\eta^2}{4\alpha}}.
  \]
  For the second part, it is enough to observe that
  \[
    K_0(x,-x,t;g) = \frac{ e^{g^2t}}{\sqrt{\pi (1-e^{-2t})}} \exp\left( - \frac{1+e^{-t}}{1-e^{-t}}x^2\right),
  \]
   and, for \( \lambda \equiv 1 \pmod{2} \),
  \[
     \theta_{\lambda}(x,-x, \bm{\mu_{\lambda}},t;g) = -\frac{2\sqrt{2} g  x}{1 - e^{- t}} \sum_{\gamma=0}^{\lambda} (-1)^{\gamma}  \left( e^{- t \mu_{\gamma} } + e^{ t( \mu_{\gamma}- 1)}  \right),
  \]
  and then proceed as in the case of \(\ZRabi(\beta;g,\Delta)\).
\end{proof}


\section{Analytic properties of the heat kernel}
\label{sec:analyt-prop-heat}

In this short section, we discuss some estimates for the absolute value of the components of the heat kernel. These estimates allow
a direct proof of the pointwise and uniform convergence of the series appearing in the heat kernel and partition function, along
with other analytic properties. Despite the apparent complication of the formulas developed in this paper, it is not difficult
to obtain these estimates. We leave the detailed discussion, including the analytic continuation, to \cite{RW2020z}. 

Let \(\lambda \in \Z_{\geq 1}\), for fixed $x,y>0$ and $t >0$, there are real functions \( C_1(x,y,t), C_2(t), C_3(t) \geq 0 \) bounded in closed
intervals of the half plane $\Re(t)>0$, such that
\begin{align} \label{eq:estimateK}
  \left| \theta_{\lambda}(\bm{\mu_\lambda},x,y,t) \right| &\leq \left|\frac{\sqrt{2} g }{1-e^{-2 t}} \right| C_1(x,y,t) \nonumber \\
  \left| \psi_{\lambda}^{\pm}(\bm{\mu_\lambda},t) \right| &\leq \left|\frac{2 g^2 }{1-e^{-2 t}} \right| C_2(t) \nonumber \\
  \left\vert \xi_{\lambda}(\bm{\mu_{\lambda}},t) \right\vert  &\le \left|\frac{2 g^2 }{1-e^{-2 t}} \right| C_3(t) \lambda  
\end{align}
uniformly for \(0 \leq  \mu_1 \leq \mu_2 \leq \cdots \leq \mu_\lambda \leq 1\).  We refer to Lemma 3.1 of \cite{RW2020z} for the proof\footnote{We note that 
  in proof of Lemma 3.1 in \cite{RW2020z} the expressions of $s(x,y,t)$ and $S_n(t)$ are not correct. 
  The correct expression for $s(x,y,t)$ is
  \[
    s(x,y,t) = t\sum_{\gamma=0}^{\frac{\lambda-1}{2}} \bigg(x \int_{\mu_{2\gamma}}^{\mu_{2\gamma+1}} (e^{t(1-\alpha)}-e^{t(\alpha -1)}) d \alpha + y \int_{\mu_{2\gamma}}^{\mu_{2\gamma+1}}(e^{t \alpha}-e^{- t \alpha}) d \alpha \bigg)
  \]
  and the correct expression for $S_n(t)$ is
  \[
    S_n(t) = - t^2 \left( \int_{ \mu_{n}}^{ \mu_{n+1}} (e^{t \alpha }- e^{-t \alpha }  ) d \alpha  \right) \sum_{\substack{n<\beta\leq \lambda-1 \\ \beta - \alpha \equiv 1 \pmod{2}} }
    \left(  e^{-2 t}\int_{ \mu_{\beta}}^{ \mu_{\beta+1} } e^{t \alpha } d \alpha  - \int_{ \mu_{\beta}}^{ \mu_{\beta+1} } e^{- t \alpha} d \alpha    \right).
  \]
  These errors do not affect the proof, which follows as written in \cite{RW2020z} with no further changes.
  }.

As mentioned in Section \ref{sec:spvl}, the general theory of the Trotter-Kato product formula assures the pointwise uniform
convergence of the heat kernel. However, with the estimates above, we can prove the uniform convergence of the heat kernel directly.
Note that a similar result holds for the partition function.

\begin{thm}
  For any closed interval $I \subset (0,\infty)$, as a function of $t$ the series given in Theorem \ref{thm:heat_kernel} are uniformly convergent
  component-wise for fixed $x,y,g,\Delta >0$.
\end{thm}

\begin{proof}

  By the estimates \eqref{eq:estimateK}, it is immediate to verify that there are constants $c_0,c_1,c_2,c_3(x,y),c_4>0$ such that any
  matrix component of $\KRabi(x,y,t; g,\Delta)$ is bounded above by 
  \[
    M = c_0 e^{c_2+ c_3(x,y)+ c_1\Delta e^{c_4}},
  \]
  and the result follows in the standard way by the Weierstrass $M$-test.
\end{proof}

In addition to verifying the convergence of the series, the estimates above are enough to show that the heat kernel and partition
functions of the QRM are holomorphic functions (with respect to the variable $t$) on certain regions of the complex plane. In particular, the
analytic continuation of the heat kernel gives the time evolution propagator of the QRM while the analytic continuation of the
partition function gives the meromorphic continuation of the spectral zeta function of the QRM. We refer the reader to \cite{RW2020z}
for details.

Next, we consider estimates with respect to the spatial variables. It is well-known, and elementary to verify, that for fixed $t>0$ there
are positive constants $a,b>0$ such that
\begin{equation}
  \label{eq:estimateMehler}
  |K_{0}(x,y,t; g )| \leq a e^{-b (x^2 + y^2)}.
\end{equation}
It is not difficult to extend this result to the case of the heat kernel of the QRM as follows.

\begin{prop}
  Let
  \[
    \KRabi(x,y,t; g,\Delta) =
    \begin{bmatrix}
      k_{1,1}(x,y,t; g, \Delta) & k_{1,2}(x,y,t; g, \Delta)\\
      k_{2,1}(x,y,t; g, \Delta) & k_{2,2}(x,y,t; g, \Delta)
    \end{bmatrix}.
  \]
  Then, for fixed $g,\Delta, t>0$, there are positive constants $a,b$ such that
  \[
    |k_{i,j}(x,y,t; g, \Delta)| \leq a e^{-b (x^2+y^2)},
  \]
  for $i,j \in \{1,2\}$.
\end{prop}

\begin{proof}
  Let us rewrite the function $\theta_{\lambda}(x,y,\bm{\mu_{\lambda}},t; g)$ as
  \[
    \theta_{\lambda}(x,y,\bm{\mu_{\lambda}},t; g) = \sqrt{2} g(x+y) \tanh(t) - \frac{\sqrt{2} g t}{\sinh(t)} s_\lambda(x,y,\bm{\mu}_{\lambda},t),
  \]
  with
  \[
    s_\lambda(x,y,\bm{\mu}_{\lambda},t) = \sum_{\gamma=1}^{\frac{\lambda}{2}} \bigg(x \int_{\mu_{2\gamma-1}}^{\mu_{2\gamma}} (e^{t(1-\alpha)}-e^{t(\alpha -1)}) d \alpha + y \int_{\mu_{2\gamma-1}}^{\mu_{2\gamma}}(e^{t \alpha}-e^{- t \alpha}) d \alpha \bigg)
  \]
  if $\lambda \equiv 0 \pmod{2}$ and 
  \[
    s_\lambda(x,y,\bm{\mu}_{\lambda},t) = \sum_{\gamma=0}^{\frac{\lambda-1}{2}} \bigg(x \int_{\mu_{2\gamma}}^{\mu_{2\gamma+1}} (e^{t(1-\alpha)}-e^{t(\alpha -1)}) d \alpha + y \int_{\mu_{2\gamma}}^{\mu_{2\gamma+1}}(e^{t \alpha}-e^{- t \alpha}) d \alpha \bigg)
  \]
  if $\lambda \equiv 1 \pmod{2}$. In any of the two cases we can verify that
  \[
    0 \leq s_{\lambda}(x,y,\bm{\mu}_{\lambda},t) \leq (x+y) (e^t + e^{-t}).
  \]
  Then, let $\Xi(x,y,t)$ be the series
  \begin{align*}
     &\sum_{\lambda=0}^{\infty} (t\Delta)^{\lambda} e^{-2g^2 (\coth(\tfrac{t}2))^{(-1)^\lambda}} \idotsint\limits_{0\leq \mu_1 \leq \cdots \leq \mu_\lambda \leq 1}  e^{4g^2 \frac{\cosh(t(1-\mu_\lambda))}{\sinh(t)}(\frac{1+(-1)^\lambda}{2}) + \xi_{\lambda}(\bm{\mu_{\lambda}},t)} (-1)^{\lambda}  \cosh \left( \theta_{\lambda}(x,y,\bm{\mu_{\lambda}},t) \right) d \bm{\mu_{\lambda}},
  \end{align*}
  then, by \eqref{eq:estimateK} and the foregoing discussion there are positive constants $a,b,c,d$ such that
  \begin{align*}
    |\Xi(x,y,t)| &\leq a \sum_{\lambda=0}^{\infty} (t\Delta b)^{\lambda} \idotsint\limits_{0\leq \mu_1 \leq \cdots \leq \mu_\lambda \leq 1} ( e^{\theta_{\lambda}(x,y,\bm{\mu_{\lambda}},t)} + e^{-\theta_{\lambda}(x,y,\bm{\mu_{\lambda}},t)} )\\
               &\leq a e^{c(x+y)} \sum_{\lambda=0}^{\infty} \frac{(t\Delta b)^{\lambda}}{\lambda!} = d e^{c(x+y)},
  \end{align*}
  and a similar one for the remaining components of the heat kernel. Combining with the estimate \eqref{eq:estimateMehler} of the
  Mehler kernel we obtain the desired result.
\end{proof}

An immediate corollary of the above result is that the heat kernel is continuous with respect to the spatial variables $x,y$. 

\begin{cor}
  For fixed $g,\Delta, t>0$, the function $\KRabi(x,y,t ; g, \Delta)$ is continuous in the variables $x,y$.
\end{cor}

A detailed analysis with the estimates above, or similar ones, may be used to prove further analytical properties of the heat
kernel or the partition functions of the QRM with respect to the space variables. We do not further pursue this direction in this
paper.

\section*{Acknowledgements}

The authors would like to express their gratitude to the QUTIS group in the university of the Basque Country (UPV/EHU) for  the hospitality during the visit on the occasion of the Workshop on Quantum Simulation and Computation in February 2018, and particularly to \'I\~nigo L.~Egusquiza for the extensive discussions during the stay in Bilbao.  Furthermore, the authors greatly appreciate the ocassional discussions over the years with Daniel Braak on the heat-semigroup for
quantum interaction systems from the physics viewpoint. These fruitful discussions have stimulated the authors very much for advancing the
study on the heat kernel. The authors would like to express also their deep thanks to Takashi Ichinose for giving them valuable comments
on a recent development of the Trotter-Kato product formula and for the longterm stimulating discussions.
This work was partially supported by JST CREST Grant Number JPMJCR14D6, Japan, and
by Grand-in-Aid for Scientific Research (C) JP16K05063 and JP20K03560 of JSPS, Japan.

C.R.B and M.W. contributed equally to this work.


\begin{flushleft}

\bigskip

 Cid Reyes-Bustos \par
 Department of Mathematical and Computing Science, School of Computing, \par
 Tokyo Institute of Technology \par
 2 Chome-12-1 Ookayama, Meguro-ku, Tokyo 152-8552 JAPAN \par\par
 \texttt{reyes@c.titech.ac.jp}

 \bigskip

 Masato Wakayama \par
 Institute of Mathematics for Industry,\par
 Kyushu University \par
 744 Motooka, Nishi-ku, Fukuoka 819-0395 JAPAN \par
 \texttt{wakayama@imi.kyushu-u.ac.jp}

 \medskip

 Current address: \\
 Department of Mathematics, \par
 Tokyo University of Science \par
 1-3 Kagurazaka, Shinjyuku-ku, Tokyo 162-8601 JAPAN \par\par
 \texttt{wakayama@rs.tus.ac.jp}

\end{flushleft}

\end{document}